\documentclass[a4paper,11pt]{scrartcl}
\usepackage{amsmath}%, mathabx}
\usepackage{amsfonts}
\usepackage{float}
\usepackage[all]{xy}
\usepackage{pinlabel}%,color}
\usepackage{todonotes}
\usepackage{wrapfig}
\usepackage{graphicx}
\usepackage{soul}

\usepackage[utf8]{inputenc}
\usepackage[T1]{fontenc}

\usepackage{appendix}

\usepackage{multicol}
\usepackage{needspace}
\usepackage{comment}

\usepackage[normalem]{ulem}

\usepackage{amsmath}
\usepackage{mathtools}
\usepackage{tensor}
\usepackage{leftidx}
\usepackage{amsthm}
\usepackage{enumitem}
\setlist[description]{%
  font={\rmfamily\mdseries \dashuline}, % set the label font
}
\usepackage{scalerel}	% for resizing boxtimes to bigboxtimes
\usepackage[varumlaut]{yfonts}
\usepackage{amssymb}
\usepackage{amsfonts}
\usepackage{bbm}
\usepackage{extarrows}
\usepackage{stmaryrd}

\usepackage[mathscr]{euscript}

\usepackage{pgf}
\usepackage{tikz}

\usepackage{enumitem}

\usepackage{fancybox}
\usetikzlibrary{arrows,shapes,trees,matrix,calc}
\usetikzlibrary{decorations.markings}
\usetikzlibrary{matrix}

\usepackage{hyperref}

\usepackage[shortcuts]{extdash}

\allowdisplaybreaks

\newcommand{\OMIT}[1]{}

\newcommand{\CC}{\mathbb{C}}

\newcommand{\ZZ}{\mathbb{Z}}

\newcommand{\lto}{\longrightarrow}
\newcommand{\lmapsto}{\longmapsto}

\newcommand{\eps}{\varepsilon}
\newcommand{\R}{\mathcal{R}}

\newcommand{\linhull}{\operatorname{span}}

\newcommand{\id}[1]{\operatorname{id}_{#1}}

\newcommand{\im}{\operatorname{im}}

\newcommand{\ot}{\otimes}

\newcommand{\lact}{\mathbin{\triangleright}}
\newcommand{\lrangle}[1]{{\langle #1 \rangle}}

\newcommand{\hol}{\operatorname{hol}} %holonomy

\newcommand{\cat}[1]{\mathcal{#1}}

\newcommand{\B}{\cat{B}}

\newcommand{\F}{\cat{F}}

\newcommand{\lmod}[1]{{#1}\!\operatorname{--mod}}
\newtheoremstyle{indented}{3pt}{3pt}{\addtolength{\leftskip}{5.5em}}{}{\bfseries}{.}{.5em}{}

\theoremstyle{plain}

\newcounter{dummy} %nur für scrbook: %\numberwithin{dummy}{chapter}

\newtheorem{theorem}[dummy]{Theorem}
\newtheorem{proposition}[dummy]{Proposition}
\newtheorem{lemma}[dummy]{Lemma}

\theoremstyle{definition}

\newtheorem{definition}[dummy]{Definition}

\theoremstyle{remark}

\newtheorem{remark}[dummy]{Remark}

\newtheorem{example}[dummy]{Example}

\theoremstyle{indented}

\newcommand{\beq}{\begin{equation}}
\newcommand{\eqq}{\end{equation}}

\newcommand{\Fun}[1]{\mathscr{F}_{\CC}(#1)}

\newcommand{\Hil}{\mathcal{H}}  %% Hilbert space hook
\newcommand{\Kit}{\mathrm{Kit}}

\newcommand{\hKit}{\mathrm{HKit}} 

\newcommand{\fk}{\mathrm{ff}}
\newcommand{\full}{\mathrm{full}}

\newcommand{\KM}{Kitaev model} 
\newcommand{\KQDM}{Kitaev quantum double model}
\newcommand{\HKM}{higher Kitaev model}
\newcommand{\sHKM}{full higher Kitaev model}
\newcommand{\HHK}{Hopf-algebraic higher Kitaev} 
\newcommand{\HAKM}{Hopf-algebraic Kitaev model}

\newcommand{\GoXo}{(1,G,X,1)}   %% maybe?
\newcommand{\GEXo}{(E,G,X,1)}   %% maybe?
\newcommand{\GEXY}{(E,G,X,Y)}   %% maybe?
\newcommand{\ooXY}{(1,1,X,Y)}   %% maybe?
\newcommand{\GEoo}{(E,G,1,1)}   %% maybe?
\newcommand{\GooY}{(1,G,1,Y)}

\newcommand{\GX}{(G,X)}   %% maybe?
\newcommand{\XY}{(Y,X)}   %% maybe?

\newcommand{\red}[1]{\textcolor{red}{#1}}
\newcommand{\blue}[1]{\textcolor{blue}{#1}}

\usepackage[top=2.0cm, bottom=2.0cm, left=2.0cm, right=2.0cm]{geometry}

\addtokomafont{sectioning}{\rmfamily}

\title{
Exactly solvable 
models for 2+1D topological phases derived from crossed modules of semisimple 
Hopf algebras}

\author{{Vincent Koppen\footnote{vincent.koppen@posteo.de}}\hspace{0.6cm} 
{Jo\~ao Faria Martins\footnote{j.fariamartins@leeds.ac.uk}}\hspace{0.6cm}
{Paul Purdon Martin}\footnote{{p.p.martin@leeds.ac.uk}}\\
{\small School of Mathematics, University of Leeds, UK}
}

\begin{document}
\maketitle

\begin{abstract}
We define an exactly solvable model for 2+1D topological phases of matter on a triangulated surface derived from a crossed module of semisimple finite-dimensional Hopf algebras, the \emph{\HHK\ model}.
This model generalizes both the Kitaev quantum double model for a semisimple Hopf algebra and  the \sHKM\  derived from a 2-group, and can hence be interpreted as a Hopf-algebraic {discrete} higher gauge theory.

We construct a family of crossed modules of semisimple Hopf algebras, $(\Fun{X} \otimes \CC E \xrightarrow{\partial} \Fun{Y} \rtimes \CC G, \lact)$, that depends on four finite groups, $E,G,X$ and $Y$. We calculate the ground-state spaces of the resulting model on a triangulated surface when  $G=E=\{1\}$ and when $Y=\{1\}$,
prove that 
those ground-state spaces
are canonically independent of the triangulations, and so depend only on the underlying surface; 
and moreover we find a 2+1D TQFT whose state spaces on surfaces give the ground-state spaces. These TQFTs are particular cases of Quinn's finite total homotopy TQFT  and hence the  state spaces assigned to surfaces are free vector spaces on sets of homotopy classes of maps from a surface to homotopy finite spaces, in this case obtained as classifying spaces of finite groupoids and finite crossed modules of groupoids.

We leave it as an open problem whether the ground-state space of the \HHK\ model on a triangulated surface is independent of the triangulation for general crossed modules of semisimple Hopf algebras,  
whether a TQFT always exists whose state space on a surface gives the ground-state space of the model, {and whether the ground-state space of the model obtained from $E,G,X,Y$ can always be given a homotopical explanation.}
\end{abstract}

{\noindent \textbf{Keywords:} {Crossed modules; Hopf Algebras;  Topological Quantum Field Theories; Hamiltonian models for topological phases; Kitaev quantum double model.}

\medskip

{\noindent \textbf{Acknowledgements:} This paper was financed by the Leverhulme trust {research project} grant RPG-2018-029: 
\textit{``Emergent Physics From Lattice Models of Higher Gauge Theory''}. JFM would like to express his gratitude to Tim Porter for discussions on Quinn-like TQFTs for the case of classifying spaces of groupoids and crossed modules / complexes of groupoids.}
VK would like to thank Ehud Meir for countless invaluable explanations about Hopf algebras, Christoph Schweigert for introducing him to the Kitaev model, and Catherine Meusburger and Thomas Vo\ss\ for further helpful discussions.
The very final stages of the writing of this paper were financed by the EPSRC
Programme Grant No. EP/W007509/1: \textit{``Combinatorial
Representation Theory: Discovering the Interfaces of
Algebra with Geometry and Topology''.} We all thank the referee for  useful comments and suggestions.

\noindent \textbf{Authorship:} {The three authors have contributed in
equal measure to the conceptualisation and
contextualisation of the paper. The initial idea, and a great part of the calculations for the construction of the {\HHK\ model}, are due to VK. The relations between the model and Quinn's finite total homotopy TQFT were due to JFM.} 

\noindent \textbf{Data Access Statement:} No data was created while producing this publication.

\noindent \textbf{Open access statement:} {For the purpose of open access, the authors have applied a Creative Commons Attribution (CC BY) licence to any Author Accepted Manuscript version arising from this submission.}

 \setcounter{tocdepth}{2}
\tableofcontents

\section{Introduction}

In his seminal paper Kitaev \cite{Kitaev} defined a lattice model, henceforth called the \emph{\KQDM}, {or simply the \emph{\KM},}
for 
(2+1)-dimensional, (2+1D), 
topological phases of matter on an 
oriented surface $\Sigma$, with a triangulation, $L$.
{Initially proposed in the context of quantum computing, as a model for an error-correcting 
{quantum} 
code, the so-called \emph{toric code}, it allows for fault-tolerant quantum gates by braiding (non-abelian) anyons \cite{Kitaev,NonAbelianAnyons}.
Due to its relation 
{\cite{balsam-kirillov}}
to the {Turaev-Viro / Barrett-Westbury constructions} 
\cite{Turaev_Viro,Barrett_Westbury} {of quantum invariants of 3-manifolds from spherical fusion categories}, 
the \KM\  
{also} 
provides a link between low-dimensional topology, Hopf algebras and tensor categories.}
In this paper we extend {the} \KM\ to handle crossed modules of Hopf algebras {more generally}, {hence extending it to Hopf-algebraic higher gauge theory,  as we now elaborate.}

The {\KM} on $L$, {a triangulation of a surface $\Sigma$,} has as input a finite group $G$, {and can be seen as a lattice gauge theory model \cite{meusburger2016hopf}.} 
The total Hilbert space of the \KM\ 
on $L$ is 
$\Hil^\Kit_L := \CC G^{L^1}$, 
where $L^1$ is the set of edges of $L$. This vector space is the free vector space on the set of discretised $G$-connections over $(\Sigma,L)$, called in \cite[\S 2.1]{companion} \emph{gauge $G$-configurations} over $(\Sigma,L)$. One can then define \emph{vertex operators}, 
$V_v^g\colon {\Hil^\Kit_L \to \Hil^\Kit_L}$, and plaquette operators, 
$F_{P,w}^{\delta_a}\colon {\Hil^\Kit_L \to \Hil^\Kit_L}$, 
where $g,a \in G$, $v \in L^0$ is a vertex of $L$, and $P \in L^2$ is a plaquette of $L$, with a %choice of 
vertex %(here called a \emph{“base-point”}) 
$w \in \partial P$. The operator $V_v^g 
\colon {\Hil^\Kit_L \to \Hil^\Kit_L}$ 
performs 
a discrete gauge transformation supported on $v$, whereas the operator $F_{P,w}^{\delta_a}\colon
{\Hil^\Kit_L \to \Hil^\Kit_L}$ ‘chooses’  the discrete gauge configurations whose holonomy around the plaquette $P$, with initial point $w$, is $a \in G$.  These operators extend 
{respectively}
to representations of the Hopf {algebras} $\CC G$, the group algebra of $G$, and $\Fun{G}=\mathrm{span}_{\mathbb{C}}\{\delta_a \mid a \in G\}$, the algebra of functions on $G$, (and if $V_v^g$ and $F_{P,v}^{\delta_a}$, for the same $v$, are put together, of the quantum double $D(G) = \Fun{G} \rtimes \CC G$). 
The \emph{local operator algebra} of 
the \KQDM\  %Kitaev's quantum double model 
is the algebra of operators ${\Hil^\Kit_L \to \Hil^\Kit_L}$ generated by all vertex and plaquette operators. 
The topological excitations of the \KM\  %Kitaev model 
are naturally those that are given by representations of the local operator algebra, as these are the excitations that cannot be destroyed by local operators \cite{Kitaev,Lan_Wen}.

The Hopf algebras $\CC G$ and $\Fun{G}$ are both semisimple (and so is the quantum double $D(G) = \Fun{G} \rtimes \CC G$). In particular \cite{larson_radford} they have unique Haar integrals, which take the form $\ell=\frac{1}{|G|}\sum_{g \in G} g \in \CC G$ and $\lambda=\delta_{1_G} \in \Fun{G}$. Define, given a vertex $v$, a projector $V_v:=V_v^\ell=\frac{1}{|G|} \sum_{g \in G} V_v^g$, called a \emph{vertex projector}, and given $P \in L^2$, the projector $F_P:=F_{P,w}^\lambda$, called a \emph{plaquette projector}. (These plaquette projectors turn out to be
independent of the vertex $w$ in the boundary of $P$.) 
The Hamiltonian of the \KM\   
is {$H_0\colon \Hil^\Kit_L \to \Hil^\Kit_L$} \cite[\S 5]{Kitaev}, where
\begin{equation}\label{eq:KM}
H_0=\sum_{v \in L^0}(1-V_v)+\sum_{P \in L^2}(1-F_P).
\end{equation}
Here $L^0$ and $L^2$ are the sets of vertices and plaquettes of $L$.
All projectors appearing in $H_0\colon {\Hil^\Kit_L \to \Hil^\Kit_L}$ are mutually commuting. It is for this reason that the 
\KM\
is called a{n exactly} solvable model {\cite{kitaev_exactly-solvable}}.
Given that the Hamiltonian is the sum of commuting projectors, it is diagonalisable, and it has a lowest-energy eigenspace, the \emph{ground-state space}.

Even though the \KM\  % Kitaev model 
defined on $(\Sigma,L)$ explicitly depends on the triangulation (or cell decomposition) $L$ of $\Sigma$, its ground-state space is naturally a topological invariant, canonically given by the free vector space on the set of homotopy classes of maps $\Sigma \to B_G$, where  $B_G$ is the classifying space of $G$. A 2+1D topological quantum field theory (TQFT) exists, sending a surface $\Sigma$ to the ground-state space of the 
\KQDM\  %Kitaev's quantum double model 
on $\Sigma$ derived from $G$. Namely consider the Dijkgraaf-Witten TQFT {\cite{DW,Morton}} with group $G$ and trivial cocycle. The latter TQFT is a special case of  {Quinn's finite  total homotopy TQFT} $\mathscr{Q}_\mathcal{B}$ {\cite[Lecture 4]{Quinn}\cite[\S 4]{martins_porter21},} which depends on the choice of a homotopy finite space $\mathcal{B}$, and sends a surface $\Sigma$ to the free vector space on the set of homotopy classes of maps $\Sigma \to \mathcal{B}$. (Dijkgraaf-Witten TQFT with
group $G$ and 
trivial cocycle {$\omega \in  H^{3}(G,U(1))$} is just $\mathscr{Q}_{B_G}$.)

\medskip

The \KQDM\ has been extended in a number of ways. {In \cite{buerschaper-et-al}, see also \cite{balsam-kirillov},} a generalised 2+1D model, henceforth called the \emph{\HAKM} was defined. The latter takes as input, instead of a finite group  $G$, a finite-dimensional semisimple complex Hopf algebra $H$, with Haar integral $\ell \in H$, {and can be formulated in the context of Hopf algebra gauge theory on a lattice \cite{meusburger2016hopf}.  The \HAKM} is an exactly solvable model defined on the vector space {$\Hil_L^{\hKit}:=H^{\otimes L^1}$}. {This is a Hilbert space, and the vertex and plaquette projectors are Hermitian, if one takes as input a finite-dimensional Hopf $C^*$-algebra $H$ \cite{buerschaper-et-al}.}
The local operator algebra of the \HAKM\ likewise contains vertex operators $V_{v,P}^h\colon {{\Hil}_L^{\hKit} \to {\Hil}_L^{\hKit}}$, where $v \in L^0$, $h \in H$, and $P\in L^2$ is an adjacent plaquette to $v$, which is used {\cite[Definition 2.2]{balsam-kirillov}} to define a total order on the set of edges adjacent to $v$ (a cyclic order is given by the orientation, and $P$ is used to specify an initial edge adjacent to $v$). {This additional data to define vertex operators is essential when $H$ is non-cocommutative}. We also have plaquette operators $F_{P,w}^\psi\colon {{\Hil}_L^{\hKit} \to {\Hil}_L^{\hKit}}$, where $\psi \in H^*$, $P\in L^2$ and $w$ a vertex in the boundary of $P$.
Note that {the dual vector space} $H^*$ is also a finite-dimensional semisimple Hopf algebra \cite{larson_radford} and hence it has a unique Haar integral $\Lambda \in H^*$. As in the Kitaev model, we can define mutually commuting vertex projectors $V_v:=V_{v,P}^\ell\colon {\Hil}_L^{\hKit} \to {\Hil}_L^{\hKit}$ (this is independent of the adjacent plaquette {\cite[\S 2.4]{balsam-kirillov},} {because the Haar integral $\ell$ is cocommutative \cite{larson_radford} and hence any coproduct of it cyclically invariant}) and plaquette projectors $F_P:=F_{P,w}^\Lambda\colon {{\Hil}_L^{\hKit} \to {\Hil}_L^{\hKit}}$. (Likewise $F_{P,w}$ is independent of the vertex in the boundary of $P$.)
The {Hamiltonian of the} \HAKM\ is also given by \eqref{eq:KM}, and hence reduces to 
the \KQDM\ %Kitaev's quantum double model 
{if} $H=\CC G$.

The ground-state space of the \HAKM\ defined on a triangulated surface $(\Sigma,L)$ is, as for the \KM, canonically independent of the triangulation $L$ of $\Sigma$. This follows from \cite[Theorem 4.1]{balsam-kirillov} since the Turaev-Viro TQFT {for the spherical fusion category $\lmod{H}$ of modules of a semisimple Hopf algebra $H$} is such that the state space assigned to a surface $\Sigma$ is canonically isomorphic to the ground-state space of the \HAKM\ for the Hopf algebra $H$, on any triangulation of $\Sigma$, see also \cite{Kadaretal}.
{However, the  \HAKM\ does not describe the state spaces of every Turaev-Viro TQFT: only  those coming from spherical fusion categories that arise as representation categories of finite dimensional semisimple Hopf algebras.
More generally, for any spherical fusion category the Levin-Wen string-net construction \cite{Levin_Wen, kirillov} provides a commuting-projector Hamiltonian model whose ground-state spaces recover the state spaces assigned to surfaces by the corresponding Turaev-Viro TQFT. Those Hilbert spaces are smaller than the ones derived from the \HAKM\ in the cases when both models are defined.}

\medskip

{The model constructed in the present paper generalizes the Kitaev model in a different direction, which we describe in the following.}

A higher gauge theory `lifting' of Kitaev's quantum double model was constructed in \cite{companion,Higher_Kitaev,DG}. {This model, here called \emph{\HKM},} can be defined on manifolds of arbitrary dimension. In \textit{loc cit}, instead of generalising to Hopf algebras, the {\KM} was extended to 2-groups \cite{Baez_Lauda}, from the point of view of discrete higher gauge theory \cite{companion,Pfeiffer}. Let us give some details.
Higher gauge theory is a categorified version of gauge theory over a manifold $M$ where one has not only path-holonomy but also 2D holonomy along surfaces \cite{Baez_Huerta,Martins_Picken} any time a 2-connection is defined on a manifold $M$. In higher gauge theory, instead of having gauge groups, we have their categorified version usually called \emph{2-groups} \cite{Baez_Lauda}. A (in this paper always strict) gauge 2-group is faithfully represented by a crossed module $\mathcal{G}=(E \xrightarrow{\partial} G,\lact)$ of groups \cite{Baez_Lauda}\cite[\S 2.1 \& 2.7]{Brown_Higgins_Sivera}.  Here $G$ and $E$ are groups, the \emph{boundary map} $\partial\colon E \to G$ is a homomorphism and  $\lact$ is a left action of $G$ on $E$ by automorphisms, satisfying appropriate compatibility relations {(the Peiffer relations)}. Roughly speaking, path-holonomies take values in $G$ and surface holonomies in $E$, with a compatibility relation stating that the boundary of the  2D holonomy along a surface coincides with the path holonomy {around} its boundary. This compatibility relation was called \emph{fake-flatness} in \cite{companion}. (This term originates from the closely related notion of fake-curvature of a 2-connection \cite{Breen_Messing}, where its vanishing is essential for surface holonomy to be defined, and implies the above compatibility relation between path and surface holonomy \cite{Baez_Huerta,Martins_Picken}.) 

{As for discrete connections \cite{BaezSF}, discrete gauge 2-group 2-connections on a manifold $M$ can be discretised} given a triangulation, or 2-lattice decomposition, $L$ of $M$ \cite{companion,Pfeiffer}, by specifying their local holonomies along  edges (with a fixed orientation) and plaquettes $P$, provided with a choice of (fixed) vertex $v_P$ in the boundary of $P$,  its \emph{base-point}. The underlying set of 2-gauge configurations in $(M,L)$ is hence $G^{L^1} \times E^{L^2}$ \cite[\S 3.2]{companion}\cite{Higher_Kitaev}.
Inside the set of 2-gauge configurations there is a subset of fake-flat configurations which satisfy the fake-flatness relations for all plaquettes \cite{companion}. These are the configurations that have well-defined 2-dimensional holonomy operators. 

The \HKM\ defined in \cite{companion, Higher_Kitaev} is a model defined on the free vector space, $\mathcal{H}_L^{\fk}$, on {the set of} all fake-flat 2-gauge configurations, and features vertex operators $V_v^g$, where $g \in G$, $v \in L^0$: they implement gauge transformations supported on a vertex $v$; edge operators $E_t^e$, where $t \in L^1$ and $e \in E$, which implement gauge transformations supported on an edge $t$, and also blob operators $B_b^{\delta_a}$, where $a \in \ker(\partial)$ and $b \in L$, which choose those fake-flat configurations that have 2D holonomy equal to $a$ around a 3-cell (a blob) $b\in L^3$. In order for the latter blob operators to be well-defined (and commute with all edge operators) it is essential to restrict to fake-flat gauge configurations \cite{Higher_Kitaev}. That restriction is not necessary in the 2+1D case as there are no blob operators since a 2-lattice decomposition of a surface has no 3-cells. {However note that if fake-flatness is dropped, there are also issues with commutativity between edge operators, which require imposing additional conventions in the model, as elaborated below.}

{{Following on from the previous paragraph, the  starting point for the model constructed in this paper is a generalisation of the 2+1D version of the \HKM, whose construction was sketched in \cite[pages 7 and 8]{Higher_Kitaev},} here called the \emph{\sHKM}.}  This model and the 2+1D case of the \HKM\ have the same ground-state space, however the \sHKM\ has a larger Hilbert space $\mathcal{H}_L^{\full}:=\CC (G^{L^1}\times E^{L^2})$, the free vector space on the set of all 2-gauge configurations, and a larger repertoire of operators, also containing plaquette operators $F_P^{\delta_g}$, where $g \in G$. Those plaquette operators {choose the configurations whose fake curvature around the plaquette $P$, with initial point $v_P$, is $g$.} 
Passing to the Haar integrals, $\ell=\frac{1}{|G|}\sum_{g \in G} g\in \CC G$, $\Lambda=\frac{1}{|E|}\sum_{e \in E} e \in \CC E$, and $\lambda=\delta_{1_G}\in \CC G^*$, we can define mutually commuting vertex projectors $V_v:=V_{v}^\ell$, plaquette projectors $F_P:=F_{P}^\lambda$ and {(after fixing conventions)} edge projectors $E_t := E_t^\Lambda$, on $\mathcal{H}_L^{\full}$.
The \sHKM\ is given by the following Hamiltonian on $\mathcal{H}_L^{\full}$ (see \cite[Equation (35)]{Higher_Kitaev}):
\begin{equation}\label{eq:2KM}
\sum_{v \in L^0} (1-V_v) + \sum_{t \in L^1} (1-E_t) + \sum_{P \in L^2} (1-F_P).
\end{equation}
This model reduces to Kitaev's quantum double model when considering crossed modules of the form $(\{1\} \xrightarrow{\partial} G)$; see \cite{Higher_Kitaev}.
 {We note that the commutation relations between  edge projectors in the \sHKM\ strongly depend on the convention chosen for edge operators, as we will clarify in this paper, in the context of crossed modules of Hopf algebras. In particular in order that all edge projectors commute in \eqref{eq:2KM}, we will stay in what we call here \emph{adequate} lattice decompositions of surfaces, as defined in \S\ref{sec:conv-cell}. This restriction is mild, and satisfied for instance by triangulations with a total order on the set of vertices.}

\medskip

In this paper, we define a  Hopf-algebraic generalisation of the \sHKM\ in 2+1D, henceforth called the \emph{\HHK\ model}. It takes as input a crossed module of semisimple Hopf algebras $(A \xrightarrow{\partial} H, \lact)$, as defined in \cite{majid2012strict}, and also considered in \cite{FARIAMARTINS2016,fregier2011,emir}. Here $A$ and $H$ are Hopf algebras, $\lact$ is a left action of $H$ on $A$ making $A$ an $H$-module algebra and coalgebra, and $\partial \colon A \to H$ is a Hopf algebra map, which together with $\lact$ satisfies two compatibility conditions (the Peiffer relations, {for crossed modules of Hopf algebras}). {Crossed modules of Hopf algebras  were characterized as a class of quantum 2-groups in \cite{majid2012strict}.}
Given  an oriented  surface $\Sigma$, with a triangulation $L = (L^0, L^1, L^2)$, {with a total order on the set of vertices,} or a more generally {an adequate} cell decomposition (see Definition \ref{def:cell-decomposition}), where each plaquette $P \in L^2$ comes with a choice of base-point, the total  space of the \HHK\ model is:
\[\mathcal{H}_L:=H^{\otimes L^1}\otimes A^{\otimes L^2}.\]
{Combining the Hopf-algebraic setting in \cite{buerschaper-et-al, balsam-kirillov} with the 2-group  setting of \cite{Higher_Kitaev,companion},} in the model we can define:
\begin{itemize}\setlength\itemsep{0em}
    \item Vertex operators, $V_{v,P}^h\colon \mathcal{H}_L \to \mathcal{H}_L$, where $v\in L^0$, $P$ is an adjacent plaquette, and $h \in H$. For each pair $(v,P)$, these operators define a representation of $H$ on $\mathcal{H}_L$. A vertex projector is defined as $V_v:=V_{v,P}^\ell$, where $\ell$ is the Haar integral in $H$. This does not depend on the chosen adjacent plaquette $P$, for $\ell$ is cocommutative.
    \item Edge operators, $E_t^a\colon \mathcal{H}_L \to \mathcal{H}_L$, where $t \in L^1$ and $a \in A$. For each edge $t$, we have a representation of $A$. The Haar integral  $\Lambda\in A$ gives the vertex projector $E_t:=E_t^\Lambda$.
    \item Plaquette operators, $F_P^\psi \colon \mathcal{H}_L \to \mathcal{H}_L$, where $\psi \in H^*$. Again this defines for each pla\-quette $P$, which recall comes equipped with a base-point $v_P$, a representation of $H^*$. If $\lambda$ is the Haar integral of $H^*$, the plaquette projector $F_P$ is defined as $F_P^\lambda$.%, where $v_P$ is the base-point of $P$. 
     \end{itemize}

We explicitly compute the commutation relations between these operators. It follows from the commutation relations, {and the properties of Haar integrals for crossed modules of semisimple Hopf algebras,} that the vertex, edge and plaquette projectors are mutually commuting, {as long as we use adequate cell decompositions.} 
Therefore the Hamiltonian defined as in \eqref{eq:2KM} defines an exactly solvable model. This is our proposal for the \HHK\ model.

The \HHK\ model reduces to the \HAKM\ when $A=\CC$ in  $(A \xrightarrow{\partial} H, \lact)$. (Here $\partial \colon \CC \to H$ is $z \in \CC \mapsto z\,1_H$, the unique Hopf algebra map. The action of $H$ is given by the counit map $\varepsilon\colon H \to \CC$.) More precisely, the {total spaces $\Hil_L$ and $\Hil_L^{\hKit}$ of both models coincide,} and so do the Hamiltonians and vertex and plaquette operators. The additional edge operators $E^z_e$, $e \in L^1$, $z \in \CC$, that were not in the \HAKM, amount to multiplication by $z$.

If $(E\xrightarrow{\partial} G,\lact)$ is a crossed module of groups, then, passing to the group algebras, we have a crossed module of {semisimple} Hopf algebras $(\CC E\xrightarrow{\partial} \CC G,\lact)$. The resulting \HHK\ model coincides with the  \sHKM\ in \cite{Higher_Kitaev} for $(E \xrightarrow{\partial} G, \lact)$. {Given that the ground-state space of the \sHKM\ coincides with that of the \HKM\ on the same triangulated surface $(\Sigma,L)$, it hence follows from \cite[\S 5.2]{companion}, and also \cite{Higher_Kitaev,DG}, that the ground-state space of the \sHKM\ for  $(E\xrightarrow{\partial} G,\lact)$ on $(\Sigma,L)$ does not depend on the triangulation of $\Sigma$ and is canonically isomorphic to the free vector space on the set of homotopy classes of maps from $\Sigma$ to $\mathcal{B}$, the classifying space of the crossed module  $(E\xrightarrow{\partial} G,\lact)$; see \cite{Brown_Higgins,brown_hha,martins_porter07}. (This was interpreted in terms of the Yetter homotopy 2-type TQFT in \cite{Higher_Kitaev}). Therefore Quinn's finite total homotopy TQFT, $\mathscr{Q}_{\mathcal{B}}$, \cite[Lecture 4]{Quinn}\cite[\S 4]{martins_porter21} again gives a TQFT whose state spaces give the ground-state spaces of the \sHKM.}

The crossed module of Hopf algebras $(\CC E\xrightarrow{\partial} \CC G,\lact)$ derived from a crossed module of groups can be generalised. Indeed if $f\colon Y \to X$ is a homomorphism of groups, and we have a left action of $G$ on $X$ and $Y$ by automorphisms, such that $f$ preserves the actions, and such that $\partial(E)\subseteq G$ acts trivially on $X$ and $Y$, we can define a crossed module of semisimple Hopf algebras, $(\Fun{X} \ot \CC E \xrightarrow{\partial} \Fun{Y} \rtimes \CC G, \lact)$. Here $\partial$ means $f^*\otimes \partial$, and the action is the product of the obvious action of $\CC G$ on $\Fun{X} \ot \CC E$ and the trivial action of $\Fun{Y}$.

 We unpack the \HAKM\ derived from $\GEXY$, and the remaining data, in some particular cases and compute its ground-state space, proving that in these particular cases the ground-state space is canonically independent of the triangulation $L$ of the surface. In all of these cases, the ground-state spaces can be derived from {Quinn's  finite total homotopy TQFT} %\cite[Lecture 4]{Quinn} 
 $\mathscr{Q}_{\mathcal{B}}$, for some homotopy finite space $\mathcal{B}$. At this point we will not be self-contained, and use a deep result of Brown--Higgins \cite[Theorem A]{Brown_Higgins}  (see also \cite{brown_hha} and the recent monograph by Brown--Higgins--Sivera \cite[\S 11.4.iii]{Brown_Higgins_Sivera}), describing the set of homotopy classes of maps $M \to \B_\mathcal{A}$, where $\mathcal{A}$ is a crossed complex and $M$ is CW-complex, in terms of homotopy classes of crossed complex maps $\Pi(M) \to \mathcal{A}$, where $\Pi(M)$ denotes the fundamental crossed complex of a CW-complex $M$ \cite{brown_hha}. {(Below note that a crossed complex homotopy between groupoid functors \cite{brown_icen} boils down to a natural transformation.)}
 These techniques to prove triangulation independence 
 {have}
 already been applied in \cite{martins_porter07,companion}, and are further developed in \cite[\S 8.2]{martins_porter21}.

The cases of the $\GEXY$-model that we will consider, on triangulated surfaces, $(\Sigma,L)$, are:
\begin{itemize}  \setlength\itemsep{0em}
    \item The $\GoXo$-case, where $E=Y=\{1\}$. {(Our data reduces to a group $G$ acting on another group $X$ by automorphisms.)} In this case we obtain a coupling between the Kitaev model for $G$ and the $|X|$-state Potts model \cite{Martin_Potts}{\cite[\S 1.1.]{Fatimah}.}
    The ground-state space for $(\Sigma,L)$ is given by the free vector space on the set of equivalence classes of functors  $\pi_1(\Sigma,\Sigma_L^0) \to X // G$, considered up to natural transformations, where $X // G$ is the action groupoid of the action of $G$ on $X$. Here $\pi_1(\Sigma,\Sigma_L^0)$ denotes the fundamental groupoid of $\Sigma$, with set of base-points being the set of vertices $v \in L^0$ in $\Sigma$. 
    
    By applying Brown--Higgins theorem, \cite[Theorem A]{Brown_Higgins}, the latter space is canonically isomorphic to the free vector space on the set of homotopy classes of maps  $\Sigma \to B_{X//G}$. Here $B_{X//G}$ is the classifying space  of the action groupoid $X // G$. Hence the ground-state space does not depend on $L$. Quinn's {finite total homotopy} TQFT, $\mathscr{Q}_{B_{X//G}}$, \cite[Lecture 4]{Quinn}\cite[\S 4]{martins_porter21}, gives a TQFT whose state space on a surface $\Sigma$ coincides with the ground-state space on $(\Sigma,L)$, for any triangulation of $L$ of $\Sigma$. 
    
    \item More generally, the $\GEXo$-case, where $Y=\{1\}$. This includes the \sHKM\ as a special case (when $X=\{1\}$). In this case, the ground-state space for a triangulated surface $(\Sigma,L)$ is canonically given by the free vector space on the set of homotopy classes of crossed module maps $\Pi_2(\Sigma,\Sigma_L^1,\Sigma_L^0) \to (X// E \to X//G,\lact)$, where $\Pi_2(\Sigma,\Sigma_L^1,\Sigma_L^0)$ is the fundamental crossed module of  $(\Sigma,\Sigma_L^1,\Sigma_L^0)$, the underlying CW-complex the triangulation $L$, with its skeletal filtration; see \cite[\S 3.3]{companion}.  Also $(X// E \to X// G,\lact) $  is a   crossed module of groupoids, %$(X// E \to X// G,\lact),$ 
    where $X//E$ and $X//G$ are %the obvious 
    action groupoids. (Here $E$ acts 
    on $X$ trivially, and the action of the groupoid $X//G$ on the groupoid $E//G$ is derived from the action of $G$ on $E$). 
    
    By applying Brown--Higgins theorem, the ground-state space is similarly seen to be independent of the triangulation $L$ of $\Sigma$, and given by the free vector space on the set of homotopy classes of maps $\Sigma \to \mathcal{B}$, where $\mathcal{B}$ is the classifying space of the crossed module ${(X// E \to X// G,\lact)}$. Hence, Quinn's {finite total homotopy} TQFT, $\mathscr{Q}_{\mathcal{B}}$, gives a TQFT whose state space coincides with the ground-state space of the underling \HHK\ model for each surface $\Sigma$.
    
    \item The $\ooXY$-case, when $E$ and $G$ are both the trivial group. In this case the ground-state space on $(\Sigma,L)$ is canonically triangulation independent, and given by $\mathscr{Q}_{B_{X//Y}}(\Sigma)$. Here $B_{X//Y}$ is the classifying space of the action groupoid of the action $\lact$ of $Y$ on $X$, where $y\lact x=f(y)x$. In order to prove that the ground-state space is indeed canonically isomorphic, regardless of the chosen triangulation $L$ of $\Sigma$, to the free vector space on the set of homotopy classes of maps $\Sigma \to B_{X//Y}$, we use a slightly different trick as in the two previous cases. Namely, we consider the dual cell decomposition $(\Sigma,L^*)$ to $(\Sigma,L)$, and identify the ground-state space with the free vector space on the set of equivalence classes of groupoid functors $\pi_1(\Sigma,\Sigma_{L^*}^0) \to X//Y$, considered up to natural transformations. (Note that we now have a base-point of $\pi_1(\Sigma,\Sigma_{L^*}^0)$ for each plaquette of $L$.)
    
    {Cf.\ \cite[\S 2.3]{balsam-kirillov},} indeed, the resulting $\ooXY$-model has a particularly simple expression in the dual cell decomposition ${(\Sigma,L^*)}$. Moreover the model reduces to the Kitaev model based on $Y$ on  ${(\Sigma,L^*)}$  when $X=\{1\}$ {(which essentially is \cite[Lemma~2.6]{balsam-kirillov})}, to the $|X|$-state Potts model if $Y=\{1\}$, and to a groupoid version of Kitaev model (with groupoid $X//Y$), also featuring edge operators, in the general case.
    {The $\ooXY$-model on the dual cell decomposition is also closely related to the construction in \cite{deresende2019quantum}.}
    
    \item Finally, we have the $\GooY$-case when $E$ and $X$ are the trivial group. 
    This is not a {proper} crossed \HHK\ model, contrary to the previous three cases, in that the crossed module of Hopf algebras is   $(\CC \xrightarrow{\partial} \Fun{Y} \rtimes \CC G, \lact)$. The $\GooY$-model is hence a special case of the \HAKM. In particular, the ground-state space {is known to be} canonically triangulation independent \cite{balsam-kirillov}.
    
    It is an open problem whether the Turaev-Viro TQFT derived from the spherical fusion category $\lmod{(\Fun{Y} \rtimes \CC G)}$, whose state spaces give the ground-state spaces of the $(1,G,1,Y)$-case, is also a particular case of {Quinn's finite total homotopy TQFT} and, in particular, if it has a homotopical explanation.
\end{itemize}

We finish this introduction by presenting the following set of open problems {which arise from the present paper}:
\begin{enumerate}  \setlength\itemsep{0em}
    \item Are the ground-state  spaces of the full $\GEXY$-model triangulation independent and is there a TQFT giving these ground-state spaces, which is a particular case of {Quinn's finite total homotopy TQFT}? 
    \item {The following are likely the most important open problems resulting from our construction.} Considering general crossed modules of semisimple Hopf algebras: 
    \begin{itemize}\setlength\itemsep{0em}
        \item Is the ground-state space of the \HHK\ model in $(\Sigma,L)$ always canonically invariant of the triangulation $L$ of $\Sigma$?
        \item Is there an underpinning TQFT whose state spaces give the ground-state spaces of the \HHK\ %Hopf higher Kitaev 
        model?
        In particular, what are the quantum gates that could be realized by the underlying mapping class group actions?
    \end{itemize}
    \item Majid defined in \cite{majid2012strict} a  generalisation of crossed modules of Hopf algebras that he called \emph{“braided crossed modules of Hopf algebras”}  $(B \xrightarrow{\partial} H,\lact)$. Here $B$ is a Hopf algebra in $\mathcal{Z}(\lmod{H})$, the braided category of Yetter-Drinfeld modules over $H$. It is an open problem whether our construction of the \HHK\ model generalises to this setting.
    
    \item {What local excitations does the \HHK\ model admit? Previously \cite{koppen}, defects of co-dimensions $2$ and $1$ (i.e.\ point-like and string-like excitations) in the Kitaev model have been studied in the general Hopf-algebraic setting.
    It is an interesting open problem to extend this to the \HHK\ model constructed in the present paper.}
    
    \item {Finally, is there a 3+1D version of the \HHK\ similar to the original 2-group \HKM\  \cite{Higher_Kitaev,companion,DG}? This extension would seem quite tricky when the Hopf algebra $H$ in $(A\xrightarrow{\partial} H,\lact)$ is not cocommutative. This is because there is not a natural way to define vertex operators, since there is no given cyclic order on the edges incident to a vertex.} A model however exists for 
 $(\CC E\xrightarrow{\partial} \CC G,\lact)$ where $(E\xrightarrow{\partial} G,\lact)$ is a crossed module of groups (as per the construction in \cite{Higher_Kitaev,companion,DG}).
 
\end{enumerate}

\noindent\textbf{Some broader motivation.}  
A primary question for the field-theoretic approach to physical modelling of gauge phenomena is that of the relationship between the representation of space(-time),
and the symmetry structure underlying the gauge fields. 
This may be cast as the choices of two categories (space and gauge, informally put)
--- yielding the category of functors between them (see e.g. 
\cite[\S2]{Pfeiffer}, \cite{Baez_Huerta,companion})
as (a basis for) the space of states.
Such a categorical formulation can somewhat obscure the direct tie to the physics
being modelled, but facilitates the development of candidate generalisations of
the `classical' suite of models.
For example in \cite{Pfeiffer,companion} one passes from 
the functor category of gauge configurations
 to a lift to higher gauge configurations --
 nominally 2-functors from a 2-lattice 2-category to a 2-group.
In this setting space(-time) is modelled by a generalisation of a CW-complex
-- depending on the requirements of the Hamiltonian,
and the gauge group becomes a kind of higher group
(or indeed lower group!
{in the sense that in Potts models it can really be just a set, since the interaction is a delta-function, furthermore requiring only graph data for the lattice}).
In parallel to this approach it is natural to ask what {\em other} 
generalisations of group characterisations of symmetry can be supported
in principle (the \$64 question is What is demanded by physical observation?,
but this begs also the question of where and how to observe -- see later,
and cf. \cite{Levin_Wen}).
And what impact on the formalisation of space--time this might have.

Hopf algebras, and in particular higher versions such as Hopf crossed modules \cite{majid2012strict,FARIAMARTINS2016}, offer a possible line of generalisation from gauge groups in this context because Hopf algebras (miraculously) bridge between the combinatorial and geometric worlds,
at least in 2D \cite{manin2018quantum}. However while the linear structure on the resultant Hilbert space
in the group cases is essentially passive, it plays a crucial role in the Hopf case
(confer e.g.\ \cite{Bullock,meusburger2016hopf}). Here we take a first step to probe this challenge in a higher setting -- a higher theory, formally, but staying on surfaces.
The version of this analysis relevant for higher Kitaev models (and hence relatively 
straightforward use of Whitehead's underlying free fundamental crossed module
technology) is discussed for example in 
\cite{companion}.

\section{Review of semisimple Hopf algebras and Hopf crossed modules}

Algebras here will by default be unital, associative and finite-dimensional, and  over some field, $\kappa$, which we will later on take to be $\CC$.
 The proofs of some crucial results on semisimple Hopf algebras and their Haar integrals can be found for example in \cite{larson_radford,larson_sweedler} and  \cite{RadfordBook}. These results had already been applied in \cite{balsam-kirillov,meusburger2016hopf,buerschaper-et-al} to construct Hopf algebraic versions of the \KQDM.
 
Crossed modules of Hopf algebras, {which are discussed later in this section, were defined in \cite{majid2012strict}, and also discussed in \cite{FARIAMARTINS2016,fregier2011,emir}.}

\subsection{Hopf algebras}
A Hopf algebra \cite{Majid_book} $H$ is a unital associative algebra, $(H,\mu,\nu)$, where the unit map $\nu: \kappa \to H$ sends $t$ to $t 1_H$, together with a compatible structure of a co-unital co-associative co-algebra, given by a unital algebra map $\Delta : H \to H \ot H$, and a co-unit $\eps : H \to \CC$.
We denote its antipode as usual by $S : H \to H$, or $S_H : H \to H$ if we want to emphasize the Hopf algebra it belongs to.
A Hopf algebra is thus a 6-tuple 
$(H,\mu,\nu,\Delta,\eps,S)$. We often write the multiplication $\mu$ as $\mu(x,y)=xy$, {and will usually write the 6-tuple simply as $H$.} But we may also write $H$ to denote the underlying vector space; or only the algebra structure. Such variations will be clear from context.

We will make extensive use of the Sweedler notation for co-multiplication:
\[ h_{(1)} \ot h_{(2)} := \Delta(h), \quad \textrm{ where } h \in H , \]
which is in general a sum of pure tensors, even though we omit the summation symbol
and summation variable.
Due to co-associativity, any $n$-fold composition of the co-multiplication tensored with suitable identity maps has the same result, and one may therefore write:
\[ h_{(1)} \ot h_{(2)} \ot\cdots\ot h_{(n)} := \Delta^{(n-1)}(h) 
= \Delta \ot \id{H}^{\ot (n-2)}(\cdots
\Delta\ot\id{H}  
(\Delta(h))), \quad \textrm{where } h \in H . \]

A fact we will frequently use when we construct our model is the following. 
\begin{lemma}\label{lem:coproduct-cocommut-el}
Any multiple coproduct of a cocommutative element $h \in H$, i.e.\ one that satisfies $h_{(1)} \ot h_{(2)} = h_{(2)} \ot h_{(1)}$, is cyclically invariant, i.e.\
\begin{equation} \label{eq:cocommutative-elements-cyclically-invariant}
h_{(1)} \ot h_{(2)} \ot\cdots\ot h_{(n)} = h_{(2)} \ot h_{(3)} \ot\cdots\ot h_{(n-1)} \ot h_{(1)} = \ldots = h_{(n)} \ot h_{(1)} \ot h_{(2)} \cdots\ot h_{(n-1)}.  
\end{equation}
\end{lemma}

\subsection{Left and right integrals for Hopf algebras, and semisimplicity}\label{Sec:semi}
Nothing here is new. We follow chapter 10 of \cite{RadfordBook} very closely.} {We work over a general field $\kappa$.}
\begin{definition}
(Cf. \cite[Definition 10.1.1]{RadfordBook}). Let $H$ be a Hopf algebra.
A \emph{left-integral}, respectively right integral,  for $H$ is an element, $\Lambda_l \in H$, respectively $\Lambda_r\in H$, such that, for all $a \in H$, we have:
\[ a \Lambda_l=\eps(a) \Lambda_l, \textrm{  respectively } \Lambda_r a=\eps(a) \Lambda_r.\]  
An \emph{integral}, $\Lambda$, for $H$, is an element, $\Lambda\in H$, that is at the same time a left and a right-integral.
\end{definition}

{Non-zero  (left or right) integrals can only exist if $H$ is  finite-dimensional; see \cite[Proposition 10.2.1]{RadfordBook}. On the other hand, if $H$ is finite-dimensional,  then both (left and right) ideals, of left-integrals, and of right-integrals, are one-dimensional (see e.g. \cite[Theorem 10.2.2, (a)]{RadfordBook}). Hence, non-zero left and right-integrals always exist in finite dimensional Hopf algebras. In the latter case, $S$ is bijective, by Larson-Sweedler Theorem \cite[Proposition 2]{larson_sweedler}, thus $S$ gives a bijection between the ideals of left and right integrals.}

\begin{definition}(Cf. \cite[Definitions 10.2.3 and 10.2.5]{RadfordBook}.) Let $(H,\mu,\nu,\Delta,\eps,S)$ be  finite-dimensional Hopf algebra. We say that $H$ is  \emph{semisimple} if its underlying unital algebra, $(H,\mu,\nu)$, is semisimple. We say that $H$ is \emph{unimodular} if the  ideals of left-integrals and of right-integrals coincide. 
\end{definition}
The following result is key for the construction in this paper, as was in \cite{balsam-kirillov}. 
For a proof see  \cite[Theorem 10.3.2 and Corollary 10.3.3]{RadfordBook}.
\begin{theorem}\label{th:semisimple-char0}
Let $H$ be a finite-dimensional Hopf  algebra over a field $\kappa$.  If $H$ is semisimple, then $H$ is unimodular.~ 
Furthermore, the following conditions are equivalent:
\begin{enumerate}  \setlength\itemsep{0em}
    \item  $H$ is semisimple;
    \item for some left integral, equivalently for all non-zero left integrals,  $\Lambda_l$,   we have $\eps(\Lambda_l)\neq 0$.
\end{enumerate}
\end{theorem}
On the assumption that $H$ is finite dimensional and semisimple, hence unimodular, it then follows that the ideal of integrals is one dimensional, and furthermore invariant under $S$.

\begin{definition}[Haar integral] \label{def:haar-integral}{Let $H$ be a finite-dimensional semisimple Hopf algebra, over a field $\kappa$. To the  unique integral $\ell \in H$, satisfying $\eps(\ell)=1$ we call \emph{the Haar integral on $H$.}}
\end{definition}
\noindent We can also define the  Haar integral for $H$ as  the unique non-zero idempotent $\ell \in H$ satisfying \[x \ell = \eps(x) \ell = \ell x, \textrm{ for all } x \in H.\]
{It follows from our discussion that, if $\Lambda$ is the Haar integral of $H$, semisimple, then $S(\Lambda)=\Lambda$.}

\subsection{Semisimple Hopf algebras in characteristic zero}
When the field $\kappa$ has characteristic zero, stronger results concerning semisimplicity for Hopf algebras can be obtained. We follow \cite[Chapter 16]{RadfordBook}. Our discussion 
closely parallels  that of \cite[Section~1]{balsam-kirillov} and \cite[Appendix A1]{meusburger2016hopf}, which similarly then go on to define Hopf gauge theory models on a lattice.
\begin{theorem}[{Adapted from \cite[Theorem 16.1.2]{RadfordBook}}] \label{prop:haar-integral0} 
Let $H$ be a finite-dimensional Hopf algebra over a field $\kappa$ of  characteristic zero. Then the following conditions are equivalent:
\begin{multicols}{2}
\begin{enumerate}  \setlength\itemsep{0em}
    \item the antipode $S$ is involutive, i.e., $S^2=\id{H}$;
    \item {$H$ has a non-zero integral $\ell$, with $\ell^2=\ell$;}
    \item $H$ is semisimple;
    \item $H^*$ is semisimple.
\end{enumerate}
\end{multicols}
Moreover, if $H$ is semisimple, the Haar integral for $H^*$ is given by $f\colon H \to \kappa$, where
\[f(a)=\frac{1}{\dim_\kappa(H)}\mathrm{Tr}(r(a)).\]
%Here $r(a)\colon H \to H$ is such that $r(a)(x)=xa.$ 
Here $r(a)\colon H \to H$, $x\mapsto xa.$. 
(Note that indeed:  $  \eps_{H^*}(f)=f(1_H)=1_{\kappa}$).
\end{theorem}
Observe that $f\in H^*$ above is cocommutative. Indeed, given $a,b \in H$:
\begin{align*}
    (\Delta_{H*} (f))(a,b)&=f(ab)= \frac{1}{\dim_\kappa(H)}\mathrm{Tr}(r(ab))= \frac{1}{\dim_\kappa(H)}\mathrm{Tr}\big (r(b)\circ r(a)\big)\\&= \frac{1}{\dim_\kappa(H)}\mathrm{Tr}\big (r(a)\circ r(b)\big)
 =(\Delta_{H*} (f))(b,a).
\end{align*}

{The following result will be crucial for us to construct our model.}
\begin{proposition}[{\cite[Theorem 1.2]{balsam-kirillov}}] \label{prop:haar-integral}
Let $H$ be a finite-dimensional semisimple Hopf algebra over $\kappa$, of characteristic zero.
The Haar integral $\ell \in H$ is cocommutative.
\end{proposition}
\begin{proof}
{By the previous theorem},  the Haar integral of $H$ can be expressed as $\ell = \frac{1}{\dim(H)} \chi_{H^*}$, where $\chi_{H^*} \in (H^*)^*$ is the regular character of the dual Hopf algebra $H^*$, and where the canonical identification $(H^*)^* \cong H$ is implicit.
The cocommutativity of the Haar integral then is an immediate consequence of the cyclicity of the trace, {as the previous calculation shows.}
\end{proof}
\begin{example}\label{ex:CG}
Let $G$ be a finite group. The group algebra $\CC G$ is semisimple by Maschke's theorem. We will regard $\CC G$ as a Hopf algebra in the usual way with $\Delta(g)=g \otimes g$, $\eps(g)=1$ and $S(g) = g^{-1}$, for all $g \in G$.
The Haar integral is $\ell=\frac{1}{|G|} \sum_{g \in G} g$. 
\end{example}
\begin{example}\label{ex:FG}
Given a finite group $G$, the algebra $\Fun{G}$ of functions $G \to \CC$, with the pointwise product, can be turned into a Hopf algebra by putting $\eps(f) = f(1_G)$ and $\Delta(f)(x\ot y) = f(xy)$ for all $x,y \in G$,
noting the isomorphism $\Fun{G} \otimes \Fun{G} \cong \Fun{G \times G}$.  This Hopf algebra is  semisimple, since it is the dual of the group algebra of $G$.
The Haar integral is $\delta_{1_G}$, where $\delta_{1_G}(x)=1$ if $x=1_G$, and otherwise $0$.
\end{example}

\subsection{Crossed modules of Hopf algebras}
Crossed modules of Hopf algebras were defined by Majid in \cite{majid2012strict}. See also \cite{FARIAMARTINS2016,emir}.
\begin{definition}\label{def:triv-braid-hopf-crossed-module}
A \emph{crossed module of Hopf algebras}, or \emph{Hopf crossed module} $(A \xrightarrow{\partial} H, \lact)$, consists of:
\begin{itemize}  \setlength\itemsep{0em}
\item
Hopf algebras $A$ and $H$, each with an invertible antipode (this condition is redundant in the finite-dimensional case),
with a Hopf algebra morphism $\partial \colon A \to H$, called \emph{boundary map}, and
\item
a left $H$-action $\lact \colon H \otimes A \to A$.
\end{itemize}
These are such that:
\begin{itemize}\setlength\itemsep{0em}
\item the action turns
$A$ into an \emph{$H$-module algebra}, i.e.:
\begin{align}\label{eq:H-mod-algebra}
    h \lact (ab)&=(h_{(1)}\lact a)\,(h_{(2)}\lact b), &&
    h \lact 1_A = \eps(h) 1_A,
\end{align}
and an \emph{$H$-module coalgebra}, i.e.:
\begin{align}
    \Delta(h \lact a)&=\big (h_{(1)}\lact a_{(1)}\big) \otimes \big (h_{(2)}\lact a_{(2)}\big), &&
    \eps(h \lact a) = \eps(h) \eps(a),
\end{align}
the latter two equations holding for all $h \in H,$ and all $a,b \in A$;
\item  the \emph{Yetter-Drinfeld condition} holds:
\begin{equation} \label{eq:yetter-drinfeld-condition-trivial}
h_{(1)} \otimes (h_{(2)} \lact a) = h_{(2)} \otimes (h_{(1)} \lact a), \qquad \text{ for all } h \in H, \text{ and all }  a \in A;
\end{equation}
\item the two \emph{Peiffer relations}, below, hold, for all $h \in H$ and all $a,b \in A$:
\begin{multicols}{2}
\begin{enumerate}\setlength\itemsep{0em}
\item[Pf 1:]	$\partial(h \lact a) %= \text{ad}_{h}(\partial(a))
= h_{(1)} \, \partial(a) \, S h_{(2)},$ 
\item[Pf 2:]	$\partial(a) \lact b
= a_{(1)} \, b\, S a_{(2)}$.
\end{enumerate}\end{multicols}
\end{itemize}
\end{definition}

\begin{example}\label{ex:gr} A \emph{crossed module of groups} $(E \xrightarrow{\partial} G, \lact)$ \cite{martins_porter07,brown_hha,Brown_Higgins} is given by a group homomorphism $\partial \colon E \to G$, together with a left action $\lact$ of $G$ on $E$ by automorphisms, such that the Peiffer relations, in the group case, as below, are satisfied:
\begin{multicols}{2}
\begin{enumerate}\setlength\itemsep{0em}
    \item $\partial(g \lact e)=g \, \partial(e) \, g^{-1}$, for all $g \in G$ and $e \in E$,
    \item $\partial(e) \lact f=e \, f \, e^{-1}$, for all $e, f \in E$.
\end{enumerate}
\end{multicols}
If we consider the induced map $\partial \colon \CC E \to \CC G$ on group algebras and the induced ‘linearised’ action $\lact$ of $\CC G$ on $\CC E$, then 
$(\CC E\xrightarrow{\partial} \CC G, \lact)$ 
is a crossed module of Hopf algebras; see \cite{FARIAMARTINS2016}. This example will be generalised in Subsection \ref{sec:main_example}.
\end{example}

{A crucial fact we will use about a 
Hopf crossed module 
$(A \xrightarrow{\partial} H,\lact)$ is the fact below, in \cite[page 4]{majid2012strict}. We will delay the proof to later in this section. 
 This will be a particular case of Lemma \ref{llem_Sprop}.}
\begin{lemma}\label{rem: S H-linear}
The antipode $S_A \colon A \to A$ is $H$-linear, namely:
\[S_A(h\lact a)=h\lact S_A(a), \textrm{ if $a \in A$ and $h \in H$.}\]
\end{lemma}

\begin{remark}[{Cf \cite[page 4]{majid2012strict}}]
Definition \ref{def:triv-braid-hopf-crossed-module} can be equivalently rephrased as follows.
$H$ is a Hopf algebra, $(A,\lact)$ together with the trivial left $H$-co-action is a Hopf algebra in the braided monoidal category of (left-left) Yetter-Drinfeld-modules over $H$, and $\partial \colon A \to H$ is a Hopf algebra morphism satisfying the Peiffer relations.
This follows straightforwardly by spelling out the definition of the braided monoidal category of Yetter-Drinfeld modules
(see e.g. \cite[\S10.6]{montgomery1993hopf}).
Note that equation \eqref{eq:yetter-drinfeld-condition-trivial} is precisely the Yetter-Drinfeld condition in the case of a trivial co-action.

We remark also that from this point of view the antipode $S \colon A \to A$ is by definition $H$-linear, since any morphism in the category of $H$-Yetter-Drinfeld modules is in particular an $H$-module morphism.
\end{remark}

The Yetter-Drinfeld condition \eqref{eq:yetter-drinfeld-condition-trivial}, {combined with the coassociativity of $\Delta$,} 
{implies}
the following property, which we note here because we will frequently make use of it.
\begin{lemma} \label{lem:yetter-drinfeld-for-many-factors}
For any $h \in H$ and $a_1 \ot  \cdots \ot a_k \in A^{\ot k}$, the expression,
\[ h_{(1)} \ot\cdots\ot h_{(n)} \ot h_{(n+1)} \lact a_1 \ot\cdots\ot h_{(n+k)} \lact a_k \in H^{\ot n} \ot A^{\ot k}, \]
is invariant under permutations, of $(h_{(1)}, \dots, h_{(n)}, h_{(n+1)}, \dots, h_{(n+k)})$, which preserve the order of the first $n$ factors. So, explicitly, if $\sigma\colon \{1,\dots, n+k\} \to  \{1,\dots, n+k\}$ is a bijection such that $\sigma$ is strictly increasing in $\{1,\dots, n\}$, then:
\[ h_{(1)} \ot\cdots\ot h_{(n)} \ot h_{(n+1)} \lact a_1 \ot\cdots\ot h_{(n+k)} \lact a_k=h_{(\sigma(1))} \ot\cdots\ot h_{(\sigma(n))} \ot h_{(\sigma(n+1))} \lact a_1 \ot\cdots\ot h_{(\sigma(n+k))} \lact a_k. \]
\end{lemma}
For instance, where each step follows from a one-step application of the  Yetter-Drinfeld condition \eqref{eq:yetter-drinfeld-condition-trivial}, we have:
\begin{align*}
 h_{(1)} \ot h_{(2)} \ot h_{(3)} \ot h_{(4)} \lact a_1 \ot  h_{(5)} \lact a_2 &=
 h_{(1)} \ot h_{(2)} \ot h_{(4)} \ot h_{(3)} \lact a_1 \ot  h_{(5)} \lact a_2\\
  &=
 h_{(1)} \ot h_{(3)} \ot h_{(4)} \ot h_{(2)} \lact a_1 \ot  h_{(5)} \lact a_2\\
&= h_{(2)} \ot h_{(3)} \ot h_{(4)} \ot h_{(1)} \lact a_1 \ot  h_{(5)} \lact a_2\\
&= h_{(2)} \ot h_{(3)} \ot h_{(5)} \ot h_{(1)} \lact a_1 \ot  h_{(4)} \lact a_2\\
%&= h_{(2)} \ot h_{(3)} \ot h_{(5)} \ot h_{(1)} \lact a_1 \ot  h_{(4)} \lact a_2\\
&= h_{(2)} \ot h_{(4)} \ot h_{(5)} \ot h_{(1)} \lact a_1 \ot  h_{(3)} \lact a_2\\
&= h_{(3)} \ot h_{(4)} \ot h_{(5)} \ot h_{(1)} \lact a_1 \ot  h_{(2)} \lact a_2.
\end{align*}
Moreover:
\begin{align*}
 h_{(1)} \ot h_{(2)} \ot h_{(3)} \ot h_{(4)} \lact a_1 \ot  h_{(5)} \lact a_2 &=h_{(1)} \ot h_{(2)} \ot h_{(3)} \ot h_{(5)} \lact a_1 \ot  h_{(4)} \lact a_2.
\end{align*}
\begin{proof}
 Consider an expression of the form below, \[h_{(\sigma(1))} \ot\cdots \ot h_{(\sigma(j))} \ot h_{(\sigma(j+1))} \ot\dots  \ot h_{(\sigma(n))} \ot h_{(\sigma(n+1))} \lact a_1 \ot\cdots\ot h_{(\sigma(n+k))} \lact a_k,\]
 where, as above, $\sigma\colon \{1,\dots, n+k\} \to  \{1,\dots, n+k\}$ is a bijection  that is strictly increasing in $\{1,\dots, n\}$.
 
 We first note that, if $\sigma(j+1)-\sigma(j)\ge 2$, then
 \begin{multline}\label{eq:referA}
   h_{(\sigma(1))} \ot\cdots \ot h_{(\sigma(j))} \ot h_{(\sigma(j+1))} \ot\dots  \ot h_{(\sigma(n))} \ot h_{(\sigma(n+1))} \lact a_1 \ot\cdots\ot h_{(\sigma(n+k))} \lact a_k\\=h_{(\sigma'(1))} \ot\cdots \ot h_{(\sigma'(j))} \ot h_{(\sigma'(j+1))} \ot\dots  \ot h_{(\sigma'(n))} \ot h_{(\sigma'(n+1))} \lact a_1 \ot\cdots\ot h_{(\sigma'(n+k))} \lact a_k.
 \end{multline}
 Here $\sigma'=\tau_{(\sigma(j),\sigma(j)+1)}\circ \sigma$, the usual composition of permutations, and where $\tau_{(\sigma(j),\sigma(j)+1)}$ is the transposition exchanging $\sigma(j)$ and $\sigma(j)+1$. This follows by a one-step application of the Yetter-Drinfeld condition \eqref{eq:yetter-drinfeld-condition-trivial}, since a term of the form
 $h_{(\sigma(j))} \otimes h_{(\sigma(j)+1)}\lact a$, where $a \in A$, which is then equal to $h_{(\sigma(j)+1)} \otimes h_{(\sigma(j))}\lact a$,  appears in the left-hand-side of
 \eqref{eq:referA}.

Second, note that, analogously,
\begin{multline}\label{eq:referB}
h_{(\sigma(1))} \ot\cdots\ot h_{(\sigma(n))} \ot h_{(\sigma(n+1))} \lact a_1 \ot\cdots\ot h_{(\sigma(n+k))} \lact a_k\\=h_{(\sigma'(1))} \ot\cdots\ot h_{(\sigma'(n))} \ot h_{(\sigma'(n+1))} \lact a_1 \ot\cdots\ot h_{(\sigma'(n+k))} \lact a_k,
\end{multline}
where $\sigma'=\tau_{(\sigma(n), \sigma(n)+1)}\circ \sigma$.

{By successively applying \eqref{eq:referA} and \eqref{eq:referB}, it follows that the statement of the lemma holds if $\sigma$ is an $(n,k)$-shuffle. (This is exemplified in the first calculation just before this proof.)}

{The general case follows by noting that, again by Yetter-Drinfeld condition \eqref{eq:yetter-drinfeld-condition-trivial}, and since consecutive transpositions generate all of the symmetric group, if $\pi$ is a permutation of $\{1,\dots,k\}$, then:
\[ h_{(1)} \ot\cdots\ot h_{(n)} \ot h_{(n+1)} \lact a_1 \ot\cdots\ot h_{(n+k)} \lact a_k =
h_{(1)} \ot\cdots\ot h_{(n)} \ot h_{(n+\pi(1))} \lact a_1 \ot\cdots\ot h_{(n+\pi(k))} \lact a_k.\] The remaining details are left to the reader.} 
\end{proof}
{When we construct our model, we will also make use of the following result, which combines the previous result with Lemma~\ref{lem:coproduct-cocommut-el}.
\begin{lemma} \label{lem:cocom-yetter-drinfeld-for-many-factors}
For any $h \in H$, cocommutative, 
and $a_1 \ot  \cdots \ot a_k \in A^{\ot k}$, the expression:
\[ h_{(1)} \ot\cdots\ot h_{(n)} \ot h_{(n+1)} \lact a_1 \ot\cdots\ot h_{(n+k)} \lact a_k \in H^{\ot n} \ot A^{\ot k} \]
is invariant under permutations, of $(h_{(1)}, \dots, h_{(n)}, h_{(n+1)}, \dots, h_{(n+k)})$, that are cyclic when restricted to the first $n$ terms. So, explicitly, if $\sigma\colon \{1,\dots, n+k\} \to  \{1,\dots, n+k\}$ is a bijection, for which  the restriction of $\sigma$  to $\{1,\dots, n\}$ gives a cyclic permutation of $\{1,\dots,n\}$, then:
\[ h_{(1)} \ot\cdots\ot h_{(n)} \ot h_{(n+1)} \lact a_1 \ot\cdots\ot h_{(n+k)} \lact a_k=h_{(\sigma(1))} \ot\cdots\ot h_{(\sigma(n))} \ot h_{(\sigma(n+1))} \lact a_1 \ot\cdots\ot h_{(\sigma(n+k))} \lact a_k. \]
\end{lemma}
For instance,
 where each step follows from one-step applications of \eqref{eq:yetter-drinfeld-condition-trivial} or \eqref{eq:cocommutative-elements-cyclically-invariant}, we have:
\begin{align*}
 h_{(1)} \ot h_{(2)} \ot h_{(3)} \ot h_{(4)} \lact a_1 \ot  h_{(5)} \lact a_2 &=
 h_{(5)} \ot h_{(1)} \ot h_{(2)} \ot h_{(3)} \lact a_1 \ot  h_{(4)} \lact a_2\\
 &=
 h_{(4)} \ot h_{(1)} \ot h_{(2)} \ot h_{(3)} \lact a_1 \ot  h_{(5)} \lact a_2\\
  &=
 h_{(3)} \ot h_{(1)} \ot h_{(2)} \ot h_{(4)} \lact a_1 \ot  h_{(5)} \lact a_2\\
  &=
 h_{(3)} \ot h_{(1)} \ot h_{(2)} \ot h_{(5)} \lact a_1 \ot  h_{(4)} \lact a_2.
\end{align*}
}
\begin{proof}
The calculation just above indicates how the proof is performed,
\begin{align*}
     h_{(1)} &\ot\cdots\ot h_{(n)} \ot h_{(n+1)} \lact a_1 \ot\cdots\ot h_{(n+k)} \lact a_k\\ &= h_{(n+k)}\otimes h_{(1)} \ot\cdots\ot h_{(n-1)} \ot h_{(n)} \lact a_1 \ot\cdots\ot h_{(n+k-1)} \lact a_k \\
     &= h_{(n+k-1)}\otimes h_{(1)} \ot\cdots\ot h_{(n-1)} \ot h_{(n)} \lact a_1 \ot \cdots\ot h_{(n+k-2)} \lact a_{k-1}  \ot h_{(n+k)} \lact a_k \\
     &= h_{(n+k-2)}\otimes h_{(1)} \ot\cdots\ot h_{(n-1)} \ot h_{(n)} \lact a_1 \ot \cdots\ot h_{(n+k-3)} \lact a_{k-2}  \\ & \phantom{-------------------}\ot h_{(n+k-1)} \lact a_{k-1}  \ot h_{(n+k)} \lact a_k \\
     & = \dots \textrm{ apply \eqref{eq:yetter-drinfeld-condition-trivial} several times } \dots\\
      &= h_{(n+1)}\otimes h_{(1)} \ot\cdots\ot h_{(n-1)} \ot h_{(n)} \lact a_1 \ot h_{(n+2)} \lact a_{2}  \ot \dots  \ot h_{(n+k)} \lact a_k \\
      &= h_{(n)}\otimes h_{(1)} \ot\cdots\ot h_{(n-1)} \ot h_{(n+1)} \lact a_1 \ot h_{(n+2)} \lact a_{2}  \ot \dots  \ot h_{(n+k)} \lact a_k.
     \end{align*}
We can now apply \eqref{eq:yetter-drinfeld-condition-trivial} several times and arbitrarily change the order of the last $k$ elements, noting that any permutation can be generated by products of consecutive transpositions.
\end{proof}

\subsection{Review of crossed product Hopf algebras}

The following structure of a \emph{cross product} (also known as \emph{smash product}) is well known and arises naturally from a crossed module of Hopf algebras \cite[Section 2]{majid2012strict}.

{The definition below is a particular case of \cite[Theorem 6.2.2.]{Majid_bigbook}, where we consider  the trivial right-coaction, $\beta\colon H \to H\otimes A$, of $A$ on $H$, where $\beta(h)=h\otimes 1_A$. 

\begin{definition}\label{def:crossed-product-algebra} \textbf{/ Proposition.}
Let $H$ and $A$ be a Hopf algebras. Suppose that $A$ is a left $H$-module algebra.
Then the \emph{cross product algebra} $A \rtimes H$, {as in   \cite[Theorem 1.6.6]{Majid_bigbook},} is the algebra with underlying vector space $A \ot H$, and multiplication:
\[ (a \ot h) \cdot_\rtimes (b \ot k) := a(h_{(1)} \lact b) \ot h_{(2)} k , \]
for $a, b \in A$ and $h, k \in H$. {This is a unital algebra, with unit $1_A\otimes 1_H$.}

If $A$ is additionally a $H$-module coalgebra, such that the Yetter-Drinfeld condition \eqref{eq:yetter-drinfeld-condition-trivial}, 
holds,
then the algebra $A \rtimes H$ becomes a Hopf algebra, with the usual tensor product coalgebra structure, 
\[\Delta(a \ot h) := (a_{(1)} \ot h_{(1)}) \ot (a_{(2)} \ot h_{(2)}),\]
together with
$\eps(a\otimes h)=\eps(a) \eps(a)$, and
\[S(a\otimes h)=(1_A \otimes S_H(h)) \cdot_\rtimes  (S_A(a)\otimes 1_H).\]
\end{definition}
\begin{proof}
    This is a particular case of \cite[Proposition 1.6.16 and Theorem 6.2.2.]{Majid_bigbook}, where the right-coaction of $A$ on $H$ is the trivial coaction, $\beta(h)=h\ot 1_A$. Indeed \cite[(6.8)]{Majid_bigbook} follows from the fact that $A$ is an $H$-module coalgebra, \cite[(6.9)]{Majid_bigbook} from \eqref{eq:H-mod-algebra}, and \cite[(6.10)]{Majid_bigbook} from \eqref{eq:yetter-drinfeld-condition-trivial}. 
\end{proof}

We can now finish the proof of Lemma \ref{rem: S H-linear}. More generally we have that:
\begin{lemma}\label{llem_Sprop}
Let $H$ and $A$ be Hopf algebras.   If $A$ is an $H$-module algebra and coalgebra, satisfying moreover the Yetter-Drinfeld conditions \eqref{eq:yetter-drinfeld-condition-trivial}, then, given $h\in H$ and $a \in A$, we have:
    \[S_A(h \lact a)=h\lact S_A(a).\]
\end{lemma}
\begin{proof} We use the Hopf algebra structure of  $A \rtimes H$. 
Given $g \in H$ and $a \in A$, we consider the element, \[(1\ot g) \cdot_\rtimes  (a\ot 1)=(g_{(1)}\lact a) \ot g_{(2)}=\big (g_{(1)} \lact a \ot 1 \big)\cdot_\rtimes \big( 1\ot g_{(2)}\big)=\big (g_{(2)} \lact a \ot 1 \big)\cdot_\rtimes \big( 1\ot g_{(1)}\big).\] 
Let us compute its antipode, in two different ways, namely
\begin{align*}
    S\big( (1\ot g) \cdot_\rtimes  (a\ot 1) \big) &= (S_A(a) \ot 1)\cdot_\rtimes \big(1 \ot S_H(g)\big)=S_A(a) \ot S_H(g),
    \end{align*}
    and
\begin{align*}
S\big (( g_{(2)} \lact a \ot 1 )\cdot_\rtimes ( 1\ot g_{(1)})\big) &=\big(1\ot S_H(g_{(1)})\big) \cdot_\rtimes \big( S_A(g_{(2)}\lact a)\ot 1\big).  \end{align*}
Now note, given $h\in H$,
\[ (1 \ot h_{(1)}) \cdot_\rtimes (S_A(a) \ot S_H(h_{(2)}))=\big(h_{(1)} \lact S_A(a)\big) \otimes \big(h_{(2)} S_H(h_{(3)})\big) =\big(h \lact S_A(a)\big)\ot 1_H, \]
and,
\[(1 \ot h_{(1)} )\cdot_\rtimes  (1\ot S(h_{(2)})) \cdot_\rtimes ( S(h_{(3)}\lact a)\ot 1)= S_A(h \lact a)\ot 1_H.\]
\end{proof}

\subsection{Crossed modules of semisimple Hopf algebras }

{From this section onwards, we work over $\mathbb{C}$. We will make strong use of Theorem \ref{prop:haar-integral0}.} 

The following is a key lemma 
to prove the commutativity between vertex and edge projectors in the commuting-projector Hamiltonian model defined in Section \ref{sec:model} below; see Theorem~\ref{thm:hamiltonian}.}

\begin{lemma} \label{lem:haar-integral-central-in-crossed-product} 
Let $H$ and $A$ be Hopf algebras, over $\CC$, with $A$ finite dimensional and semisimple. Suppose that $A$ is a $H$-module coalgebra. Then the Haar integral $\Lambda$ of $A$ is $H$-invariant, namely:
\[h \lact \Lambda = \eps(h) \Lambda; \textrm{ for all } h \in H.
\]

\end{lemma}
\begin{proof}
{As recalled in Theorem \ref{prop:haar-integral0},}  the Haar integral in $A$ can be expressed by the regular character $\chi_{A^*} \in (A^*)^*$ of the dual Hopf algebra $A^*$, i.e.\ $\Lambda = \frac{1}{\dim(A)} \chi_{A^*}$, where the canonical identification $(A^*)^* \cong A$ is implicit.
The regular character of the dual Hopf algebra $A^*$ has the following expression as an element of $A$, in terms of the co-multiplication of $A$:
\[ \chi_{A^*} = \sum_{i=1}^{\dim(A)} (e_i)_{(1)} e^i\big((e_i)_{(2)}\big), \]
where $(e_i)_{i=1}^{\dim(A)}$ and $(e^i)_{i=1}^{\dim(A)}$ are dual bases for $A$ and $A^*$, respectively.
Hence, we compute, omitting the summation symbol and letting $h \in H$:
\begin{align*}
h \lact \Lambda &= h \lact \left( (e_i)_{(1)} \, e^i((e_i)_{(2)}) \right) \\
&= (h \lact (e_i)_{(1)})\, e^i((e_i)_{(2)}) \\
&= (h_{(1)} \lact (e_i)_{(1)}) \, e^i\big (S^{-1}(h_{(3)}) \lact h_{(2)} \lact (e_i)_{(2)}\big) \\
%&= (h_{(1)} \lact (e_i)_{(1)}) e^i(S^{-1}(h_{(3)}) \lact h_{(2)} \lact (e_i)_{(2)}) \\
&\stackrel{\text{$A$ $H$-mod coalg.}}{=} (h_{(1)} \lact e_i)_{(1)}\, e^i\big(S^{-1}(h_{(2)}) \lact (h_{(1)} \lact e_i)_{(2)}\big) \\
&\stackrel{\text{(*)}}{=} (e_i)_{(1)} \, e^i\big(h_{(1)} \lact S^{-1}(h_{(2)}) \lact (e_i)_{(2)}\big) \\
&\stackrel{S^{-1}=S}{=} \eps(h)\, (e_i)_{(1)} e^i((e_i)_{(2)}) \\
&= \eps(h) \,  \Lambda.
\end{align*}
Here in the fourth step (*) we have used the elementary property of dual bases which states that $\sum_i f(e_i) \ot e^i = \sum_i e_i \ot (e^i\circ f)$, for any linear map $f \colon A \to A$ (which can be verified by evaluating both sides of the equation on the basis). Hence we have, for all $g \in H$, and where $L_g\colon A \to A$, $x \mapsto gx$.
\begin{align*}
  e^i \circ L_g   \otimes (e_i)_{(1)} \otimes (e_i)_{(2)}= 
  e^i   \otimes (g \lact e_i)_{(1)} \otimes (g\lact e_i)_{(2)},
\end{align*}
and so, for $h \in H$:
\begin{align*}
  e^i \circ L_{h_{(1)}}   \otimes (e_i)_{(1)} \otimes S^{-1}(h_{(2)})\lact (e_i)_{(2)}= 
  e^i   \otimes (h_{(1)} \lact e_i)_{(1)} \otimes S^{-1}(h_{(2)}) \lact ( h_{(1)} \lact e_i)_{(2)}.
\end{align*}\end{proof}

\begin{example} In the case of a crossed module of Hopf algebras 
$(\CC E\xrightarrow{\partial} \CC G, \lact)$ derived from a crossed module of finite groups $(E\xrightarrow{\partial}G,\lact)$, see Example \ref{ex:gr}, the compatibility relation between the Haar integral in $\CC E$ and the action of  $\CC G$ simply means that:
\begin{align*}
    g \lact \left ( \frac{1}{|E|}\sum_{e \in E} e\right)= \frac{1}{|E|}\sum_{e \in E} e, \quad \text{ for all } g \in G.
\end{align*}
\end{example}

\subsection{A class of examples of crossed modules of semisimple Hopf algebras from crossed products of groups algebras with dual group algebras}\label{sec:main_example}

Recall, from examples \ref{ex:CG} and \ref{ex:FG}, our conventions for how given a finite group $G$ we can form the Hopf algebras $\CC G$ and $\Fun{G}$.  
The following result is proved in Subsection \ref{Proof:GEXY}, in the Appendix.

\begin{proposition} \label{pr:GEXY}
Let $(E \xrightarrow{\partial} G, \lact)$ be a crossed module of groups.
Let $X$ and $Y$ be finite groups on which $G$ acts by automorphisms. Let $f \colon Y \to X$ be a $G$-equivariant group morphism.
Suppose that the restrictions of the actions of $G$ on $Y$ and $X$, respectively, to $\im(\partial) \subseteq G$ each are trivial.

The action below  gives $\Fun{Y}$ a $\CC G$-module algebra and coalgebra structure,  $$(g \lact \varphi)(y) = \varphi(g^{-1}\lact y), \textrm{  where  $g\in G$, $\varphi \in \Fun{Y}$ and $y \in Y$,  } $$  and the Yetter-Drinfeld condition \eqref{eq:yetter-drinfeld-condition-trivial} holds, so we can form $\Fun{Y} \rtimes \CC G$ as in Definition \ref{def:crossed-product-algebra}.

Furthermore, we have a crossed module of Hopf algebras, $(\Fun{X} \ot \CC E \xrightarrow{\partial} \Fun{Y} \rtimes \CC G, \lact)$,
where\begin{align*}
\partial \colon \Fun{X} \otimes \CC E &\lto \Fun{Y} \rtimes \CC G, \\
\xi \ot e &\lmapsto f^* \xi \ot \partial(e),
\end{align*}
{(here $(f^* \xi)(y)=\xi(f(y))$, if $y \in Y$)} and
\[ (\varphi \ot g) \lact (\xi \ot e) := \varphi(1) (g \lact \xi) \ot (g \lact e), \]
for $\varphi \ot g \in \Fun{Y} \ot \CC G$, $\xi \ot e \in \Fun{X} \ot \CC E$, where $(g \lact \xi)(x) := \xi(g^{-1} \lact x)$,
for $x \in X$.
\end{proposition}

\begin{lemma} \label{lem:haar-integrals-of-example}
The Hopf algebras $\Fun{X} \ot \CC E$ and $\Fun{Y} \rtimes \CC G$ are semisimple with Haar integrals:
\[ \delta_{1_X} \ot \Big( \frac{1}{\vert E \vert} \sum_{e \in E} e \Big) \in \Fun{X}\otimes\CC E, \qquad \textrm{ and }\qquad
 \delta_{1_Y} \ot \Big( \frac{1}{\vert G \vert} \sum_{g \in G} g \Big) \in \Fun{Y}\rtimes\CC G . \]
\end{lemma}
\begin{proof}{It is clear that both Hopf algebras are finite dimensional and involutory, hence semisimple, by Theorem~\ref{prop:haar-integral0}.}
We verify the defining properties of the Haar integral in Definition \ref{def:haar-integral}.
For $\Fun{X} \ot \CC E$ this is straightforward, since it is just a tensor product of two Hopf algebras and $\delta_{1_X} \in \Fun{X}$ and $\frac{1}{\vert E \vert} \sum_{e \in E} e\in \CC E$ are the Haar integrals of the two Hopf algebras.

To show the idempotence of the Haar integral for  $\Fun{Y}\rtimes\CC G$ we calculate:
\begin{align*}
\Big( \delta_{1_Y} \ot \Big(  \frac{1}{\vert G \vert} \sum_{g \in G} g \Big) \Big) \cdot_\rtimes \Big( \delta_{1_Y} \ot \Big(  \frac{1}{\vert G \vert} \sum_{g \in G} g \Big) \Big) &= \frac{1}{\vert G \vert^2} \sum_{g,h \in G} \delta_{1_Y} (g \lact \delta_{1_Y}) \ot gh \\
&= \delta_{1_Y} \delta_{1_Y} \ot \frac{1}{\vert G \vert^2} \sum_{g,h \in G} gh
%&= \delta_{1_Y} \ot \frac{1}{\vert G \vert^2} \sum_{g,h \in G} g \\
= \delta_{1_Y} \ot \frac{1}{\vert G \vert} \sum_{g \in G} g .
\end{align*}

Let us now show the remaining left and right invariance properties of the Haar integral.
Let $y \in Y$ and $h\in G$.
Then we have, on the one hand,
\begin{align*}
(\delta_y \ot h) \cdot_\rtimes \Big( \delta_{1_Y} \ot \Big( \frac{1}{\vert G \vert} \sum_{g \in G} g \Big) \Big) &= \delta_y (h \lact \delta_{1_Y}) \ot \Big( \frac{1}{\vert G \vert} \sum_{g \in G} h g \Big) = \delta_y \delta_{1_Y} \ot \Big( \frac{1}{\vert G \vert} \sum_{g \in G} g \Big) \\
&= \delta_y(1_Y)  \Big(\delta_{1_Y} \ot \frac{1}{\vert G \vert} \sum_{g \in G} g \Big)= \eps (\delta_y \otimes h)\Big( \delta_{1_Y} \ot \Big ( \frac{1}{\vert G \vert} \sum_{g \in G} g \Big)\Big),
\end{align*}
and on the other hand,
\begin{align*}
\Big( \delta_{1_Y} \ot \Big( \frac{1}{\vert G \vert} \sum_{g \in G} g \Big) \Big) \cdot_\rtimes (\delta_y \ot h) &= \frac{1}{\vert G \vert} \sum_{g \in G} \delta_{1_Y} \big(  g \lact \delta_y \big) \ot  g h  = \frac{1}{\vert G \vert} \sum_{g \in G} \delta_{1_Y} \delta_{\left(  gh^{-1}\lact y \right)} \ot g  \\
&= \delta_y (1_Y)\Big(  \delta_{1_Y} \ot   \frac{1}{\vert G \vert} \sum_{g \in G} g \Big) 
= \eps (\delta_y \otimes h)\Big( \delta_{1_Y} \ot  \frac{1}{\vert G \vert} \sum_{g \in G} g \Big).
\end{align*}
{(In particular, by Theorem~\ref{prop:haar-integral0}, if follows again that the two Hopf algebras are semisimple.)}
\end{proof}

\section{Hopf-algebraic higher Kitaev model:
Commuting-projector Hamiltonian model for  oriented surfaces {with cell decomposition} from a crossed module of semisimple Hopf algebras} \label{sec:model}

{In \cite{Higher_Kitaev,companion} a 2-group {commuting-projector} Hamiltonian model for 2+1D and 3+1D topological phases was defined. The model, henceforth called \emph{\HKM}, takes as input a 2-group, which can be represented by a crossed module $(E \xrightarrow{\partial} G, \lact)$ of groups Definition \ref{ex:gr}; see \cite{Baez_Lauda}\cite[\S 2.5 \& 2.7]{Brown_Higgins_Sivera}. Given a {manifold $M$, with a triangulation (or more generally a 2-lattice decomposition \cite{companion}) $L$,} the total Hilbert space of the \HKM\ is the free vector space on the set of all fake-flat 2-gauge configurations in {$(M,L)$} \cite[\S 3.2.1]{companion}.  

In 2+1D there is a variant of the \HKM, whose construction was sketched in \cite[page  8]{Higher_Kitaev}. Its total Hilbert space is the free vector space on the set of all 2-gauge configurations in $(M,L)$.
We will call this  latter model the \emph{\sHKM}. We note that in order that the edge projectors in the \sHKM\ commute with each other,  certain restrictions on the 2-lattice decompositions, `adequacy', to be defined later, must be put.  This restriction is however  mild enough for us to be able to work with triangulations of surfaces, with a total order in the set of vertices.

{Crossed modules of Hopf algebras were defined in \cite{majid2012strict}, and there related to a special case of strict quantum 2-groups; see also \cite{FARIAMARTINS2016,emir,fregier2011}.}
In this section, we define a Hopf-algebraic generalization of the 2+1D \sHKM, defined on a surface $\Sigma$ with an adequate cell decomposition.
 The  input datum is a crossed module $(A \xrightarrow{\partial} H, \lact)$ of finite-dimensional semisimple Hopf algebras $A$ and $H$, see Definition \ref{def:triv-braid-hopf-crossed-module}, {which we will fix throughout this section}. 
To the model constructed in this section we  
give the name:  %call 
 %the 
 \emph{\HHK\ model}.

\subsection{Cell decompositions and the total state space of the model} 

\subsubsection{{Conventions for cell decompositions of surfaces}}\label{sec:conv-cell}
Let $\Sigma$ be a compact oriented surface.% (possibly with boundary).
\begin{definition}[Cell decomposition of a surface] \label{def:cell-decomposition}
By a \emph{cell decomposition $L$ of $\Sigma$} we mean the following.
We have an embedded graph $L$ in $\Sigma$ with no loops and no univalent vertices, with finite sets $L^0$ and $L^1$ of \emph{vertices} and \emph{edges}, such that its complement in $\Sigma$ is a disjoint union of open $2$-disks, forming the set $L^2$ of \emph{plaquettes}. We 
will require 
that the closure, in $\Sigma$, of each plaquette  is a topological 2-disk, and that its boundary is a 1-sphere $S^1$,  with a CW-decomposition constructed by vertices and edges of $\Gamma$, with no repetitions of either edges or vertices.  

We moreover require that two different plaquettes have at most one edge in common in their boundary, and that at most one edge can connect a given pair of vertices.

 Furthermore, each edge is oriented and each plaquette ${P}$ has a distinguished vertex {$v_P$} in its boundary, called its \emph{base-point}.
 {In other words $L$ includes a function $v:L^2 \rightarrow L^0$ such that $v_P \in \partial P$ and a function $\omega:L^1 \rightarrow L^0$ such that $\omega_e \in \partial e$.}
 \end{definition}
\noindent{The choice of orientations on edges and of a base-point in each plaquette is essential to establish conventions for edge operators and plaquette operators, and was also considered in \cite{Higher_Kitaev,companion}.} The conditions on boundaries of plaquettes, and that we have at most one edge between a pair of vertices, arise only in order to reduce the number of cases to consider in the proofs we give.
\begin{remark}
In particular, $L$ is in this way a regular CW decomposition of $\Sigma$, which additionally has that the attaching maps of the $2$-cells are cellular, as the $2$-lattices defined in \cite{companion}.
\end{remark}
\begin{remark}
The cell decompositions considered here in particular include triangulations of surfaces,  {which in this paper we regard to be} equipped with an orientation of each edge and a choice of a vertex (base-point) for each triangle, {unless explicitly stated otherwise}. Here triangulations, as usual, come with the restriction that two simplices can only intersect along a common face.
\end{remark}

We will initially define our model for any cell decomposition of an oriented surface $\Sigma$. This model has the best properties when we restrict to, here called, \emph{adequate} cell decompositions, which we will now define.

Let $P$ be a plaquette, with base-point $v_P$. Its boundary is topologically $S^1$, and has a
preferred orientation, anti-clockwise, or positive, induced by the orientation that $\Sigma$ induces in the plaquette. Let $(e_1^P, \dots, e_{n_{P}}^P)$ be the ordered set of edges in the boundary of $P$, starting at the base-point $v_P$ and going around the boundary in an anti-clockwise way, with respect to the plaquette.

\begin{definition}[Adequate cell decomposition]\label{def:ad-cell-decomposition} Let $L$ be a cell decomposition of $\Sigma$. Let $P$ be a plaquette. Let $e$ and $e'$ be different edges in the boundary of $P$. We say that the pair $(e,e')$ is \emph{adequate} (with respect to $P$) if one of the following conditions holds:
\begin{enumerate}\setlength\itemsep{0em}
    \item either $e$ and $e'$ are both oriented clockwise or anti-clockwise, with respect to $P$,
    \item or, if $e$ and $e'$ have opposite orientations, the edge which is oriented clockwise appears later than than the one oriented 
    anti-clockwise in the list $(e_1^P, \dots, e_{n_{P}}^P)$.
\end{enumerate}
What this means in practice, should be clear from Figure \ref{fig:adequate}.

We say that $L$ is adequate if given any plaquette, $P$, all pairs, $(e,e')$ of different edges in the boundary of $P$ are adequate with respect to $P$.

\begin{figure}[ht!]
 \labellist
\pinlabel $\small{v_P}$ at 490 10
\pinlabel $\small{e}$ at 453 340
\pinlabel $\small{e'}$ at 109 340
\pinlabel $\textrm{adequate}$ at 300 -60
\pinlabel $\small{v_P}$ at 1283 10
\pinlabel $\small{e}$ at 1246 340
\pinlabel $\small{e'}$ at 902 340
\pinlabel $\textrm{adequate}$ at 1069 -60
\pinlabel $\small{v_P}$ at 2076 10
\pinlabel $\small{e}$ at 2039 340
\pinlabel $\small{e'}$ at 1695 340
\pinlabel $\textrm{adequate}$ at 1886 -60
\pinlabel $\small{v_P}$ at 2875 10
\pinlabel $\small{a}$ at 2832 340
\pinlabel $\small{a'}$ at 2488 340
\pinlabel $\textrm{non-adequate}$ at 2690 -60
\pinlabel $\small{P}$ at 285  150
\pinlabel $\small{P}$ at 1078  150
\pinlabel $\small{P}$ at 1871  150
\pinlabel $\small{P}$ at 2664  150
\endlabellist
\centering
\includegraphics[scale=0.17]{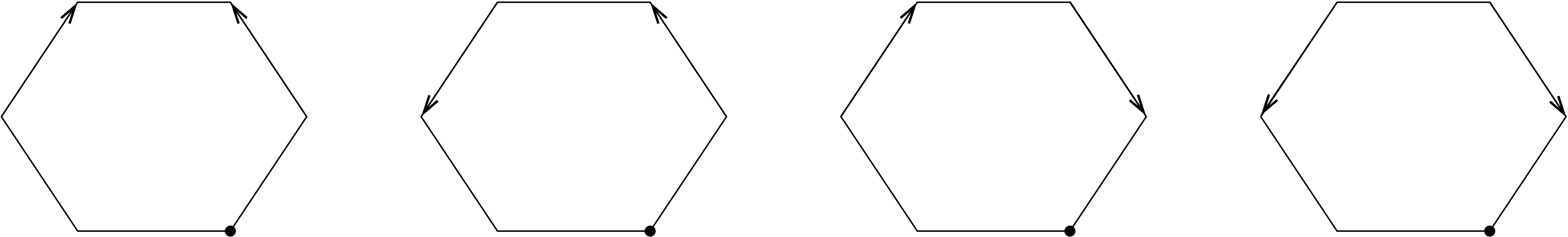}
\vspace{5px}
\caption{{Adequate and non-adequate pairs of edges on a plaquette $P$.}}
\label{fig:adequate}
\end{figure}
\end{definition}
\vspace{-15px}
\noindent Any oriented surface can be given an adequate cell decomposition, given the following result.
\begin{lemma}\label{lem:adequatefromtriang}
Let $\Sigma$ be a compact oriented surface as before.
Let $L$ be a triangulation of $\Sigma$. Suppose that we put a total order on the set of vertices of $L$. We can then base each plaquette by its `minimal' vertex, and we can also orient each edge $(a,b)$ in increasing order. This cell decomposition is adequate.
\end{lemma}
\begin{proof}
This follows from a case-by-case analysis (two cases only). Consider a triangle $P$, with vertices $a$, $b$ and $c$, with $a<b<c$. This can be embedded in the plane, so that the orientation of $P$ coincides with the orientation in its image induced by the plane. 
There are two possibilities for how the embedding looks like, up to orientation preserving diffeomorphism of the plane, as shown below,
$$\vcenter{\xymatrix@C=15pt@R=15pt{&& b\ar[dr]\\
&a\ar[ur]\ar[rr] && c } } \qquad \textrm{ and } \vcenter{\xymatrix@C=15pt@R=15pt{&& c\ar@{<-}[dr]\\
&a\ar[ur]\ar[rr] && b }} $$
Each clearly only gives adequate pairs of edges.
\end{proof}

\subsubsection{{The total vector space of the \HHK\ model}}

{Let $(A \xrightarrow{\partial} H, \lact)$ be a crossed module of finite-dimensional semisimple Hopf algebras.} Given an oriented surface $\Sigma$ with a cell decomposition $L$, we define the \emph{total state space} of the \HHK\ model as:
\begin{equation} \label{eq:total-state-space}
\mathcal{H}_L 
:= H^{\ot L^1} \ot A^{\ot L^2} := \bigotimes_{e \in L^1} H \ot \bigotimes_{P \in L^2} A \ .
\end{equation}
On the vector space 
$\mathcal{H}_L  $
we will define three families of linear endomorphisms: vertex operators, edge operators and plaquette operators.

By the \emph{support} of an endomorphism $T$ of $\mathcal{H}_L = H^{\ot L^1} \ot A^{\ot L^2}$, we mean the largest subset, {$M = M^1 \sqcup M^2 \subseteq L^1 \sqcup L^2$, such that there exists an endomorphism $T_S$ of $\mathcal{H}_M := H^{\ot M^1} \ot A^{\ot M^2}$ such that $T = T_M \ot \id{\mathcal{H}_L \setminus \mathcal{H}_M}$.}

We will occasionally write $H_e$ or $A_P$, where $e \in L^1$ and $P \in L^2$, for copies of the Hopf algebras $H$ and $A$ seen as the tensor factors in $\mathcal{H}_L$ associated with the edge $e$ or the plaquette $P$, respectively.

\begin{remark}
As in \cite{buerschaper-et-al}, one could consider $C^*$-Hopf algebras $H$ and $A$ and this would turn $\mathcal{H}_L$ into a Hilbert space and make the projectors, and hence the Hamiltonian, Hermitian.
\end{remark}

\subsection{Vertex operators
{of the model on $(A \xrightarrow{\partial} H, \lact)$}
}\label{sec:vertex_ops}

{Fix a pair $(\Sigma,L)$, as in Definition \ref{def:cell-decomposition}.}
Let $v \in L^0$ be a vertex and let $P \in L^2$ be an adjacent plaquette.
Such a pair $(v,P)$ is also called a \emph{site} {\cite{Kitaev}\cite[Definition 2.1]{balsam-kirillov} \cite[\S 3.1]{buerschaper-et-al}.} {Note that $v$ is not necessarily the base-point of $P$.}

We define now a family of vertex  operators, $V_{v,P}^h, h \in H$, acting on $\mathcal{H}_L$, whose support is the set of the edges and plaquettes incident to the vertex $v$.
Here $P$ is a plaquette adjacent to $v$, which is used to fix the convention for the vertex operator \cite[Definition 2.2]{balsam-kirillov}.
{The definition of the vertex operators in our model follows that of the vertex operators in the \HAKM\ \cite{buerschaper-et-al}, and reduces to the latter when $A = \CC$ {(up to minor differences in conventions)}.
At the same time, the vertex operators here recover the vertex gauge spikes in the \sHKM\ defined in \cite{companion} for a Hopf crossed module $(A \xrightarrow{\partial} H, \lact) = (\CC E \xrightarrow{\partial} \CC G, \lact)$ induced by a crossed module of groups $(E \xrightarrow{\partial} G, \lact)$.}

Consider the set of edges incident to (that is, ending or starting at) the vertex $v$.
The orientation of the surface $\Sigma$ induces a cyclic order on this set, by going around the vertex in the counterclockwise order with respect to the orientation of $\Sigma$.
The choice of plaquette $P$ incident to $v$ furthermore lifts this cyclic order to a linear one by starting with the edge that comes right after $P$ in the counterclockwise order around $v$.
We write the edges in this specified linear order as $(e_1, \dots, e_n)$.
These edges might be oriented away from or towards $v$.
Define $\theta_i := +1$ in the former case and $\theta_i := -1$ in the latter, for $i = 1, \dots, n$. These conventions should be clear from Figure \ref{fig:vertex_op}.

There are two canonical left actions of $H$ on itself: By left multiplication, $h \ot x \lmapsto hx$ for any $h,x \in H$, on the one hand, and by right multiplication pulled back along the antipode, $h \ot x \lmapsto x S(h)$, on the other hand; {see \cite[\S 3.1]{buerschaper-et-al}}.
In order to treat both cases at once, we use the notation \begin{equation}\label{eq: <>}h^\lrangle{\theta} = \begin{cases}
h^\lrangle{+1} := h, &\text{ if } \theta = +1, \\
h^\lrangle{-1} := S(h), &\text{ if } \theta = -1,
\end{cases}
\end{equation}
which is justified, since $S$, being involutive {(Proposition \ref{prop:haar-integral})}, defines an action of the group $\{ +1,-1 \} \cong \ZZ_2$. {In some cases this notation will be combined with the Sweedler notation, where it is understood that for instance $h_{(1)}^{\langle +1 \rangle} \otimes h_{(2)}^{\langle -1\rangle}$ means $h_{(1)} \otimes S(h_{(2)})$, and $h_{(1)}^{\langle-1\rangle} \otimes h_{(2)}^{\langle -1 \rangle}$ means $S(h_{(1)}) \otimes S(h_{(2)})$, so we always apply comultiplication before applying $S$.}

Furthermore, consider the (possibly empty) set of all plaquettes $P \in L^2$ whose base-point is $v$ and choose any order on this set, say $(P_1, \dots, P_k)$.
The definition of the vertex operator will be independent of this choice, {from Lemma \ref{lem:yetter-drinfeld-for-many-factors}.}
\begin{definition}[Vertex operator] \label{def:vertex-operator}
For any $h \in H$, the \emph{vertex operator, $V_{v,P}^h$, based at  site $(v,P)$}, is
\begin{align*}
V_{v,P}^h : \mathcal{H}_L &\lto \mathcal{H}_L, \\
v_{e_1} \ot\cdots\ot v_{e_n} \ot X_{P_1} \ot\cdots\ot X_{P_k} &\lmapsto \big(h_{(1)} v_{e_1}^\lrangle{\theta_1}\big)^\lrangle{\theta_1} \ot\cdots\ot \big(h_{(n)} v_{e_n}^\lrangle{\theta_n}\big)^\lrangle{\theta_n} \\
&\phantom{\lmapsto----} \ot h_{(n+1)}\lact X_{P_1} \ot\cdots\ot h_{(n+k)}\lact X_{P_k},
\end{align*}
where $v_{e_1} \ot\cdots\ot v_{e_n} \ot X_{P_1} \ot\cdots\ot X_{P_k} \in H_{e_1} \ot\cdots\ot H_{e_n} \ot A_{P_1} \ot\cdots\ot A_{P_k} = H^{\ot n}\ot A^{\ot k}$ are in the tensor factors of $\mathcal{H}_L$ associated with the edges $(e_1,\dots,e_n)$ and plaquettes $(P_1,\dots,P_k)$.
It is implicitly understood that $V_{v,P}^h$ acts as the identity on all remaining tensor factors.
\end{definition}
\begin{figure}[ht!]
 \labellist
\pinlabel $\small{v}$ at 410 171
\pinlabel $\small{P_2}$ at 125 195
\pinlabel $\small{P_1}$ at 650 185
\pinlabel $\small{P}$ at 360 390
\pinlabel $\small{e_4}$ at 460 301
\pinlabel $\small{e_1}$ at 296 326
\pinlabel $\small{e_2}$ at 256 120
\pinlabel $\small{e_3}$ at 491 110
\endlabellist
\centering
\includegraphics[scale=0.20]{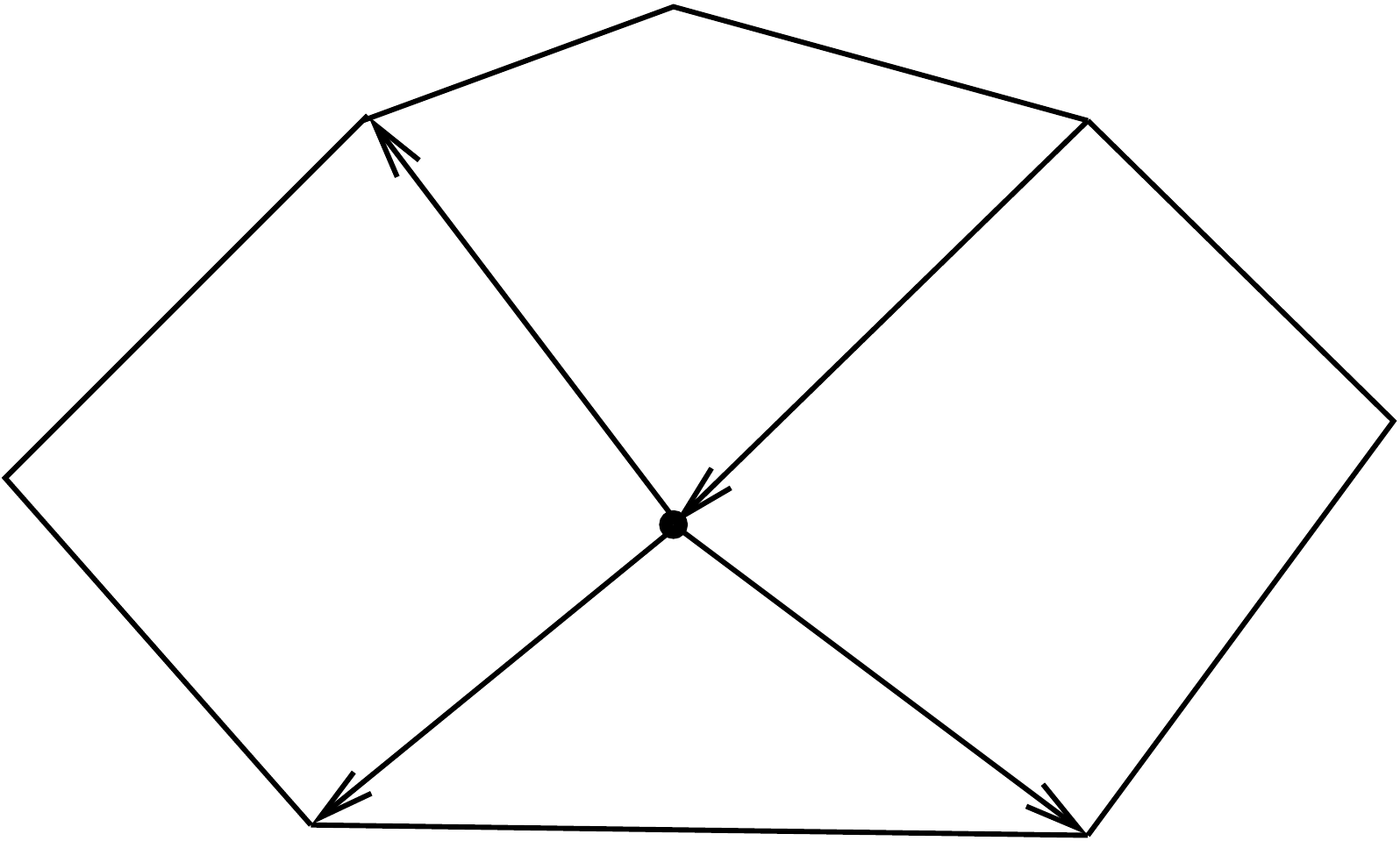}
\caption{Example of conventions for vertex operator  $V_{v,P}^h$. Here the plaquettes $P_1$ and $P_2$ each are based in $v$. Furthermore $n=4$, $k=2$, and the sequence $(\theta_1,\theta_2,\theta_3,\theta_4)$ is $(1,1,1, -1)$.}
\label{fig:vertex_op}
\end{figure}

In the case of the configuration in Figure \ref{fig:vertex_op}, the vertex operator, $V_{v,P}^h$, takes the form below,
\begin{align*}
v_{e_1} \ot v_{e_2} \ot v_{e_3} & \ot v_{e_4} \ot X_{P_1} \ot X_{P_2} \\& \stackrel{V_{v,P}^h}{\longmapsto} 
h_{(1)}v_{e_1}  \ot h_{(2)} v_{e_2} \ot h_{(3)} v_{e_3} \ot  v_{e_4} S(h_{(4)}) \ot  h_{(5)}\lact X_{P_1} \ot h_{(6)}\lact  X_{P_2}.
\end{align*}
\begin{remark}By Lemma \ref{lem:yetter-drinfeld-for-many-factors}, the definition of the vertex operator $V_{v,P}^h$ is independent of the order chosen on the set $\{P_1,\dots,P_k\}$.
 The vertex operator $V_{v,P}^h$ does however in general depend  on the plaquette $P$.
\end{remark}

\begin{lemma}\label{lem:vertex-operators-representation} Given any vertex $v \in L^0$ and adjacent plaquette $P\in L^2$, the endomorphisms $( V_{v,P}^h : \mathcal{H}_L  \lto \mathcal{H}_L )_{h \in H}$ define a representation of $H$
 on $\mathcal{H}_L$, that is:
\[ V_{v,P}^h \circ V_{v,P}^{h'} = V_{v,P}^{h h'} \quad \text{ for all } h, h' \in H . \]
\end{lemma}
\begin{proof}
This follows from the fact that $\Delta^{(n+k-1)}: H \to H^{\otimes (n+k)} $ is an algebra map, and the fact that the antipode is an anti-homomorphism of algebras. \end{proof}

Owing to the fact that a semisimple Hopf algebra $H$ comes with a distinguished idempotent, the Haar integral $\ell \in H$, as reviewed in Subsection \ref{Sec:semi}, we can define the following projectors on $\mathcal{H}_L$, which will enter the definition of the Hamiltonian of the exactly solvable model.

{(Cf. \cite[\S 2.4]{balsam-kirillov} and \cite{buerschaper-et-al} for the similar case of vertex projectors on the \HAKM.)} {Recall by Proposition  \ref{prop:haar-integral}} that the Haar integral $\ell \in H$ is cocommutative. {By combining Lemma \ref{lem:coproduct-cocommut-el} with Lemma \ref{lem:cocom-yetter-drinfeld-for-many-factors}, it follows that} the operator $V_{v,P}^\ell$ depends only on the cyclic order of the set $(e_1,\dots,e_n)$ of edges incident to the vertex $v$, and is hence independent of the choice of adjacent plaquette $P$.

\begin{definition}[Vertex projector]\label{def:vert_proj}
Let $\ell \in H$ be the Haar integral of $H$.
Then we define the \emph{vertex projector, {$V_v\colon \mathcal{H}_L \lto \mathcal{H}_L$}, based at the vertex $v$} as below, choosing any plaquette $P$ adjacent to $v$, \[ V_v := V_{v,P}^\ell . \]
\noindent (Note that a cyclic order on the set  $(e_1,\dots,e_n)$ is given by the orientation of the surface.)
\end{definition}
Given that $\ell^2=\ell$, it follows from Lemma \ref{lem:vertex-operators-representation} that indeed $V_v =V^\ell_{v,P}$ is a projector.
\begin{lemma} \label{lem:vertex-operators-commute}
Let $v_1$ and $v_2 \in L^0$ be two distinct vertices.
Then the corresponding vertex operators commute with each other, that is:
\[ V_{v_1,P_1}^{h_1} \circ V_{v_2,P_2}^{h_2} = V_{v_2,P_2}^{h_2} \circ V_{v_1,P_1}^{h_1} \quad \text{ for all } h_1, h_2 \in H , \]
where $P_1$ and $P_2 \in L^2$ are any plaquettes adjacent to $v_1$ and $v_2$, respectively.
\end{lemma}

\begin{proof}
The vertex operators only have intersecting supports if $v_1$ and $v_2$ are the end-points of a common edge $e \in L^1$.
In this case, the intersection of their support is only the copy of $H$ that is associated with $e$.
On this tensor factor the two vertex operators in question act via left or right multiplication, depending on whether the edge $e$ is directed away from or towards a given vertex.
Since this differs for the two vertices, one of the operators acts via left multiplication and the other via right multiplication, the vertex operator commute with each other.
\end{proof}

 {In particular, putting $h_1 = h_2=\ell \in H$, and using this last lemma and the previous one, vertex projectors commute: $[V_{v_1},V_{v_2}]=0,$ for all $v_1,v_2 \in L^0.$}

\subsection{Edge operators
{of the model on $(A \xrightarrow{\partial} H, \lact)$}
}

Fix a pair $(\Sigma,L)$, as in Definition \ref{def:cell-decomposition}.
Let $e \in L^1$ be an edge.
With respect to the orientations of the edge $e$ and the surface $\Sigma$, we can define the \emph{plaquette $P \in L^2$ on the left side of the edge $e$} and the \emph{plaquette $Q \in L^2$ on the right side of $e$, {as we go from the starting vertex 
$\omega_e$ 
to the target vertex of $e$.}}

\begin{figure}[ht!]
 \labellist
\pinlabel $\small{v_P}$ at 40 91
\pinlabel $\small{v_Q}$ at 794 302
\pinlabel $\small{e}$ at 495 191
\pinlabel $\small{P}$ at 231 174
\pinlabel $\small{Q}$ at 661 182
\pinlabel $\small{v}$ at 495 122
\pinlabel $\small{e_1}$ at 91 58
\pinlabel $\small{e_2}$ at 224 30
\pinlabel $\small{e_3}$ at 388 76
\pinlabel $\small{d_1}$ at 800 220
\pinlabel $\small{d_2}$ at 825 99
\pinlabel $\small{d_3}$ at 750 58
\pinlabel $\small{d_4}$ at 555 86
\endlabellist
\centering
\includegraphics[scale=0.25]{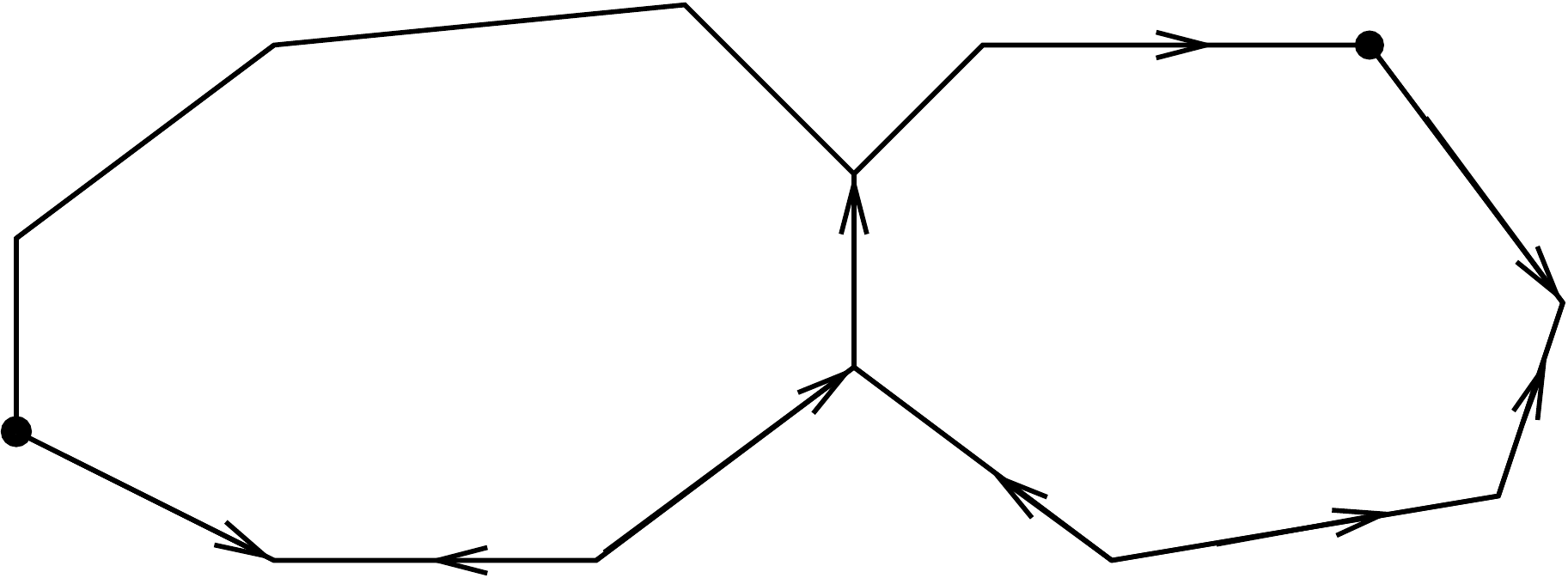}
\caption{{Conventions 
for 
edge operators $E^a_e$. Here $l=3$ and $(\theta_1,\theta_2,\theta_3)=(1,-1,1)$, whereas $r=4$ and $(\sigma_1,\sigma_2,\sigma_3,\sigma_4)=(1,-1,-1,1)$.}}
\label{fig:edge}
\end{figure}
We define now a family of edge operators acting on $\mathcal{H}_L$, whose support {consists of} the plaquettes $P$ and $Q$ and the set of the edges in the boundary of $P$ and $Q$.
These operators do not appear in the \HAKM.
For a Hopf crossed module $(A \xrightarrow{\partial} H, \lact) = (\CC E \xrightarrow{\partial} \CC G, \lact)$ induced by a crossed module of groups $(E \xrightarrow{\partial} G, \lact)$
they recover the edge gauge transformations in the \sHKM\ \cite{Higher_Kitaev} and in its fake-flat subspace the gauge spikes in the \HKM\ of \cite{companion}.

Denote by $v_P, v_Q \in L^0$ the base-points of $P$ and $Q$, respectively, and let $v \in L^0$ be the starting vertex of the edge $e$, i.e.\ the vertex from which $e$ points away.
Let $(e_1,\dots, e_\ell)$ be the (possibly empty) ordered set of edges connecting $v_P$ with $v$ in the boundary of $P$, starting from $v_P$ and going in counterclockwise direction (with respect to the orientation of $\Sigma$) around $P$. {(Another way to see this sequence of edges is as the unique path from $v_P$ to $v$ that does not pass through $e.$)}
The set is empty precisely when $v_P = v$.
For $i = 1,\dots,\ell$, let $\theta_i := +1$ if $e_i$ is oriented in counterclockwise direction around $P$, and let $\theta_i := -1$ otherwise. {These conventions should be clarified by Figure \ref{fig:edge}.}

Similarly, let $(d_1,\dots, d_r)$ be the (possibly empty) ordered set of edges connecting $v_Q$ with $v$ in the boundary of $Q$, starting from $v_Q$ and going in clockwise direction (with respect to the orientation of $\Sigma$) around $Q$.
The set is empty precisely when $v_Q = v$.
For $j = 1,\dots,r$, let $\sigma_j := +1$ if $d_j$ is oriented in clockwise direction around $Q$, and let $\sigma_j := -1$ otherwise.

\begin{definition}[Edge operator] \label{def:edge-operator}
For any $a \in A$, the \emph{edge operator $E_e^a$ on $\mathcal{H}_L$, based at the edge $e$}, is
\begin{align*}
E_e^a : \mathcal{H}_L &\lto \mathcal{H}_L, \\
v_e \ot v_{e_1} \ot\cdots\ot v_{e_\ell} \ot v_{d_1} \ot\cdots\ot v_{d_r} \ot X_P \ot X_Q &\lmapsto \partial a_{(3)} v_e \\
&\phantom{\lmapsto} \ot (v_{e_1})_{(1)} \ot\cdots\ot (v_{e_\ell})_{(1)} \\
&\phantom{\lmapsto} \ot (v_{d_1})_{(2)} \ot\cdots\ot (v_{d_r})_{(2)} \\
&\phantom{\lmapsto} \ot \big( (v_{e_1})_{(2)}^\lrangle{\theta_1} \cdots (v_{e_\ell})_{(2)}^\lrangle{\theta_\ell} \lact a_{(1)} \big) X_P \\
&\phantom{\lmapsto} \ot X_Q \big( (v_{d_1})_{(1)}^\lrangle{\sigma_1} \cdots (v_{d_r})_{(1)}^\lrangle{\sigma_r} \lact S ( a_{(2)})\big),
\end{align*}
where we use the notation {defined in \eqref{eq: <>}, and where}
\begin{align*}
&\phantom{xxx}v_e \ot v_{e_1} \ot\cdots\ot v_{e_\ell} \ot v_{d_1} \ot\cdots\ot v_{d_r} \ot X_P \ot X_Q \\
&\phantom{-----------}\in H_e \ot H_{e_1} \ot\cdots\ot H_{e_\ell} \ot H_{d_1} \ot\cdots\ot H_{d_r} \ot A_P \ot A_Q \\
&\phantom{-----------}= H^{\ot 1+\ell+r}\ot A^{\ot 2}
\end{align*}
are in the tensor factors of $\mathcal{H}_L$ associated with the edges $e$, $(e_1,\dots,e_\ell)$ and $(d_1,\dots, d_r)$, and plaquettes $P$ and $Q$, respectively.
It is implicitly understood that the edge operator is defined to act as the identity on all remaining tensor factors in $\mathcal{H}_L$.

\end{definition}

\begin{example}\label{eg:edge}
In the case of the configuration in Figure \ref{fig:edge}
the edge operator, $E^a_e$, is:
\begin{align*}
v_e &\ot \left (v_{e_1} \ot v_{e_2}\ot v_{e_3}\right)\ot \left (v_{d_1} \ot v_{d_2} \ot v_{d_3}\ot v_{d_4} \right )\ot \left( X_P \ot X_Q \right) \\
&\stackrel{E^a_e}{\longmapsto} \partial a_{(3)} v_e 
\ot\left ( (v_{e_1})_{(1)} \ot (v_{e_2})_{(1)} \ot (v_{e_3})_{(1)} \right)\ot \left ((v_{d_1})_{(2)} \ot (v_{d_2})_{(2)} \ot (v_{d_3})_{(2)} \ot (v_{d_4})_{(2)} \right) 
\\ & \qquad \ot  \left( \left ( (v_{e_1})_{(2)} \, S\big( (v_{e_2})_{(2)}\big) (v_{e_3})_{(2)}\right)\lact a_{(1)} \right) X_P \\ & \qquad  \ot  X_Q \left( (v_{d_1})_{(1)} S\big( (v_{d_2})_{(1)}\big )  S\big( (v_{d_3})_{(1)}\big ) (v_{d_4})_{(1)}\right)  \lact S ( a_{(2)}).
\end{align*}
\end{example}
We also note that, due to the Yetter-Drinfeld condition \eqref{eq:yetter-drinfeld-condition-trivial}, and using that $S$ is an anti-coalgebra-map, {with $S^2=\id{H}$}, 
     we can substitute any term $(v_{e_i})_{(1)}\otimes (v_{e_i})_{(2)}^\lrangle{\theta_i}$ by $(v_{e_i})_{(2)}\otimes (v_{e_i})_{(1)}^\lrangle{\theta_i}$ in the formula for the edge operator, and the same for the $v_{d_i}$.
    To see this note that if $v \in H$ and $a \in A$ then:
\begin{equation}\label{eq:compatibility<>}\begin{split}
 v_{(1)} \otimes S(v_{(2)}) \lact a &= S(S(v_{(1)})) \otimes S(v_{(2)}) \lact a=
  S((Sv)_{(2)}) \otimes S(v)_{(1)} \lact a\\&\stackrel{\text{Y.-D.}}{=} 
   S((Sv)_{(1)}) \otimes S(v)_{(2)} \lact a=
   S(S(v_{(2)})) \otimes S(v_{(1)}) \lact a\\&= v_{(2)} \otimes S(v_{(1)}) \lact a .
  \end{split}
    \end{equation}

As in the case of vertex operators we have:
\begin{lemma}\label{lem:edge-ops-is-a-rep}
Given an edge $e \in L^1$, the edge operators $E^a_e\colon \mathcal{H}_L \to \mathcal{H}_L$ form a representation of $A$ on $\mathcal{H}_L$. I.e., given any $a,b \in A$ we have:
$$E^b_e \circ E^a_e=E^{ba}_e.$$
\end{lemma}
In the proof below, and others, in the case of iterated use of Sweedler notation, we will e.g. write: $v_{(1,2)}$ instead of $(v_{(1)})_{(2)}$.
\begin{proof} (Sketch.)
The proof is quite clear from the explicit form of the edge operators, together with the fact that $A$ is a $H$-module algebra.
It is crucial to note that, applying \eqref{eq:compatibility<>}, given $v \in H$ and $a,b \in A$, we have:
\begin{equation}\label{calc: s act}
    S\big (v_{(1)} \big )\lact a \ot S\big (v_{(2)} \big )\lact b
    =\big (S( v)\big)_{(1)} \lact a \ot \big (S( v)\big)_{(2)} \lact b.
\end{equation}
As an example, which shows the type of calculations required, let us prove we have a representation in the case of the configuration in Figure \ref{fig:edge}, and (in order to simplify the calculations) where the base-point of $P$ is moved to $v$.
Applying $E^a_e$ followed by $E^b_e$ yields:
\begin{align*}
v_e \ot &\left (v_{d_1} \ot v_{d_2} \ot v_{d_3}\ot v_{d_4} \right )\ot \left( X_P \ot X_Q \right) \\
&\stackrel{E^a_e}{\lmapsto} \partial a_{(3)} v_e \ot 
 \left ((v_{d_1})_{(2)} \ot (v_{d_2})_{(2)} \ot (v_{d_3})_{(2)} \ot (v_{d_4})_{(2)} \right) 
\ot    a_{(1)}  X_P \\ & \qquad  \ot  X_Q \left( (v_{d_1})_{(1)} S\big( (v_{d_2})_{(1)}\big )  S\big( (v_{d_3})_{(1)}\big ) (v_{d_4})_{(1)}\right)  \lact S ( a_{(2)})\\
&\stackrel{E^b_e}{\lmapsto} 
\partial (b_{(3)} a_{(3)}) v_e \ot 
 \left ((v_{d_1})_{(2,2)} \ot (v_{d_2})_{(2,2)} \ot (v_{d_3})_{(2,2)} \ot (v_{d_4})_{(2,2)} \right) 
\ot  b_{(1)}  a_{(1)}  X_P \\ & \qquad  \ot  X_Q \left( (v_{d_1})_{(1)} S\big( (v_{d_2})_{(1)}\big )  S\big( (v_{d_3})_{(1)}\big ) (v_{d_4})_{(1)}\right)  \lact S ( a_{(2)}) \\ & \quad \quad \quad \quad\quad \quad \left (  (v_{d_1})_{(2,1)}  S\big((v_{d_2})_{(2,1)}\big )  S\big( (v_{d_3})_{(2,1)}\big ) (v_{d_4})_{(2,1)} \right)  \lact S ( b_{(2)})\\
&=\partial (b_{(3)} a_{(3)}) v_e \ot 
 \left ((v_{d_1})_{(2)} \ot (v_{d_2})_{(2)} \ot (v_{d_3})_{(2)} \ot (v_{d_4})_{(2)} \right) 
\ot  b_{(1)}  a_{(1)}  X_P \\ & \qquad  \ot  X_Q \left( (v_{d_1})_{(1,1)} S\big( (v_{d_2})_{(1,1)}\big )  S\big( (v_{d_3})_{(1,1)}\big ) (v_{d_4})_{(1,1)}\right)  \lact S ( a_{(2)}) \\ & \quad \quad \quad \quad\quad \quad \left (  (v_{d_1})_{(1,2)} S\big((v_{d_2})_{(1,2)}\big )  S\big( (v_{d_3})_{(1,2)}\big ) (v_{d_4})_{(1,2)} \right)  \lact S ( b_{(2)})\\
&=\partial (b_{(3)} a_{(3)}) v_e \ot 
 \left ((v_{d_1})_{(2)} \ot (v_{d_2})_{(2)} \ot (v_{d_3})_{(2)} \ot (v_{d_4})_{(2)} \right) 
\ot  b_{(1)}  a_{(1)}  X_P \\ & \qquad  \ot  X_Q \left( (v_{d_1})_{(1)} S\big( (v_{d_2})_{(1)}\big )  S\big( (v_{d_3})_{(1)}\big ) (v_{d_4})_{(1)}\right)_{(1)}  \lact S (b_{(2)} a_{(2)})\\
&=E_e^{ba} \big(v_e \ot \left (v_{d_1} \ot v_{d_2} \ot v_{d_3}\ot v_{d_4} \right )\ot \left( X_P \ot X_Q \right) \big)
.\end{align*}
In the penultimate step we used the calculation in \eqref{calc: s act}, and the fact that $A$ is an $H$-module algebra.\end{proof}

\begin{lemma}\label{edge-ops-adequate-commute}Let $\Sigma$ be an oriented surface with a cell decomposition $L$. 
\begin{enumerate}\setlength\itemsep{0em}
    \item Suppose that $e$ and $e'$ are  %not 
    edges 
    not
    on the boundary of the same plaquette. Then, given any $a,a' \in A$ we have:
    $$E^a_e \circ E^{a'}_{e'}= E^{a'}_{e'} \circ E^a_e.$$
    \item Suppose that $e$ and $e'$ are in the boundary of the same plaquette $P$, and moreover that the pair $(e,e')$ is adequate, Definition \ref{def:ad-cell-decomposition}. Given any cocommutative elements, $a$ and $a'$ in $A$, we have:
    $$E^a_e \circ E^{a'}_{e'}= E^{a'}_{e'} \circ E^a_e.$$
\end{enumerate}    
\end{lemma}
The proof of this lemma is deferred to \S \ref{sec:commutation_edge}. The reason why we must restrict  to adequate pairs, $(e,e')$, in the second item  is that otherwise the application of $E_e^a$ affects the value of the edges between $v_P$ and the initial vertex of $e'$ in an uncontrollable way, and analogously for the application of $E_{e'}^{a'}$.

Analogously to the vertex projectors, we can now define the edge projectors by acting with the Haar integral of the Hopf algebra $A$.
This gives an idempotent endomorphism, because the Haar integral is idempotent and the edge operators are a representation of $A$:
\begin{definition}[Edge projector]\label{def:edge_proj}
Let $\Lambda \in A$ be the Haar integral of $A$.
Then we define the \emph{edge projector, $E_e : \mathcal{H}_L \lto \mathcal{H}_L$, based at the edge $e$} as
\[ E_e := E_e^\Lambda . \]
\end{definition}
Given the previous two results, we can see that if $L$ is an adequate cell decomposition of $\Sigma$ then all edge projectors commute. This is in general not the case for non-adequate cell decompositions.

\subsection{Plaquette operators
{of the model on $(A \xrightarrow{\partial} H, \lact)$}
}

Fix a pair $(\Sigma,L)$, as in Definition \ref{def:cell-decomposition}.
Let $P \in L^2$ be a plaquette and let $v_P \in L^0$ be its base-point.

We define a family of plaquette operators acting on $\mathcal{H}_L$, whose support is $P$ and set of the edges in the boundary of $P$. 
{The definition of the plaquette operators in our model follows that of the plaquette operators in the \HAKM\ \cite{buerschaper-et-al}, and reduces to the latter when $A = \CC$ (up to minor differences in conventions)}. {The starting point is the `1-holonomy operator' in \cite[Equation (5)]{Higher_Kitaev}.}

Denote by $(e_1, \dots, e_n)$ the ordered set of edges in the boundary of $P$, starting at the base-point $v_P$ and going around the boundary in a counterclockwise way with respect to the orientation of $\Sigma$.
For $i = 1, \dots, n$, let $\theta_i = +1$ if the edge $e_i$ is oriented counterclockwise along the boundary of $P$ 
(i.e. if $e_i$ is the first edge traversed proceeding counterclockwise along $\partial P$ from $\omega_{e_i}$);
and let $\theta_i = -1$ otherwise. Figure \ref{fig:Plaquette} should clarify these conventions.

\begin{figure}[ht!]
 \labellist
\pinlabel $\small{v_P}$ at -28 80
\pinlabel $\small{P}$ at 231 100
\pinlabel $\small{e_1}$ at 90 59
\pinlabel $\small{e_2}$ at 224 28
\pinlabel $\small{e_3}$ at 390 73
\pinlabel $\small{e_4}$ at 404 172
\pinlabel $\small{e_5}$ at 110 157
%\pinlabel $(\theta_1,\theta_2,\theta_3,\theta_4,\theta_5)=(1,-1,1,1,-1)$ at 900 40
\endlabellist
\centering
\includegraphics[scale=0.25]{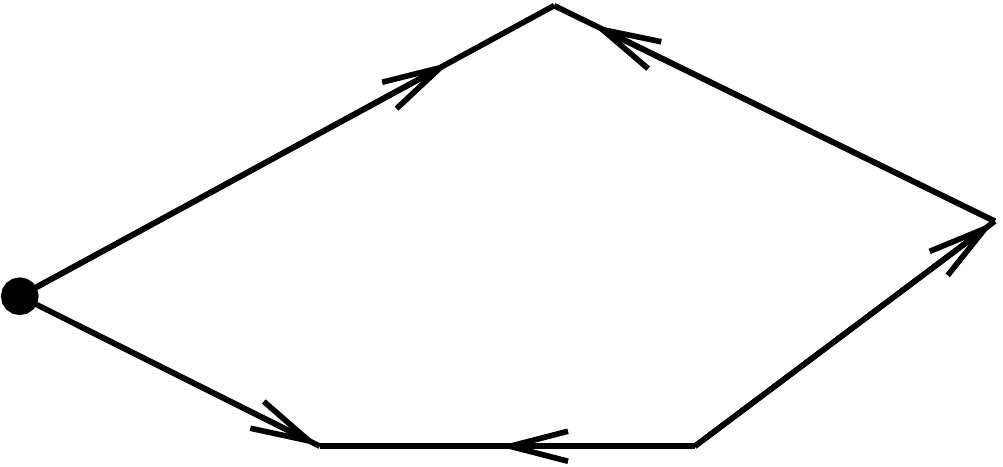}
\caption{{Conventions for plaquette operators. Here $n=5$ and $(\theta_1,\theta_2,\theta_3,\theta_4,\theta_5)=(1,-1,1,1,-1)$.}}
\label{fig:Plaquette}
\end{figure}

{Cf. the very similar case of plaquette operators in the \HAKM\ in \cite[Definition 2.3]{balsam-kirillov} and \cite[Definition 1]{buerschaper-et-al}, our conventions are opposite. In \textit{loc cit} only the formula for the plaquette operators when all the edge orientations match the direction in which we go around the plaquette 
%in 
is presented. It is however understood \cite[\S 3.2]{buerschaper-et-al} that if an edge does not have the same orientation we first apply $S$ to the corresponding  {edge tensor factor}.}

\begin{definition} \label{def:plaquette-operator}
For any $\varphi \in H^*$, the \emph{plaquette operator $F_P^\varphi$ on $\mathcal{H}_L$ based at the plaquette $P$} is
\begin{align*}
F_P^\varphi : \mathcal{H}_L &\lto \mathcal{H}_L, \\
v_{e_1} \ot\cdots\ot v_{e_n} \ot X_P &\lmapsto \big( v_{e_1}^\lrangle{\theta_1}\big)_{(1)}^\lrangle{\theta_1} \ot\cdots\ot \big( v_{e_n}^\lrangle{\theta_n}\big)_{(1)}^\lrangle{\theta_n} \ot (X_P)_{(2)} \\
&\phantom{xxxxx} \varphi\left( \big(v_{e_1}^\lrangle{\theta_1}\big)_{(2)} \cdots \big(v_{e_n}^\lrangle{\theta_n}\big)_{(2)}\,\, ( S \circ \partial) \big( (X_P)_{(1)}\big) \right),
\end{align*}
where we use the notation {in \eqref{eq: <>},} and 
where $v_{e_1} \ot\cdots\ot v_{e_n} \ot X_P \in H_{e_1} \ot\cdots\ot H_{e_n} \ot A_P = H^{\ot n}\ot A$ are in the tensor factors of $\mathcal{H}_L$ associated with the edges $(e_1,\dots,e_n)$ and the plaquette $P$.
\end{definition}
\begin{remark}
For concrete values of the edge orientations $\theta_j$ the formula for the plaquette operators simplifies significantly, as
the antipode is involutive, $S^2=\id{H}$, and an anti-coalgebra map. Explicitly:
$$\big( v_{e_1}^\lrangle{\theta_1}\big)_{(1)}^\lrangle{\theta_1}\otimes \big(v_{e_1}^\lrangle{\theta_1}\big)_{(2)}
 =\begin{cases} ( v_{e_1})_{(1)} \otimes ( v_{e_1})_{(2)}, \textrm{ if $\theta_1=1$,}\\
  (v_{e_1})_{(2)} \otimes S\big((v_{e_1})_{(1)}\big),  \textrm{ if $\theta_1=-1$.}
 \end{cases}$$
 The idea is that on the tensor factors $H_{e_1} \ot\cdots\ot H_{e_n}$ the plaquette operator acts {via the} left or right co-multiplication depending on the orientation of the edge $e_j$ relative to the plaquette $P$. The formula in Definition \ref{def:plaquette-operator} captures this in one closed form for all possible configurations of edge orientations. 
\end{remark}
\begin{example}In the case of the example in Figure \ref{fig:Plaquette}, the plaquette operator  $F_P^\varphi$ is:
\begin{align*}
\big (v_{e_1} &\ot v_{e_2} \ot v_{e_3} \ot v_{e_4} \ot v_{e_5}\big) \ot X_P \\
& \quad \stackrel{F_P^\varphi}{\lmapsto} 
(v_{e_1})_{(1)} \ot (v_{e_2})_{(2)} \ot (v_{e_3})_{(1)} \ot (v_{e_4})_{(1)} \ot (v_{e_5})_{(2)} \ot (X_P)_{(2)} \\
& \quad \quad\quad\quad\quad\quad \varphi\Big ( (v_{e_1})_{(2)}\,\, S\big ( (v_{e_2})_{(1)}\big) \,\,  (v_{e_3})_{(2)} \,\, (v_{e_4})_{(2)} \, S\big ( (v_{e_5})_{(1)} \big)  \, (S\circ \partial)\big ( (X_P)_{(1)} \big) \Big) . 
\end{align*}
\end{example}
\begin{lemma}\label{lem:plaq-ops-a-rep}
Let $P \in L^2$ be a plaquette.
Then the plaquette operators $(F_P^\varphi)_{\varphi\in H^*}$ define a representation of $H^*$ on $\mathcal{H}_L$, that is:
\[ F_P^{\varphi'} \circ F_P^{\varphi} = F_P^{\varphi'\cdot \varphi} \quad\text{ for all } \varphi, \varphi' \in H^* . \]\end{lemma}
\begin{proof}(Sketch.)
We show the proof in the particular case of the configuration in Figure \ref{fig:Plaquette}, which shows the type of calculations required. Applying  $F_P^{\varphi}$, and then  $F_P^{\varphi'}$ we get:
\begin{align*}
v_{e_1} &\ot v_{e_2} \ot v_{e_3} \ot v_{e_4} \ot v_{e_5} \ot X_P \\
& \quad \stackrel{F^\varphi_P}{\lmapsto} 
(v_{e_1})_{(1)} \ot (v_{e_2})_{(2)} \ot (v_{e_3})_{(1)} \ot (v_{e_4})_{(1)} \ot (v_{e_5})_{(2)}  \ot (X_P)_{(2)} \\
& \quad \quad \quad\quad \quad \varphi\Big ( (v_{e_1})_{(2)}\,\, S\big ( (v_{e_2})_{(1)}\big) \,\,  (v_{e_3})_{(2)} \,\, (v_{e_4})_{(2)} \, S\big ( (v_{e_5})_{(1)} \big)  \, (S\circ \partial)\big ( (X_P)_{(1)} \big) \Big)\\
& \quad \smash{\stackrel{F^{\varphi'}_P}{\lmapsto}}\,\, 
(v_{e_1})_{(1,1)} \ot (v_{e_2})_{(2,2)} \ot (v_{e_3})_{(1,1)} \ot (v_{e_4})_{(1,1)} \ot (v_{e_5})_{(2,2)}  \ot (X_P)_{(2,2)} \\
& \quad \quad \quad\quad \quad \varphi\Big ( (v_{e_1})_{(2)}\,\, S\big ( (v_{e_2})_{(1)}\big) \,\,  (v_{e_3})_{(2)} \,\, (v_{e_4})_{(2)} \, S\big ( (v_{e_5})_{(1)} \big)  \, (S\circ \partial)\big ( (X_P)_{(1)} \big) \Big)\\
& \qquad \qquad \quad\quad \quad \varphi'\Big ( (v_{e_1})_{(1,2)}\,\, S\big ( (v_{e_2})_{(2,1)}\big) \,\,  (v_{e_3})_{(1,2)} \,\, (v_{e_4})_{(1,2)} \, S\big ( (v_{e_5})_{(2,1)} \big)  \, (S\circ \partial)\big ( (X_P)_{(2,1)} \big) \Big)\\
& \quad  = (v_{e_1})_{(1)} \ot (v_{e_2})_{(2)} \ot (v_{e_3})_{(1)} \ot (v_{e_4})_{(1)} \ot (v_{e_5})_{(2)}  \ot (X_P)_{(2)} \\
& \quad \quad \quad\quad \quad \varphi\Big ( (v_{e_1})_{(2,2)}\,\, S\big ( (v_{e_2})_{(1,1)}\big) \,\,  (v_{e_3})_{(2,2)} \,\, (v_{e_4})_{(2,2)} \, S\big ( (v_{e_5})_{(1,1)} \big)  \, (S\circ \partial)\big ( (X_P)_{(1,1)} \big) \Big)\\
& \qquad \qquad \quad\quad \quad \varphi'\Big ( (v_{e_1})_{(2,1)}\,\, S\big ( (v_{e_2})_{(1,2)}\big) \,\,  (v_{e_3})_{(2,1)} \,\, (v_{e_4})_{(2,1)} \, S\big ( (v_{e_5})_{(1,2)} \big)  \, (S\circ \partial)\big ( (X_P)_{(1,2)} \big) \Big) \\
& \quad  = (v_{e_1})_{(1)} \ot (v_{e_2})_{(2)} \ot (v_{e_3})_{(1)} \ot (v_{e_4})_{(1)} \ot (v_{e_5})_{(2)}  \ot (X_P)_{(2)} \\
& \quad \quad \quad \varphi\Big ( (v_{e_1})_{(2,2)}\,\, S\big ( (v_{e_2})_{(1)}\big)_{(2)} \,\,  (v_{e_3})_{(2,2)} \,\, (v_{e_4})_{(2,2)} \, S\big ( (v_{e_5})_{(1)} \big)_{(2)}  \, \big ((S\circ \partial)\big ( (X_P)_{(1)} \big) \big)_{(2)}\Big)\\
& \qquad \qquad  \quad \varphi'\Big ( (v_{e_1})_{(2,1)}\,\, S\big ( (v_{e_2})_{(1)}\big)_{(1)} \,\,  (v_{e_3})_{(2,1)} \,\, (v_{e_4})_{(2,1)} \, S\big ( (v_{e_5})_{(1,2)} \big)_{(1)}  \, \big ((S\circ \partial)\big ( (X_P)_{(1)} \big) \big)_{(1)}\Big). 
\end{align*}Since $A$ is an $H$-module algebra, the last formula is exactly the result of the action by $F^{\varphi'\varphi}_P$.
\end{proof}

Since $H$ is semisimple, then so is $H^*$, as we recalled in Theorem \ref{prop:haar-integral0}, note that we are working only with finite-dimensional Hopf algebras over $\CC$. Hence, we again obtain idempotent endomorphisms of $\mathcal{H}_L$ by acting with the Haar integral of $H^*$ via the plaquette operators:
\begin{definition}\label{def:plaq_proj}
Let $\lambda \in H^*$ be the Haar integral of $H^*$.
Then we define the \emph{plaquette projector $F_P : \mathcal{H}_L \lto \mathcal{H}_L$ based at the plaquette} $P$ as
\[ F_P := F_P^{\lambda} . \]
\end{definition}

{For the case of a crossed module of Hopf algebras  induced by a crossed module of groups $(E \xrightarrow{\partial} G, \lact)$, the plaquette projectors recover the plaquette  projectors in \cite[Equation (20)]{Higher_Kitaev}. Those latter plaquette projectors choose the 2-gauge configurations that are fake flat around a plaquette. Plaquette projectors do not arise in the \HKM\ \cite{companion}, due to the fake-flatness being, 
 {by inception},  imposed in the 2-gauge configurations.}

\begin{lemma} \label{lem:plaquette-operators-commute}
Let $P_1 \in L^2$ and $P_2 \in L^2$ be any two distinct plaquettes.
Then
\[ F_{P_1}^{\varphi_1} \circ F_{P_2}^{\varphi_2} = F_{P_2}^{\varphi_2} \circ F_{P_1}^{\varphi_1}, \quad \text{ for all } \varphi_1, \varphi_2 \in H^* . \]
\end{lemma}
\begin{proof}
If $P_1$ and $P_2$ are not adjacent to each other, {in other words if they do not have an edge in common,} then the corresponding plaquette operators $F_{P_1}^{\varphi_1}$ and $F_{P_2}^{\varphi_2}$, for any $\varphi_1, \varphi_2 \in H^*$, have disjoint support in $\mathcal{H}_L$ and therefore they commute.

Let us assume that $P_1$ and $P_2$ are adjacent to each other. Given the restrictions we put on $L$  there exists one single edge, $e$, which lie in both the boundary of $P_1$ and the boundary of $P_2$, as in Figure \ref{fig:Plaquette-ops-commute}.
\begin{figure}[ht!]
 \labellist
\pinlabel $\small{v_{P_1}}$ at -25 80
\pinlabel $\small{P_1}$ at 231 100
\pinlabel $\small{e}$ at 395 75
\pinlabel $\small{P_2}$ at 510 60
\pinlabel $\small{v_{P_2}}$ at 488 141
\endlabellist
\centering
\includegraphics[scale=0.25]{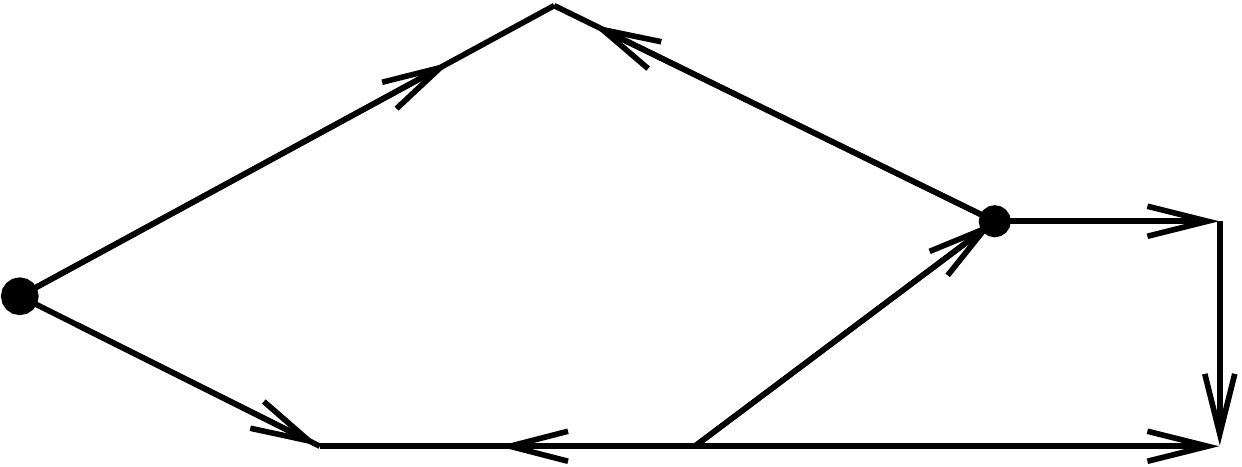}
\caption{{Two plaquettes, $P_1$ and $P_2$, intersecting along an edge $e$.}}
\label{fig:Plaquette-ops-commute}
\end{figure}
Given that $e$ has opposite orientations relative to $P_1$ and to $P_2$, the corresponding operators will act via opposite-sided co-multiplication (that is, one of them left and the other right) of the Hopf algebra $H$ assigned to $e$.
Therefore, due to coassociativity the operators  $F_{P_1}^{\varphi_1}$ and $F_{P_2}^{\varphi_2}$  commute.
\end{proof}

\noindent {In particular, putting $\varphi_1 = \varphi_1=\lambda \in H^*$, and using this last lemma together with the previous one, plaquette projectors commute: $[F_{P_1},F_{P_2}]=0,$ for each $P_1,P_2 \in L^2.$}

\subsection{Edge orientation reversal and plaquette base-point shift}

The total state space $\mathcal{H}_L$ of the \HHK\ model defined in this section does not depend on the edge orientations or plaquette base-points of the cell decomposition $L$, but the vertex, edge and plaquette operators defined on it do.

Now we define linear maps on $\mathcal{H}_L$ which represent the reversal of the orientation of an edge and the moving around of the base-point of a plaquette.
{By showing that these maps commute with some of the  vertex, edge and plaquette operators, in a suitable way,} our calculations in the Appendix in Subsection \ref{sec:proofs-of-commutation-relations} showing the commutation relations of the latter operators will simplify significantly, because it will then suffice to carry out the computations for only the most convenient configurations of edge orientations and base-points.

The edge orientation reversals and plaquette base-point shifts can also be of independent interest, towards showing the independence of the ground-state space of the model of the cell decomposition, a problem which we leave for future work.

Consider the total state space $\mathcal{H}_L = H^{\ot L^1} \ot A^{\ot L^2}$ of the \HHK\ model.
\begin{definition}
For any edge $e \in L^1$, we define the \emph{edge orientation reversal,}
\[ R_e : \mathcal{H}_L \lto \mathcal{H}_L, \]
to act as the antipode $S$ of $H$ on the tensor factor in $\mathcal{H}_L$ associated with the edge $e$ and as the identity on all remaining tensor factors. 
Locally we hence have, showing an edge $e$ connecting vertices $a$ and $b$, $$\xymatrix{&a  \ar[r]_{v_e} &b} \qquad \smash{\stackrel{R_e}{\longmapsto}} \xymatrix{&  a \ar@{<-}[r]_{S(v_e)} & b}  .$$
\end{definition}
Due to the involutivity of the antipode, the edge orientation reversal is an isomorphism and, more specifically, $R_e^2 = \id{\mathcal{H}_L}$. 

\begin{definition}
For any plaquette $P \in L^2$, we define the \emph{(positive) base-point shift},
\[ T^+_P : \mathcal{H}_L \lto \mathcal{H}_L, \]
as follows.
Let $A$ be the tensor factor in $\mathcal{H}_L$ associated with the plaquette $P$ 
and let $H$ be the tensor factor in $\mathcal{H}_L$ associated with the edge in 
the boundary of $P$ incident to the base-point of $P$ in counterclockwise direction around $P$.
Then define $T^+_P$ to act as 
\begin{align*}
H \ot A \ni v \ot X &\stackrel{T_P^+}{\lmapsto} v_{(1)} \ot v_{(2)}^\lrangle{-\theta} \lact X =v_{(2)} \ot v_{(1)}^\lrangle{-\theta} \lact X.
 \end{align*}
Here $\theta \in \{ +1, -1 \}$ is $+1$ if the edge is oriented counterclockwise around $P$ and otherwise $-1$. {We have used Equations \eqref{eq:yetter-drinfeld-condition-trivial}} and \eqref{eq:compatibility<>}.
{The new base-point, $v'_P$ of $P$, is the next vertex of $P$ in counterclockwise order around $P$.}

We analogously define the \emph{(negative) base-point shift}
\[ T^-_P : \mathcal{H}_L \lto \mathcal{H}_L \] to act as:
\begin{align*}
H \ot A \ni v \ot X &\stackrel{T^-_P}{\lmapsto} v_{(1)} \ot v_{(2)}^\lrangle{\sigma} \lact X= v_{(2)} \ot v_{(1)}^\lrangle{\sigma} \lact X,
\end{align*} where $H$ be the tensor factor in $\mathcal{H}_L$ associated with the edge in the boundary of $P$ incident to the base-point of $P$ in \emph{clockwise} direction around $P$. Again, $\sigma \in \{ +1, -1 \}$ is $+1$ if the edge is oriented counterclockwise around $P$ and otherwise $-1$. The new base-point, $v'_P$ of $P$ is the next vertex of $P$ in clockwise order around $P$.

We call the edge from $v_P$ to $v'_P$  the \emph{edge along which the plaquette base-point shift was performed.}
\end{definition}
For example, for a square plaquette:
$$\vcenter{\xymatrix@C=30pt{ \bullet \ar@{-}[d]  \ar@{-}[r]  &\bullet \ar@{-}[d]  \\ v_P\ar[r]_{v} \ar@{}[ur]|X&   \bullet }}  \stackrel{T_P^+}{\lmapsto} \vcenter{\xymatrix@C=30pt{ \bullet \ar@{-}[d]  \ar@{-}[r]  &\bullet \ar@{-}[d]  \\ \bullet \ar[r]_{v_{(1)}} \ar@{}[ur]|{S(v_{(2)})\lact X}&   v'_P }}  =  \vcenter{\xymatrix@C=30pt{ \bullet \ar@{-}[d]  \ar@{-}[r]  &\bullet \ar@{-}[d]  \\ \bullet \ar[r]_{v_{(2)}} \ar@{}[ur]|{S(v_{(1)})\lact X}&   v'_P }} \textrm{ and } \vcenter{\xymatrix@C=30pt{ \bullet \ar@{-}[d]  \ar@{-}[r]  &\bullet \ar@{-}[d]  \\ v_P\ar@{<-}[r]_{v} \ar@{}[ur]|X&   \bullet }}  \stackrel{T_P^+}{\lmapsto}  \vcenter{\xymatrix@C=30pt{ \bullet \ar@{-}[d]  \ar@{-}[r]  &\bullet \ar@{-}[d]  \\ \bullet \ar@{<-}[r]_{v_{(1)}} \ar@{}[ur]|{v_{(2)}\lact X}&   v'_P }}  =  \vcenter{\xymatrix@C=30pt{ \bullet \ar@{-}[d]  \ar@{-}[r]  &\bullet \ar@{-}[d]  \\ \bullet \ar@{<-}[r]_{v_{(2)}} \ar@{}[ur]|{v_{(1)}\lact X}&   v'_P }}   $$
$$\vcenter{\xymatrix@C=30pt{ \bullet \ar@{-}[d]  \ar@{-}[r]  &\bullet \ar@{-}[d]  \\ \bullet \ar[r]_{v} \ar@{}[ur]|X&   v_P }}  \stackrel{T^-_P}{\lmapsto}  \vcenter{\xymatrix@C=30pt{ \bullet \ar@{-}[d]  \ar@{-}[r]  &\bullet \ar@{-}[d]  \\ v'_P \ar[r]_{v_{(1)}} \ar@{}[ur]|{v_{(2)}\lact X}&   \bullet }} =   \vcenter{\xymatrix@C=30pt{ \bullet \ar@{-}[d]  \ar@{-}[r]  &\bullet \ar@{-}[d]  \\ v'_P \ar[r]_{v_{(2)}} \ar@{}[ur]|{v_{(1)}\lact X}&   \bullet }}  \textrm{ and } 
\vcenter{\xymatrix@C=30pt{ \bullet \ar@{-}[d]  \ar@{-}[r]  &\bullet \ar@{-}[d]  \\ \bullet \ar@{<-}[r]_{v} \ar@{}[ur]|X&   v_P }}  \stackrel{T^-_P}{\lmapsto}  \vcenter{\xymatrix@C=30pt{ \bullet \ar@{-}[d]  \ar@{-}[r]  &\bullet \ar@{-}[d]  \\ v'_P \ar@{<-}[r]_{v_{(1)}} \ar@{}[ur]|{S(v_{(2)})\lact X}&   \bullet }}  =   \vcenter{\xymatrix@C=30pt{ \bullet \ar@{-}[d]  \ar@{-}[r]  &\bullet \ar@{-}[d]  \\ v'_P \ar@{<-}[r]_{v_{(2)}} \ar@{}[ur]|{S(v_{(1)})\lact X}&   \bullet }} $$
\begin{remark}\label{rem:shifts_reverse}
Observe that the base-point shifts are isomorphisms and that, more specifically, if $P'$ has the same underlying plaquette as $P$ but its base-point is shifted once in counterclockwise direction around $P$, then $T^-_{P'} = (T^+_P)^{-1}$.
\end{remark}
\begin{remark}
{Note that moving the base-point of a plaquette $P$ all the way around its boundary, coming back to the initial base-point, does not in general induce the identity on $\Hil_L$. This operation however restricts to the identity on a certain subspace $\Hil_L^{\text{ff}}$ (the \emph{fake-flat subspace)} of $\Hil_L$, as shown in Lemma \ref{lem:base-point-all-around} in the Appendix.}
\end{remark}

The following very useful result holds. Most calculations are quite trivial. We shall show examples of the more involved ones in  Subsection \ref{sec:proofs-base-point-shifts-orientation-reversals}, in the Appendix.

\begin{proposition}\label{prop: commut-vertex-ops-shifts_etc}
Let $\Sigma$ be an oriented surface with cell decomposition $L$.
\begin{enumerate}\setlength\itemsep{0em}

\item {Base-point shifts commute with edge orientation reversals.} 

\item Vertex operators commute with edge orientation reversals.
\item  Edge orientation reversals along an edge commute with edge operators based on a different edge.\\
\textbf{NB:} the orientation of the edge in which we apply the edge operator affects the conventions for this edge operator in such a way that we cannot expect  edge operators to commute with reversal of the orientation of the same edge.  
\item  Plaquette operators commute with edge orientation reversals.
\item \label{it:bps-vertex} Base-point shifts commute with vertex operators.
\item \label{it:bps-edge} {Base-point shifts and  edge operators commute, as long as the base-point shift is not along the same edge in which we apply an edge operator}
\item {Base-point shifts on a plaquete commute with plaquette operators on a different plaquette. Base point-shifts commute with plaquette operators $F^\phi_P$, on the same plaquette, if $\phi\in H^*$ is cocommutative. In particular, base-point shifts  commute with plaquette projectors.}
\end{enumerate}
\end{proposition}
\begin{proof}(Sketch.)
We deal with each point separately.
\begin{enumerate}\setlength\itemsep{0em}
    \item   This follows at once from \eqref{eq:yetter-drinfeld-condition-trivial}, and with the fact that $S$ is anti-coalgebra map.
    \item \label{it:vertex-eor} This follows at once from the fact that $S$ is an anti-algebra map.
    \item \label{it:edge-eor} This follows from the explicit formula for the edge operators, combined with the fact that  $S$ is anti-coalgebra map. The Yetter-Drinfeld relations \eqref{eq:yetter-drinfeld-condition-trivial} also play a crucial role.
    \item This again follows from the  fact that  $S$ is anti-coalgebra map.
    \item The proof is  more involved, and can be found in \S \ref{subs:basepoint-shifs and vertex ops}. Lemma \ref{lem:yetter-drinfeld-for-many-factors} plays a crucial role.
    \item The proof is in \S \ref{sec: commu_edge_shift}.    
    \item This is the most laborious proof, and a comprehensive discussion is in  \S\ref{sec:shift-plaquettes}.  The tricky case is when the base-point shift happens in the same plaquette where we apply the plaquette operator.
    \end{enumerate}
\end{proof}
\vspace{-15px}

\subsection{Commutation relations between vertex, edge and plaquette operators}

Having defined all operators, we can now state some of their commutation relations, which we prove in the Appendix. {The following generalises \cite[Theorem 2.4]{balsam-kirillov}.} We do not require $L$ to be adequate here.
\begin{proposition} \label{prop:local-operator-algebra}
Let $\Sigma$ be a compact oriented surface with cell decomposition $L$, and let $(A \xrightarrow{\partial} H)$ be a crossed module of finite-dimensional semisimple Hopf algebras.
Then the associated vertex, edge and plaquette operators, on $\mathcal{H}_L = H^{\ot L^1} \ot A^{\ot L^2}$, enjoy the following relations:
\begin{enumerate}
    \item Let $e$ be any edge of $L$ and $v$ be any vertex of $L$. For all $h \in H$ and $a \in A$ we have:
\begin{enumerate}
\item[1.a)] $E_e^{h_{(1)}\lact a} \circ V_{v,P}^{h_{(2)}} = V_{v,P}^h \circ E_e^a$, if $v$ is the starting vertex of $e$,
\item[1.b)] $E_e^a \circ V_{v,P}^h = V_{v,P}^h \circ E_e^a$, if $v$ is not the starting vertex of $e$.
\end{enumerate}
\item Let $P\in L^2$ be a plaquette and let $e \in L^1$ be an edge.
\begin{enumerate}
\item[2.a)] If $e$ is in the boundary of $P$, then
\[ F_P^{\varphi} \circ E_e^{a} = E_e^a \circ F_P^\varphi,
\quad\text{ for all cocommutative } a \in A, {\textrm{and cocommutative }} \varphi \in H^*. \]
\item[2.b)] If $e$ is not in the boundary of $P$, then
\[ F_P^{\varphi} \circ E_e^{a} = E_e^a \circ F_P^\varphi,
\quad\text{ for all } a \in A, \varphi \in H^*. \]
\end{enumerate}
\item Let $P \in L^2$ be a plaquette and let $v \in L^0$ be a vertex.
\begin{enumerate}
\item[3.a)] If $v$ is in the boundary of $P$, then
\[ F_P^{\varphi} \circ V_{v,P'}^h = V_{v,P'}^h \circ F_P^\varphi,
\quad\text{ for all cocommutative } h \in H, \textrm{ and cocommutative } \varphi \in H^*, \]
where $P' \in L^2$ may be any plaquette with base-point $v$.
\item[3.b)] 
If $v$ is not in the boundary of $P$, then
\[ F_P^{\varphi} \circ V_{v,P'}^h = V_{v,P'}^h \circ F_P^\varphi,
\quad\text{ for all } h \in H, \varphi \in H^* . \]
\end{enumerate}
\end{enumerate}
\end{proposition}
\begin{proof}
See Lemmas \ref{lem:edge-operators-and-vertex-operators}, \ref{lem:edge-and-plaquette-operators-commute} and \ref{lem:plaquette-operators-and-vertex-operators} in the Appendix.
\end{proof}

The commutation relations between plaquette operators, and vertex and edge operators become considerably more complicated outside the cocommutative case, and a general formula would become quite involved. Below we show a sample, for a very particular configuration of the base-point of the plaquette. It is not possible to reduce all other cases to this one, outside the cocommutative case, in that plaquette operators, in general, only commute with base-point shift in the cocommutative case.
\begin{proposition}\label{prop:local-operator-algebra-1}Let $P \in L^2$ be a plaquette with base-point $v \in L^0$ and let $e \in L^1$ be an edge in the boundary of $P$, 
with starting vertex $v$. We have:
\begin{enumerate}
    \item  For all $a \in A$, and  $\varphi \in H^*,$ then:
    \begin{enumerate}
        \item[1.a)] if $e$ is oriented counterclockwise: 
$$F_P^{\varphi (S \partial a_{(3)} \cdot ? \cdot \partial a_{(1)} )} \circ E_e^{a_{(2)}} = E_e^a \circ F_P^\varphi,$$
     \item[1.b)] If $e$ is oriented clockwise then:
     $$ [E_e^a , F_P^\varphi]=0.$$
     \end{enumerate}
\item For all $h \in H$ and $\varphi \in H^*$ we have:
$$F_P^{\varphi(S h_{(3)} \cdot ? \cdot h_{(1)})} \circ V_{v,P}^{h_{(2)}} = V_{v,P}^h \circ F_P^\varphi.$$
\end{enumerate}
Here, given $g,h \in H$,  we define $\varphi(S (g) \cdot ? \cdot h)\colon H \to \CC$  by $x\mapsto \varphi(S (g)) \cdot x\cdot h)$.  % (see e.g. \cite{kassel}).
\end{proposition}
\begin{proof}See lemmas 
    \ref{lem:edge-and-plaquette-operators-commute} and \ref{lem:plaquette-operators-and-vertex-operators} in the Appendix.
\end{proof}

As a consequence of the fact that the vertex, edge and plaquette operators define on $\mathcal{H}_L$ representations of the algebras $H$, $A$ and $H^*$, respectively, and that they satisfy the commutation relations in Proposition \ref{prop:local-operator-algebra-1}, we obtain in total a representation of the algebra which has underlying vector space $H^* \ot (A \rtimes H)$, with the multiplication defined by
\[ (\varphi \ot (a \ot h)) \cdot (\psi \ot (b \ot k)) := \varphi \ \psi(S\partial a_{(3)} Sh_{(3)} \cdot ? \cdot h_{(1)} \partial a_{(1)}) \ot (a_{(2)} \ot h_{(2)})  \]
$\varphi, \psi \in H^*, a, b \in A, h, k \in H$.

In the case that $A = \CC$ is the trivial Hopf algebra,
this algebra is the Drinfeld double $D(H)$ and
we therefore recover the well-known fact that in the Kitaev model for a semisimple Hopf algebra $H$ the local operators form a representation of the Drinfeld double $D(H)$ \cite{buerschaper-et-al, balsam-kirillov}.

\subsection{Commuting-projector Hamiltonian and its ground-state space}\label{sec:grspace}

Recall the definition of vertex $V_v$, edge $E_e$, and plaquette $F_P$ projectors, in Definitions \ref{def:vert_proj}, \ref{def:edge_proj}
and \ref{def:plaq_proj}.    
\begin{theorem} \label{thm:hamiltonian}
Let $\Sigma$ be a compact oriented surface with an adequate cell decomposition $L$, and let $(A \xrightarrow{\partial} H)$ be a crossed module of finite-dimensional semisimple Hopf algebras.
Then the following operators, on $\mathcal{H}_L = H^{\ot L^1} \ot A^{\ot L^2}$, are 
pairwise commuting projectors:
\begin{align*}
V_v, \text{ for all } v \in L^0; 
&&E_e,  \text{ for all } e \in L^1; &&
F_P, \text{ for all } P \in L^2.
\end{align*}

As a consequence, the following operator (the \emph{Hamiltonian}) is a diagonalizable endomorphism of $\mathcal{H}_L$:
\begin{equation}\label{eq:hamiltonian}
h := \sum_{v \in L^0} (1 - V_v) + \sum_{e \in L^1} (1 - E_e) + \sum_{P \in L^2} (1 - F_P),
\end{equation}
whose eigenspace for its lowest eigenvalue $0$ (the \emph{ground-state space}) is
\begin{equation*}
\ker h = \im\Big( \prod_{v \in L^0} V_v \prod_{e \in L^1} E_e \prod_{P \in L^2} F_P \Big)
= \left\{ x \in \mathcal{H}_L \mid \substack{V_v(x) = x \text{ for all } v \in L^0,\\ E_e(x)=x \text{ for all } e \in L^1, \\ F_P(x) = x \text{ for all } P \in L^2.} \right\} 
\end{equation*}
\end{theorem}
\begin{proof}
The commutativity of the projectors follows from Lemmas  \ref{lem:vertex-operators-commute},
\ref{edge-ops-adequate-commute},
\ref{lem:plaquette-operators-commute}, \ref{lem:edge-operators-and-vertex-operators}, \ref{lem:edge-and-plaquette-operators-commute} and \ref{lem:plaquette-operators-and-vertex-operators}. Apart from the case of commutativity between  vertex and edge projectors, which requires Lemma \ref{lem:haar-integral-central-in-crossed-product}, the commutativity relations are proven more generally for $V_{v,P}^a, E_e^b$ and $P_P^\psi$, for all cocommutative elements $a \in H$, $b \in A$ and $\psi \in H^*$. In particular (by Proposition \ref{prop:haar-integral}) the commutativity  relations hold for $V_{v,P}^\ell, E_e^\Lambda$ and $P_P^\lambda$, where $\ell$, $\Lambda$ and $\lambda$, recall, are the Haar integrals of $H$, $A$ and $H^*$.

Mutually commuting projectors are simultaneously diagonalizable and, hence, the map $h$ is diagonalizable, with eigenvalues being all the non-negative integers up to $|L^0| + |L^1| + |L^2|$.
The ground-state space is therefore the eigenspace for eigenvalue $0$, i.e.\ the kernel, which is the space on which all the projectors occurring in the Hamiltonian act by $0$. \end{proof}

\begin{proposition}\label{prop:gr_state}
The ground-state space $\ker h$ can %equivalently
be described as follows.
A vector \[ x = \bigotimes_{e \in L^1} v_e \ot \bigotimes_{P \in L^2} X_P \in H^{\ot L^1} \ot A^{\ot L^2}, \] which by abuse of notation we consider to be a sum of pure tensors while omitting the summation symbol, is a ground state, if, and only if, the following conditions are satisfied:
\begin{enumerate}\setlength\itemsep{0em}
    \item \label{it:plaquettes}
For every $P \in L^2$, and, following the notation of Definition \ref{def:plaquette-operator},
  $(e_1, \dots, e_n)$ are the edges in the boundary of $P$ in counterclockwise order, starting at $v_P$, we have:
\begin{align*}
& {y} \ot ((v_{e_1}^{\langle \theta_1 \rangle})_{(1)})^{\langle \theta_1 \rangle} \ot\cdots\ot ((v_{e_n}^{\langle \theta_n \rangle})_{(1)})^{\langle \theta_n \rangle} \ot (X_P)_{(2)} \ot \left( (v_{e_1}^{\langle \theta_1 \rangle})_{(2)} \cdots (v_{e_n}^{\langle \theta_n \rangle})_{(2)} S\big(\partial (X_P)_{(1)} \big)\right)\\
&\qquad \qquad \qquad = {y} \ot v_{e_1} \ot\cdots\ot v_{e_n} \ot X_P \ot 1_H . 
\end{align*}
(Here $\theta_j = +1$ if the edge $e_j$ is oriented counterclockwise, around $P$, and $\theta_j = -1$ otherwise.)
\item \label{it:edges}
For every $e \in L^1$: 
\[ E_e^a ( x ) = \eps(a) x, \quad \text{ for all } a \in A.\]
\item\label{it:vertices}
For every $v \in L^0$, and for any choice of $P\in L^2$, adjacent to $v$: \[ V_{(v,P)}^h ( x ) = \eps(h) x, \quad \text{ for all } h \in H.\] \end{enumerate}
\end{proposition}
{This result ultimately follows from the representation-theoretical properties of the Haar integral; {see e.g. \cite[Corollary 1.4]{balsam-kirillov}\cite[Lemma B.9]{meusburger2016hopf}}.}
\begin{proof}
By the previous proposition, the ground state is $\bigcap_{v \in L^0} \im\big( V_v\big)\cap \bigcap_{e \in L^1} \im\big( E_e\big) \cap  \bigcap_{P \in L^2} \im\big( F_P\big).$
For any module $M$ over a finite-dimensional semisimple Hopf algebra {$H'$}, the Haar integral {$\ell' \in H'$} projects onto the subspace of {$H'$}-invariants: \[ \ell' . M = M^{{H'}} := \{ m \in M \mid h.m = \eps(h) m, \text{ for all } h \in {H'} \} . \]
(This follows since the Haar integral $\ell'$ is idempotent and $x\ell'=\ell' x= \eps(x) \ell'$, for all $x \in H'$.)

Now apply this respectively to $H^*$, $A$ and $H$, and their actions on $\Hil_L$ given respectively by plaquette operators, edge operators and vertex operators.
{This leads directly to Items \ref{it:edges}  and \ref{it:vertices}, showing when an element is invariant under all edge and all vertex projectors.}

{{The formula for when an element is invariant under  all plaquette projectors, $F_P^{\lambda}$, in Item \ref{it:plaquettes}}, follows from the fact that an element $x \in \Hil_L$ lies in the space of $H^*$-invariants under the representation $F_P^{(\_)}$ of $H^*$ if, and only if, for any $\psi \in H^*$:}  
\begin{align*}
& {y} \ot ((v_{e_1}^{\langle \theta_1 \rangle})_{(1)})^{\langle \theta_1 \rangle} \ot\cdots\ot ((v_{e_n}^{\langle \theta_n \rangle})_{(1)})^{\langle \theta_n \rangle} \ot (X_P)_{(2)} \ \psi\left( (v_{e_1}^{\langle \theta_1 \rangle})_{(2)} \cdots (v_{e_n}^{\langle \theta_n \rangle})_{(2)} S\partial (X_P)_{(1)} \right)\\
&= {y} \ot v_{e_1} \ot\cdots\ot v_{e_n} \ot X_P \, \psi(1_H), 
\end{align*}
{since recall $\eps(\psi)=\psi(1_H)$.}
\end{proof}

\section{Some particular cases of the \HHK\ model}

\subsection{The $\GEXY$-class of models}

Let $(E \xrightarrow{\partial} G, \lact)$ be a crossed module of groups.
Let $X$ and $Y$ be finite groups on which $G$ acts by automorphisms, and let $f : Y \to X$ be a $G$-equivariant group morphism, such that  $\im(\partial) \subseteq G$ acts trivially on $Y$ and $X$. In Subsection \ref{sec:main_example} we constructed a crossed module of semisimple Hopf algebras $\big(\Fun{X} \ot \CC E \xrightarrow{\partial} \Fun{Y}\rtimes\CC G, \lact\big)$. The resulting class of \HHK\ models as defined in Section \ref{sec:model} is here called the $\GEXY$-class of models.

{In this section we investigate} {more closely} three particular cases of the associated $\GEXY$-class of models. We will also determine their ground-state spaces on a surface $\Sigma$, proving that they are canonically independent of the  adequate cell decomposition of $\Sigma$.

In order to simplify our discussion, we will only discuss the typical case when $\Sigma$ is an oriented surface without boundary and $L$ is a {triangulation} of $\Sigma$, {with a total order on the set of vertices. This gives an adequate cell decomposition as explained in Lemma \ref{lem:adequatefromtriang}.

The particular cases considered for $\big(\Fun{X}\ot\CC E \xrightarrow{\partial} \Fun{Y}\rtimes\CC G, \lact\big)$ are:
\begin{itemize}\setlength\itemsep{0em}
    \item The $\GoXo$-case, where the groups $E$ and $Y$ both are trivial. This turns out to give Potts model \cite{Martin_Potts} when $G$ is trivial. The general case of the model is a coupling between Kitaev quantum-double model \cite{Kitaev} and Potts model.
    \item A generalisation of $\GoXo$-case, the $\GEXo$ case, where also $E$ may be non-trivial will be briefly discussed.
    \item The $\ooXY$ case where the groups $G$ and $E$ each are trivial.
\end{itemize}

{Note that the $\GEoo$-case, where the groups $X$ and $Y$ each are trivial, coincides with the $n=2$ case of the 2-group Kitaev model in \cite[Equation (35)]{Higher_Kitaev},  in the Introduction called the \sHKM. In \textit{loc cit}, its ground-state space  was proven to be related with Yetter homotopy 2-type TQFT \cite{Yetter}, and in particular {to be} canonically independent of the {triangulation} of $\Sigma$. The $\GEoo$-case  reduces to the Kitaev model \cite{Kitaev} when $E$ is trivial; see \cite[page 8]{Higher_Kitaev}.}

{In all the particular cases above, the ground-state space is identified with the $\CC$-linear span of the set of homotopy classes of maps $\Sigma \to B$. Here $B$ is a certain space, actually a homotopy 1-type or 2-type, \cite{martins_porter07,Porter}, obtained as the classifying space of a certain groupoid, or crossed module of groupoids, \cite{Brown_Higgins,brown_hha} derived from $(G,E,X,Y)$. (The way the groupoids and crossed modules of groupoids are  constructed depends on the example, so we do not have a general treatment of the full $(G,E,X,Y)$-model.)}

{The latter identification of the ground-state spaces as  free vector spaces on  sets of homotopy classes of maps from $\Sigma$ into 1-types or 2-types, here homotopy finite spaces \cite{Quinn,martins_porter21},  means in particular that Quinn's {{finite total homotopy}} TQFT \cite{Quinn,martins_porter21} provides a TQFT whose state spaces are canonically isomorphic to the ground-state spaces of the corresponding \HHK\ model, in the particular cases we consider. A discussion is in Subsubsection \ref{sec:relQuinn}.}

{The calculation of the {ground-state space} of the $\GoXo$, $(G,E,X,1)$  and $\ooXY$-model is  based on results of Brown and Higgins \cite[Theorem A]{Brown_Higgins}, and Brown-Higgins-Sivera \cite[11.4.iii]{Brown_Higgins_Sivera}, expressing the set of homotopy classes of maps $M \to B_\mathcal{A}$, where $\mathcal{A}$ is a crossed complex of groupoids, a generalisation of crossed modules, and $B_\mathcal{A}$ is its classifying space, in terms of homotopy classes of crossed complex maps $\Pi(M) \to \mathcal{A}$. Here $M$ is any space with a CW-complex structure and $\Pi(M)$ denotes its fundamental crossed complex \cite{brown_hha}.} {The classifying space of a crossed complex was originally constructed in \cite{Brown_Higgins}, using a simplicial setting, for which some more explanation can be found in \cite{martins_porter21}, and, more recently, in a cubical setting in \cite[Chapter 11]{Brown_Higgins_Sivera}.}

It should however be {re-emphasised} that the tricks that permit the application of Brown-Higgins theorem are quite different in the $\GEXo$- and $\ooXY$-models. It is an open problem whether the results can be generalised to the full $\GEXY$-model.

In the $\GooY$-case, the Hopf algebra $\Fun{X}\otimes\CC E$ becomes $\CC$, so this example is not a crossed module case. In particular our model reduces to the Hopf-algebraic Kitaev model (with Hopf algebra $\Fun{Y}\rtimes\CC G)$ \cite{buerschaper-et-al}. Therefore its ground-state space is known to be triangulation-independent.

\subsection{The $\GoXo$-case {and its relation to Quinn's {finite total homotopy} TQFT}}\label{goxo}
Let $G$ be a finite group. Let $X$ be a finite group on which $G$ acts by automorphisms. The general construction in Subsection \ref{sec:main_example} for $E=1=Y$ gives a crossed module of Hopf algebras:
\[ ( \Fun{X} \xrightarrow{\partial} \CC G, \lact ) .\]
Here $G$ acts on $\Fun{X}$ as $(g\lact \psi)(x)=\psi(g^{-1} \lact x)$, for $\psi \in \Fun{X}$ and $x \in X$.
The boundary map $\partial$ is trivial: $\partial(\psi)=\psi(1_X) 1_G$. Recalling Examples \ref{ex:CG} and \ref{ex:FG},
the Haar integral of $\CC G$ is $\ell=\frac{1}{|G|} \sum_{g \in G} g$ and the Haar integral of $\Fun{X}$ is $\Lambda =\delta_{1_X}$. Clearly $g \lact \Lambda=\Lambda$ for each $g\in G$; cf. Lemma  \ref{lem:haar-integral-central-in-crossed-product}.

\subsubsection{An explanation of the $\GoXo$-model}
Let $\Sigma$ be a surface with a triangulation $L$, with a total order on the set of vertices, with the adequate cell decomposition in Lemma \ref{lem:adequatefromtriang}.
\begin{definition}
The total state space assigned to $(\Sigma,L)$ is 
\[ \mathcal{H}_{L,X,G} := \CC G^{\otimes L^1} \ot \Fun{X}^{\otimes L^2}.
\] 
\end{definition}
The total state space assigned to $(\Sigma,L)$ is naturally described in terms of, {what we call here,} $\GX$-colourings.
We define:
\begin{definition} \label{def:colourings}
 For sets $A$ and $B$, an \emph{$(A,B)$-colouring, $\F=(\F_1,\F_2)$, of $L$} is a labelling, $\F_1\colon L^1 \to A,$ of the edges of $L$ by elements of $A$, and another (independent) labelling $\F_2\colon L^2 \to B$ of the plaquettes of $L$ by elements of $B$.
\end{definition}
\noindent Clearly $\mathcal{H}_{L,X,G}$ is isomorphic to the free vector space on the set of $\GX$-colourings $\F$ of $L$. The language of $\GX$-colourings, similar to that of \cite{Kitaev,companion,martins_porter07}, is quite suitable for specifying the actions of the vertex, edge and plaquette operators of the $\GoXo$-model.

\medskip

Let $\F=(\F_1,\F_2)$ be a $\GX$-colouring of $(\Sigma,L)$.
\medskip

\noindent{\textbf{Vertex operators}} \\
Let $v \in L^0$ be a vertex and let $P \in L^2$ be an adjacent plaquette.
The vertex operators $V_{v,P}^g$ depend  only on the underlying vertex $v$, $V_{v,P}^g = V_v^g$, due to cocommutativity of the group algebra $\CC G$.
Given a vertex $v \in L^0$ and an element $g \in G$, we obtain the $\GX$-{vertex operator}, $V_v^g(\F) = (V_v^g(\F)_1, V_v^g(\F)_2)$:
\begin{itemize}\setlength\itemsep{0em}
\item given $t \in L^1$:
\[\big(V_v^g(\F)\big)_1(t)=\begin{cases}  g \, \F_1(t),\textrm{ if $v$ is the starting vertex of $t$}, \\
  \F_1(t)\, g^{-1}, \textrm{ if $v$ is the target vertex of $t$}, 
  \\
  \F_1(t), \textrm{ if $v$ is not incident to $t$};
\end{cases}\]
\item given $P \in L^2$:
\[\big(V_v^g(\F)\big)_2(P)=\begin{cases}  g\lact \F_2(P),\textrm{ if $v$ is the base-point of  $P$},\\
  \F_2(P), \textrm{ if $v$ is not the base-point of $P$}.
\end{cases}\]
\end{itemize}
The
 vertex projector is  given by:
\[
V_v=V_v^{\ell}=\frac{1}{|G|}\sum_{g \in G} V_v^g.
\]
\noindent \textbf{Edge operators}\\
In order to describe the edge operators, we need another bit of notation. Let $\F$ be a $\GX$-colouring of $(\Sigma,L)$. Let $P \in L^2$ be a plaquette. Let $t$ be an edge in the boundary of $P$, and $w$ be the initial point of $t$. Let as usual $v_P$ denote the base-point of $P$.

Let us define the following notation $\hol_{P,t}^\F$ {(where $\hol$ stands for \emph{holonomy})}:
\begin{itemize}\setlength\itemsep{0em}
\item If $w=v_{P}$, let $\hol_{P,t}^\F=1_G$
\item Otherwise, consider the unique path from $v_{P}$ to $w$, that does not pass through $t$.
Then let $\hol_{P,t}^\F$ be the product of each element of $G$ assigned to the edges {transversed} when going from $v_{P}$ to $w$ (or its inverse if the edge is transcribed in the opposite direction to its orientation).
Some examples are  in the diagram below {(an edge going from a vertex $a$ to a vertex $b$ is denoted $t_{(a,b)}$):}
\end{itemize}
\vspace{-1em}
\begin{align*}
\vcenter{\xymatrix{v\ar[rr]^{t_{(v,v_P)}} \ar[rrd]_{t_{(v,w)}} && v_P\\
            && w\ar[u]_t
            }}&\stackrel{\hol_{P,t}^\F}{\longmapsto} \F_1(t_{(v,v_P)})^{-1} \F_1(t_{(v,w)}),  
&&\vcenter{\xymatrix{v_P\ar[rr]^{t_{(v_P,v)}} \ar[rrd]_{t_{(v_P,w)}} && v\\
            && w\ar[u]_t
            }} \stackrel{\hol_{P,t}^\F}{\longmapsto}  \F_1(t_{(v_P,w)}), 
%%%%%%%%%            
\\
%%%%%%%%%
             \vcenter{\xymatrix{v\ar@{<-}[rr]^{t_{(v_P,v)}} \ar[rrd]_{t_{(v,w)}} && v_P\\
            && w\ar[u]_t
            }} &\stackrel{\hol_{P,t}^\F}{\longmapsto}  \F_1(t_{(v_P,v)}) \F_1(t_{(v,w)}),
            && \vcenter{\xymatrix{v_P\ar[rr]^{t_{(v_P,v)}} \ar@{<-}[rrd]_{t_{(w,v_P)}} && v\\
            && w\ar[u]_t
            }} \stackrel{\hol_{P,t}^\F}{\longmapsto}  \F_1(t_{(w,v_P)})^{-1}.
\end{align*}

Given an oriented edge $t \in L^1$, let $P$ and $Q$ be the two plaquettes that have $t$ in common, where $P$ is on the left-hand side (noting that $\Sigma$ is oriented).
Then, {unpacking Definition \ref{def:edge-operator}}, for $\xi \in \Fun{X}$: 
\[ E_t^{\xi}(\F)= 
\xi\Big( \big( (\hol_{P,t}^\F)^{-1} \lact \F_2(P) \big) \big( (\hol_{{Q},t}^\F)^{-1} \lact \F_2(Q) \big)^{-1} \Big) \F .
\]

The Haar integral $\Lambda$ in $\Fun{X}$  is the delta function $\delta_{1_X}$. So the edge projector $E_t$ is such that:
\[E_t(\F)=E_t^\Lambda(\F)=\begin{cases} \F, \textrm{ if }  (\hol_{P,t}^\F)^{-1} \lact \F_2(P) =  (\hol_{Q,t}^\F)^{-1} \lact \F_2(Q),\\
0,\textrm{ otherwise}.
\end{cases}
\]

\noindent \textbf{Plaquette operators} \\
Let $P$ be a plaquette, with base-point $v_P$. Let $\varphi$ be an element of the dual algebra of $\CC G$, i.e.\ $\varphi$ is canonically a function $\varphi\colon G \to \CC$. Given a $\GX$-colouring $\F$, we need a new bit of notation, $\hol^\F_{\partial P}$, to be the product of the elements of $G$ assigned to the edges of the boundary of $P$, when we trace the boundary of $P$ counterclockwise from $v_P$ to $v_P$. As in the definition of $\hol_{P,t}^\F$ above, if an edge %$t$ 
is 
%traced 
{transversed}
in the opposite orientation, we put $\F_1(t)^{-1}$ instead. An example is below:
\[\vcenter{
\xymatrix{v\ar[rr]^{t_{(v,v_P)}} \ar[rrd]_{t_{(v,w)}} && v_P\\
            && w\ar[u]_{t_{(w,v_P)}}}}
         \stackrel{\hol_{\partial P}^\F}{\longmapsto} \F_1(t_{(v,v_P)})^{-1} \F_1(t_{(v,w)}) \F_1(t_{(w,v_P)}).\]

{Unpacking Definition \ref{def:plaquette-operator}, we have for the plaquette operator, given $\varphi \in (\CC G)^*$:}
\[F_P^\varphi(\F)=\varphi(\hol^\F_{\partial P}) \F.
\]
The plaquette projector is $F_P(\F) = F_P^{\delta_{1_G}}(\F)$, using the Haar integral of $(\CC G)^*$, see Example \ref{ex:FG}. Hence
\[F_P(\F)=
\begin{cases}\F, \textrm{ if } \hol^\F_{\partial P} =1_G,\\
0, \textrm{ otherwise}.
\end{cases}
\]

\begin{example}
 Suppose $G$ is the trivial group. Let us see that the $\GoXo$-model coincides with the  {$|X|$-Potts model on $(\Sigma,L^*)$ \cite{Martin_Potts}, where $|X|$ is the cardinality of $X$.} Here $L^*$ is the dual lattice to $L$, \cite[Section 3]{Turaev_Viro}, so we have one vertex (of $L^*$) for each plaquette of $L$,  an edge (of $L^*$)  connecting a pair of vertices  any time two triangles (in $L$) share an edge, and finally each vertex of $L$ gives a plaquette of $L^*$. {(We will go back to dual cell decompositions later, in \S \ref{sec:XY_indual}.)}

Indeed, the Hilbert space for the $\GoXo$-model, with $G=\{1_G\}$, {is isomorphic to the free vector space on the set of labellings of the vertices of $L^*$ by elements of $X$.} Vertex projectors and plaquette projectors each are the identity. {Each edge projector acts as the identity if the labels on the plaquettes on each side of the edge coincide, and as zero otherwise. This is exactly the $|X|$-Potts model, in the conventions of \cite[\S 1.1.]{Fatimah}. }
\end{example}

{Note that the Hamiltonian for the $\GoXo$-model does not depend on the group structure in $X$, and only on the cardinality of the set $X$. However the edge operators do depend on the group operation in $X$.}

\begin{remark}{If $X$ has only one element then the $\GoXo$-model coincides with Kitaev's quantum double model \cite{Kitaev} for $G$.}
{The general case of the $\GoXo$-model can  be seen as a coupling between the $|X|$-state Potts model and the Kitaev model for $G$.} 
\end{remark}

\subsubsection{Fully flat $\GX$-colourings of $(\Sigma,L)$ and the ground-state of the $\GoXo$-model}\label{gsdgoxo}
\begin{definition}
A $\GX$-colouring of $(\Sigma,L)$ is called \emph{fully flat} if:
\begin{itemize}\setlength\itemsep{0em}
    \item \textbf{(Flatness on plaquettes)} Given any plaquette $P \in L^2$, $\hol_{\partial P}^\F=1_G$.
    \item \textbf{(Flatness on edges)} Given any edge $t \in L^1$,
    \[(\hol_{{P_1},t}^\F)^{-1} \lact \F_2(P_1) =  (\hol_{{P_2},t}^\F)^{-1} \lact \F_2(P_2) , \]
    where $P_1$ and $P_2$ are the two plaquettes that have $t$ as the common edge in their boundaries.
\end{itemize}
We let $\Phi_{\GoXo}(\Sigma,L)$ denote the set of fully flat $\GX$-colourings of $(\Sigma,L)$.
\end{definition}
\begin{lemma}
The set $\Phi_{\GoXo}(\Sigma,L)$ is invariant under vertex operators $V_v^g$. Moreover, we have an action of $G^{L^0}$ on the set of fully flat $\GX$-colourings, {where} $(g_v)_{v \in L^0} \in G^{L_0}$, acts as
$\prod_{v \in L^0} V_v^{g_v}$.
\end{lemma}
\begin{proof}
{An easy calculation proves that indeed  $\Phi_{\GoXo}(\Sigma,L)$ is invariant under all vertex operators. 
The second statement follows from the fact that vertex operators $V_v^g$ and $V_w^h$ commute for any two distinct vertices $v,w \in L^0$.}
\end{proof}
{It is clear that the subspace of the total state space of the $\GoXo$-model on $(\Sigma, L)$ of vectors that are invariant under all plaquette projectors and edge projectors is isomorphic to the free vector space on  $\Phi_{\GoXo}(\Sigma,L)$.}
Moreover, the ground-state space of the $\GoXo$-model is the subspace of $G^{L^0}$-invariant vectors in the free vector space on $\Phi_{G,X}(\Sigma,L)$. Therefore:
\begin{lemma}
{The ground-state space of the $\GoXo$-model on $(\Sigma,L)$ is canonically isomorphic to the free vector space on $\Phi_{\GoXo}(\Sigma,L)/G^{L^0}$.}
\end{lemma}

Since $G$ acts on the underlying set of $X$, we can define the action groupoid, $X//G$. Objects of $X//G$  are elements $x \in G$. The set of morphisms $x \to y$ is given by the set of all pairs $(x,g)\in  X \times G$ such that $g^{-1} \lact x =y$. The composition of the morphisms,
\[(x \xrightarrow{(x,g)} g^{-1} \lact x)\textrm{   and   }(g^{-1} \lact x  \xrightarrow{(g^{-1} \lact x,h)} (gh)^{-1} \lact x), \] is \[(x \xrightarrow{(x,gh)} (gh)^{-1} \lact x) 
.\]
{(We shall use several other examples of action groupoids, $Z//H$, below, for a group $H$ acting on a set $Z$, which will follow the conventions just stated. The set of objects of $Z//H$ is $Z$, the set of morphisms is $Z\times H$, and source and target maps, and composition are as above.)}

{Given a triangulation $L$ of $\Sigma$,
we put $\Sigma^0_L$ for the set of vertices of $L$, and $\Sigma^1_L$ for the subspace of $\Sigma$  made out of the union of 0- and 1-simplices of $L$. So $\Sigma_L^i$ is the $i$-skeleton of the CW-decomposition given by $L$; see e.g. 
\cite[Definition 16]{companion}.} \begin{lemma}\label{main_idea}
There is a one-to-one correspondence between fully flat $\GX$-colourings, $\F$, of $(\Sigma,L)$ and groupoid functors
\[f_\F\colon \pi_1(\Sigma,\Sigma^0_L) \to X // G.
\]
This means that we have a one-to-one correspondence between fully flat $\GX$-colourings $\F=(\F_1,\F_2)$ of $(\Sigma,L)$ and pairs of maps $\F_0\colon L^0 \to X$ and $\F_1\colon L^1 \to G$ ($\F_1$ does not change), such that:
\begin{itemize}\setlength\itemsep{0em}
    \item  \textbf{flatness around plaquettes}: given any plaquette $P \in L^2$, $\hol_{\partial P}^\F=1_G$,
    \item given any edge $t \in L^1$, with starting vertex $v$ and target vertex $w$:
    $\F_0(w)=\F_1(t)^{-1} \lact \F_0(v)$.
\end{itemize}
\end{lemma}
\begin{proof} Cf. \cite[\S 2.5 \& \S 5.1.6]{companion}.
The fundamental groupoid  {$\pi_1(\Sigma,\Sigma^0_L)$} is the quotient of the free groupoid on the graph $(L^0, L^1)$ consisting of the vertices and edges of $L$, with a relation at each plaquette $P \in L^2$. {That relation says that the  composition of the elements of {$\pi_1(\Sigma,\Sigma^0_L)$} assigned to the edges in the boundary $P$, when going around $P$ from the base-point $v_P$ to $v_P$ again is the identity;  see \cite{brown_hha},} as explained in \cite[\S 2.5]{companion}. Therefore $\F_1$ gives a functor   {${\pi_1(\Sigma,\Sigma^0_L)}\to \{\ast\} // G\cong G$.} To {lift} it to a functor ${\pi_1(\Sigma,\Sigma^0_L)}\to X // G$ we start by giving the value of the latter on vertices.

Let $(\F_1,\F_2)$ be a fully flat $\GX$-colouring. Then if $P$ is a plaquette, its base-point $v_P$ can be assigned the element $R(P,v_p) := \F_2(P) \in X$. We can extend this to the other vertices of $P$ using the rule $R(P,v') = \F_1(t)^{-1} \lact R(P,v) $, whenever $t$ is an edge in the boundary of $P$, {connecting $v$ to $v'$, and} oriented from $v$ to $v'$. {(If $t$ is oriented from $v'$ to $v$ we put instead $R(P,v') = \F_1(t) \lact R(P,v) $.)} Given the flatness condition of $\F_1$ around plaquettes, this gives a well-defined map $v'\mapsto  R(P,v')$ from the set of vertices in the boundary of $P$ to $X$. 

Now let us see that $R(P,v)$ depends only on $v$ and not on the plaquette $P$ that $v$ belongs to. If $v$ is a vertex, and $v$ belongs to $P$ and $P'$, and there is an edge $t$ incident to $v$, separating $P$ and $P'$, then the flatness condition of $\F$ on the edge $t$ implies that $R(P,v)=R(P',v)$. This can easily be used to show that all plaquettes $Q$ containing $v$ give the same $R(Q,v)$. 

So given $v \in L^0$ we can put $\F_0(v):=R(P,v)$, where $P$ is any plaquette containing $v$.

The rest of the statements follow straightforwardly.
\end{proof}
\begin{lemma}
Under the correspondence of Lemma \ref{main_idea}, we have a one-to-one correspondence between orbits of $G^{L^0}$ in $\Phi_{G,X}(\Sigma,L)$ and equivalence classes of groupoid functors $f_\F\colon \pi_1(\Sigma,\Sigma_L^0) \to X // G$, considered up to natural transformations.
\end{lemma}
\begin{proof} Cf. \cite[\S 4.3.1]{companion}. These calculations only require checking compatibility between the languages of functors and natural transformations and that of vertex operators. 
\end{proof}

{Given a groupoid $\Gamma$, we let $B_{\Gamma}$ be its classifying space {\cite{Willerton}}. Classifying spaces of groupoids are particular cases of classifying spaces of crossed complexes, as defined in \cite{Brown_Higgins,brown_hha,martins_porter07}.}

{
Cf. \cite[\S 5.2]{companion}.
We will now use \cite[Theorem A]{Brown_Higgins} (see also \cite[Theorem 7.16]{brown_hha} and \cite[\S 11.4.iii]{Brown_Higgins_Sivera}).  {This theorem gives a canonical identification between the set of groupoid functors $\pi_1(\Sigma, \Sigma_{L}^0) \to X//G$, considered up to natural transformation\footnote{{In the case of functors between groupoids, crossed complex homotopies (the language of \cite{Brown_Higgins}) boil down to natural transformations between groupoid functors.}}, and homotopy classes of maps $\Sigma \to B_{X//G}$, for any triangulation $L$ of $\Sigma$. 
This permits us to see that the ground-state space of the $\GoXo$-model is canonically independent of the chosen triangulation $L$ of $\Sigma$. }}

\begin{theorem} \label{thm:GX-ground-states}
Consider a pair of finite groups $G$ and $X$, with $G$ acting on $X$ by automorphisms. Let $\Sigma$ be an oriented surface with a triangulation $L$, {with a total order on the set of vertices, and consider the adequate cell decomposition in Lemma \ref{lem:adequatefromtriang}}. {There is a canonical isomorphism between the ground-state space of the \HHK\ model for $\GoXo$ and the free vector space on the set of} homotopy classes of maps from $\Sigma$ to $B_{X // G}$, the classifying space of the groupoid $X // G$. In particular the ground-state space of the model is canonically independent of  the triangulation $L$ of $\Sigma$.
\end{theorem}
\begin{proof}
This follows directly by applying Brown-Higgins classification theorem \cite[Theorem A]{Brown_Higgins}{\cite[Theorem 11.4.19]{Brown_Higgins_Sivera}.} For explanation see \cite[\S 5.2]{companion}, \cite{martins_porter07} {and \cite[\S 8.2]{martins_porter21}}.
\end{proof}

\subsubsection{Relation with Quinn's {finite total homotopy} TQFT}\label{sec:relQuinn}
Let $m$ be a non-negative integer.
A space $X$ is called an \emph{$m$-type} \cite{martins_porter07,Porter}, if  $\pi_i(X,x)=0$ whenever $i>m$, for all possible choices of base-point $x \in X$. A space is called \emph{homotopy finite} \cite{Quinn,martins_porter21,HLA} if {it is an $m$-type, for some $m$, and furthermore} it has only a finite number of path-components and all of their homotopy groups are finite. 
Classifying spaces of groupoids $\Gamma$ are homotopy 1-types, whose fundamental group at each point is isomorphic to a corresponding hom-group $\hom_\Gamma(a,a)$ in $\Gamma$, where $a$ is an object of $\Gamma$. Consequently, classifying spaces of finite groupoids are homotopy finite spaces.

Quinn constructed in \cite{Quinn} what he called the \emph{{finite total homotopy} TQFT, $\mathscr{Q}_{\mathcal{B}}$}; see also 
 {\cite[Section 4]{martins_porter21} for a more recent account}. It is a $(n+1)$-TQFT defined for all spatial dimensions $n$; in particular it {gives a}   $(2+1)$-dimensional TQFT. Quinn's {finite total homotopy} TQFT, $\mathscr{Q}_{\mathcal{B}}$,  depends on a (fixed) homotopy finite space $\mathcal{B}$, called \emph{the base space}. Explicitly, the $(n+1)$-TQFT $\mathscr{Q}_{\mathcal{B}}$ sends a closed $n$-manifold $M$ to the free vector space on the set of homotopy classes of maps $M \to \mathcal{B}.$  {The linear maps assigned to cobordisms $W\colon M \to N$ between manifolds are derived from the \emph{homotopy order}, called in \cite{HLA,Baez_Dolan,martins_porter21} \emph{homotopy cardinality}, of certain  spaces of functions $W \to \mathcal{B}$.}

We hence have:
\begin{theorem}[Relation with Quinn's {finite total homotopy} TQFT]\label{QTFT1}{
There exists a (2+1)-dimensional TQFT whose state spaces are canonically isomorphic to the ground-state spaces of the Hopf-algebraic higher Kitaev model for $\GoXo$. Namely consider Quinn's {finite total homotopy} TQFT $\mathscr{Q}_{\mathcal{B}}$
\cite{Quinn,martins_porter21} with base space $\mathcal{B}=B_{X//G}$. This TQFT sends each surface $\Sigma$ to the free vector space on the set of homotopy classes of maps $\Sigma \to B_{X//G}$.}
\end{theorem}
\begin{proof}This follows directly from Theorem \ref{thm:GX-ground-states} and the construction of Quinn's {finite total homotopy} TQFT $\mathscr{Q}_{\mathcal{B}}$.
\end{proof}

\noindent Note that the {finite total homotopy TQFT $\mathscr{Q}_{\mathcal{B}}$ is for $\mathcal{B}= B_{X//G}$  constructed as the Dijkgraaf-Witten TQFT, with trivial cocyle \cite{DW}, except % that 
using a groupoid, in this case $X // G$,  rather than a group. This follows from the discussion in \cite{martins_porter07}. {More details can be found in \cite[\S 8.2 and \S 8.4.1]{martins_porter21}, in the more general languages of crossed complexes and of extended TQFTs.}

\subsection{The ground-state space of the $\GEXo$-model}

The results in Subsubsection \ref{gsdgoxo} extend to the $\GEXo$-model. So consider a crossed module of groups $(E \xrightarrow{\partial} G, \lact)$, and another group $X$ on which $G$ acts by automorphisms, such that $\im (\partial) \subseteq G$ acts trivially on $X$. We will use the language of crossed modules of groupoids \cite{brown_hha}\cite{brown_icen}\cite[\S 2.1]{companion}, and in particular the fundamental crossed module ${\Pi_2(\Sigma,\Sigma^1_L,\Sigma^0_L)}$ of a surface $\Sigma$, with a CW-decomposition \cite{brown_hha}.  Here $\Sigma$ is given the CW-decomposition arising from  a triangulation $L$.

{Consider the trivial action $\lact$ of $E$ on $X$, so if $x \in X$ and $e \in E$, then
$e \lact x = x =\partial(e) \lact x$. Therefore, the action groupoid $X//E$ is totally disconnected (i.e.\ there are no morphisms between different objects), and we have a groupoid functor $f\colon X // E \to X // G$, which is the identity on objects, 
 {and sends the morphism $(x,e)$ to $(x,\partial(e)).$} We have a groupoid action  of $X//G$ on $X//E$ such that} \[(x \xrightarrow{(x,g)} g^{-1} \lact x)  \lact ( g^{-1} \lact x \xrightarrow{(x,e)}  g^{-1} \lact x)= {(  x \xrightarrow{(x,g\lact e)}   x).}\]
This gives a crossed module of groupoids  denoted by $\mathcal{G}(G,E,X)$. 
The same type of argument as in Subsubsection \ref{gsdgoxo}, combined with the construction in the last section in \cite{companion} gives:
\begin{theorem}  
Let $(\Sigma,L)$ be an oriented surface with a triangulation, {with a total order on the set of vertices, with the  adequate cell decomposition of Lemma \ref{lem:adequatefromtriang}}.

There exists a canonical isomorphism between the ground-state space of the $\GEXo$-model and the free vector space on the set of homotopy classes of crossed module maps {$\Pi(\Sigma,\Sigma^1_L,\Sigma^0_L) \to \mathcal{G}(G,E,X)$}, as defined in \cite{brown_icen,Brown_Higgins}\cite[\S 4.3.1]{companion}.
In particular, there is a canonical isomorphism between the ground-state space  {of the $\GEXo$-model} and the free vector space on the set of homotopy classes of maps $\Sigma \to B_{\mathcal{G}(G,E,X)}$, the classifying space \cite{Brown_Higgins,brown_hha} of $\mathcal{G}(G,E,X)$.
\end{theorem}
\begin{proof}
On the combinatorial side, the proof of the statements follows the same pattern as in the $\GoXo$-case. The fact that $\partial(E)\subseteq G$ acts trivially in $X$  plays a key role.
The final statement follows (directly) from \cite[Theorem A]{Brown_Higgins}{\cite[Theorem 11.4.19]{Brown_Higgins_Sivera}}, as explained in \cite[\S 5.2]{companion}.
\end{proof}
\begin{remark}
 Cf.\ Theorem \ref{QTFT1}, including the preliminary discussion in {Subsubsection}~\ref{sec:relQuinn}. Classifying spaces of crossed modules of groupoids are homotopy 2-types \cite{Brown_Higgins,brown_hha,martins_porter07}. If the crossed modules are finite then their classifying spaces have only a finite number of path-components, each of which with finite fundamental and second homotopy groups for all choices of a based point. And hence they are homotopy finite spaces.
In particular, Quinn's {finite total homotopy} TQFT $\mathscr{Q}_{\mathcal{B}}$, with base space $\mathcal{B}=B_{\mathcal{G}(G,E,X)}$, again gives a TQFT whose state spaces are the ground-state spaces of the $\GEXo$-model. 
\end{remark}

\subsection{The $\ooXY$-model {and its relation to Quinn}{'s} {finite total homotopy}  {TQFT}}

Let $f\colon  Y \to X$ be a group homomorphism, {where $X$ and $Y$ are finite}. We consider the {Hopf} crossed module defined in Subsection \ref{sec:main_example} for the case when $G$ and $E$ are both the trivial group.
We thus obtain a crossed module of Hopf algebras $(\Fun{X} \xrightarrow{f^*} \Fun{Y}, \lact)$, where the action of $\Fun{Y}$ on $\Fun{X}$ is the trivial one, namely: \[\psi\lact \eta= \eps(\psi) \eta=\psi(1_Y) \eta.\]
We will use the Haar integrals in $\Fun{X}$, $\Fun{Y}$ and  $\Fun{Y}^*\cong\CC Y$, as in Examples \ref{ex:CG} and \ref{ex:FG}.

Let $\Sigma$ be an oriented surface (with no boundary). 
{We continue to work with adequate cell decompositions arising from triangulations of $L$, with a total order on the set of vertices, as in Lemma \ref{lem:adequatefromtriang}.}
First, as a vector space, the total state space $\mathcal{H}_L$ as defined in Section \ref{sec:model} is here: \[ \mathcal{H}_L = \Fun{Y}^{{\otimes L^1}} \ot \Fun{X}^{{{\otimes L^2}}} \cong \linhull_\CC ( Y^{L^1} \times X^{L^2} ) . \]

As in Subsection \ref{goxo}, it is convenient to reformulate the total state space in terms of $\XY$-colourings $\R=(\R_1,\R_2)$ of $(\Sigma,L)$, pairs of maps $\R_2\colon L^2 \to X$ and $\R_1\colon L^1 \to Y$, see Definition \ref{def:colourings}.
Clearly $\mathcal{H}_L$ is isomorphic to the free vector space on the set of $\XY$-colourings of $(\Sigma,L)$. 
 
If $P$ is a plaquette, with vertices $a_P<b_P<c_P$, {recall} we put as base point  $v_P=a_P$.  Put $d_0(P)=(b_P,c_P)$, $d_1(P)=(a_P,c_P)$ and $d_2(P)=(a_P,b_P)$.

\subsubsection{Plaquette operators}\label{sec:XY-plaqu}
Let $z\in Y {\subset \Fun{Y}^* = \CC Y}$. Let $P \in L^2$. {We will now determine the plaquette operator $F_P^z$, in the $\ooXY$ case, according to the general definition of plaquette operators in Definition \ref{def:plaquette-operator}.}

Suppose that the orientation on $P$ induced by the order of the vertices is opposite to the orientation induced by {$\Sigma$}. The plaquette operator
$F_P^z$ is such that the colours on the plaquette $P$ and its edges transform as:
\begin{align*}\vcenter{  
\xymatrix@R=8pt@C=8pt{&& b_P \ar[ddr]^{y_{d_0(P)}} \\ && x\\
& a_P \ar[uur]^{y_{d_2(P)}} \ar[rr]_{y_{d_1(P)}} && c_P  } } 
 &\stackrel{ F_P^z(\R)}{\lto} \sum_{\substack{ p' p=y_{d_2(P)} \\
{q' q=y_{d_0(P)}} \\ 
r' r=y_{d_1(P)}^{-1} \\ 
{s' s=x} }}
\delta(p,q)\, \delta(p,r^{-1})\,
\delta(f(p),(s')^{-1}) \,
\delta(z,p) \,
\delta(f(z),(s')^{-1})
\vcenter{\xymatrix@R=8pt@C=8pt{& b_P \ar[ddr]^{q'} \\ & s\\
 a_P \ar[uur]^{p'} \ar[rr]_{(r')^{-1}} && c_P  }}\\
&=
\vcenter{\xymatrix@R=8pt@C=8pt{&& b_P \ar[ddr]^{y_{d_0(P)} z^{-1}} \\ && f(z)x\\
& a_P \ar[uur]^{ y_{d_2(P)} z^{-1}} \ar[rr]_{ z y_{d_1(P)}} && c_P \, . }}
\end{align*}
Outside of $P$ the colourings on simplices do not change. 

If the orientation on $P$ induced by the total order of the vertices of the triangle is the same as the orientation induced by that of $\Sigma$,  the plaquette operator $F_P^z$ is:
\[
\vcenter{\xymatrix@R=8pt@C=8pt{&& b_P \ar[ddr]^{y_{d_0(P)}} \\ && x\\
& a_P \ar[uur]^{y_{d_2(P)}} \ar[rr]_{y_{d_1(P)}} && c_P  }} \stackrel{ F_P^z(\R)}{\lto}\vcenter{
\xymatrix@R=8pt@C=8pt{&& b_P \ar[ddr]^{z y_{d_0(P)} } \\ && f(z)x\\
& a_P \ar[uur]^{ z y_{d_2(P)} } \ar[rr]_{  y_{d_1(P)} z^{-1}} && c_P  \, .}}
\]

\subsubsection{Edge operators}\label{sec:XYedge}
Let $t \in L^1$ be an oriented edge. {Given $\psi \in \Fun{X}=\mathrm{span}_{\CC}\{\delta_x \mid x \in X\}$, let us determine the edge operator $E_t^{\psi}$, as in  Definition \ref{def:edge-operator}. Let  $P$ and $Q$ be the two plaquettes that have $t$ as a common edge.}
Given that the action of $\Fun{Y}$ on $\Fun{X}$ is the trivial one, namely $\psi \lact \eta=\psi(1_Y) \eta$, the formula for the edge  operators $E_t^\psi$ in Definition \ref{def:edge-operator}  simplifies enormously. In particular the edge operators $E_t^\psi$  are independent of the base-points of the plaquettes $P$ and $Q$, and only depend on the colouring of the edge $t$. {Let us explain the general formula for $E_t^{\delta_w}$ if $w\in X$. In the figures below, $t$ is the middle edge coloured by $z\in Y$. The other edges in the figure can have arbitrary orientations.} For $w \in X$, the edge operator $E_t^{\delta_w}$ is:
\[
\begin{gathered}\xymatrix@R=8pt@C=8pt{&& & \\
&\ar@{-}[rru]^{y} \ar@{-}[rrd]_{y'} &x && x'& \ar@{-}[llu]_{y'''} \ar@{-}[lld]^{y''}\\
&&&\ar[uu]^z
} \end{gathered}
\quad \stackrel{
E_t^{\delta_w}}{\longmapsto}  \quad   \delta(f(z)w^{-1}x,x')
\begin{gathered}
\xymatrix@R=8pt@C=8pt{& & \\
\ar@{-}[rru]^{y} \ar@{-}[rrd]_{y'} &x && x'& \ar@{-}[llu]_{y'''} \ar@{-}[lld]^{y''}\\
&&\ar[uu]^z
}    
\end{gathered}
\]
({Note that $\delta(f(z)w^{-1}x,x')= \delta(x\, {x'}^{-1}\, f(z),w)$}.) The edge projector $E_t=E_t^{\delta_{1_X}}$ is hence:
\[
\begin{gathered}\xymatrix@R=8pt@C=8pt{&& & \\
&\ar@{-}[rru]^{y} \ar@{-}[rrd]_{y'} &x && x'& \ar@{-}[llu]_{y'''} \ar@{-}[lld]^{y''}\\
&&&\ar[uu]^z
} \end{gathered} \quad
\stackrel{
E_t}{\longmapsto}   \quad  \delta(f(z)x,x') 
\begin{gathered}\xymatrix@R=8pt@C=8pt{& & \\
\ar@{-}[rru]^{y} \ar@{-}[rrd]_{y'} &x && x'& \ar@{-}[llu]_{y'''} \ar@{-}[lld]^{y''}\\
&&\ar[uu]^z
}\end{gathered}
\]

\subsubsection{Vertex operators}\label{sec:vertex_ops_XY} {We now unpack how the vertex operators defined in the general case in Subsection \ref{sec:vertex_ops} look like in the case of the $\ooXY$-model.
Let $v\in L^0$, and choose a plaquette $P$ that has $v$ as a vertex. Given $\psi \in \Fun{Y}=\mathrm{span}_\CC\{\delta_z \mid z \in Y\}$, we determine the vertex operator $V_{v,P}^\psi.$ It suffices to consider the case $V_{v,P}^{\delta_z}$, where $z \in Y.$}

Let $t_1, \dots, t_n$ be the edges of $(\Sigma,L)$ incident to $v$, starting from the second (hence last) edge of $P$ incident to $v$, when going around $v$ counterclockwise, and in counterclockwise order. We let $\theta_{v,t_i}=-1$ if the edge $t_i$ is pointing towards $v$ and $\theta_{v,t_i}=1$ if $t_i$ is pointing away from $v$.
Let $P_1,\dots, P_k$ be the plaquettes of $(\Sigma,L)$, possibly none, that have $v$ has a base-point. 
Fix a $\XY$-colouring $\R$. Then $ V_{v,P}^{\delta_z}$ {can only change} the colours of edges $t_1,\dots, t_n$ and the plaquettes $P_1,\dots, P_k$. Let $(y_1,\dots, y_n; x_1, \dots x_k)$ be those colours. {(It turns out that only the colours of the edges change under the vertex operators.)} Put:
\[\Delta^{(n+k)}(z)=\sum_{\substack{w_1,\dots,w_{n+k} \in Y \\  w_1\dots w_{n+k}=z }} w_1 \otimes \dots \otimes w_{n+k}.\]
Noting that the action of $\Fun{Y}$ on $\Fun{X}$ is the trivial one $\psi \lact \phi= \psi(1_Y) \phi$, 
it follows that:
\begin{align*}
&(y_1,\dots, y_n; x_1, \dots, x_k)\\ &\stackrel{V_{v,P}^{\delta_z}}{\longmapsto} \sum_{w_1\dots w_{n+k} = z} \delta(w_1,y_1^{\theta_{v,t_1}}) \dots \delta(w_n,y_n^{\theta_{v,t_n}})\, \delta(w_{n+1},1_Y)\dots \delta(w_{n+k},1_Y)\, (y_1,\dots, y_n; x_1, \dots, x_k)\\
&=\sum_{w_1\dots w_{n+k} = z} \delta(w_1,y_1^{\theta_{v,t_1}}) \dots \delta(w_n,y_n^{\theta_{v,t_n}}) (y_1,\dots, y_n; x_1, \dots, x_k)\\
& =\delta(y_1^{\theta_{v,t_1}} \dots y_n^{\theta_{v,t_n}},z)\, (y_1,\dots, y_n; x_1, \dots, x_k).
\end{align*}
Hence given an $\XY$-colouring $\R$ we have:
\[ V_{v,P}^{\delta_z}(\R)=\delta\big(\R_1(t_1)^{\theta_{v,t_1}} \dots \R_1(t_n)^{\theta_{v,t_n}} , z\big)\, \R.
\]

The vertex projector $V_{v}=V_{v,P}^{\delta_{1_Y}}$ takes the form:
\[V_v(\R)=\delta\big(\R_1(t_1)^{\theta_{v,t_1}} \dots \R_1(t_n)^{\theta_{v,t_n}}, 1_Y\big )\, \R. 
\]

\subsubsection{The $\ooXY$-model on the dual cell decomposition.}\label{sec:XY_indual}

The explicit forms of the plaquette, edge and vertex operators in the $\ooXY$-model become very transparent in the dual cell decomposition to $(\Sigma,L)$. {(This is essentially as in \cite[Lemma~2.6]{balsam-kirillov}, which corresponds to the $X=\{1\}$ case.)} So consider the dual cell decomposition $L^*$ to $L$; see \cite[Section 3]{Turaev_Viro}. {We have a 0-cell $P^*$ of $L^*$ for each 2-cell (plaquette) $P$ of $L$, at its barycentre. We have an oriented $1$-cell $t^*$ of $L^*$ for each oriented 1-cell (edge) $t \in L^1$, oriented from the vertex of $L^*$ corresponding to the plaquette on the right-hand side of $t$, towards the vertex corresponding to the plaquette on the left-hand side of $t$.} Finally we have a $2$-cell $v^*$ of $L^*$ for each $0$-cell (vertex) $v$ of $L$. The local configuration is as in Figure \ref{fig:dual_lat}.
\begin{figure}[ht!]
    \centering
    \includegraphics[scale=0.25]{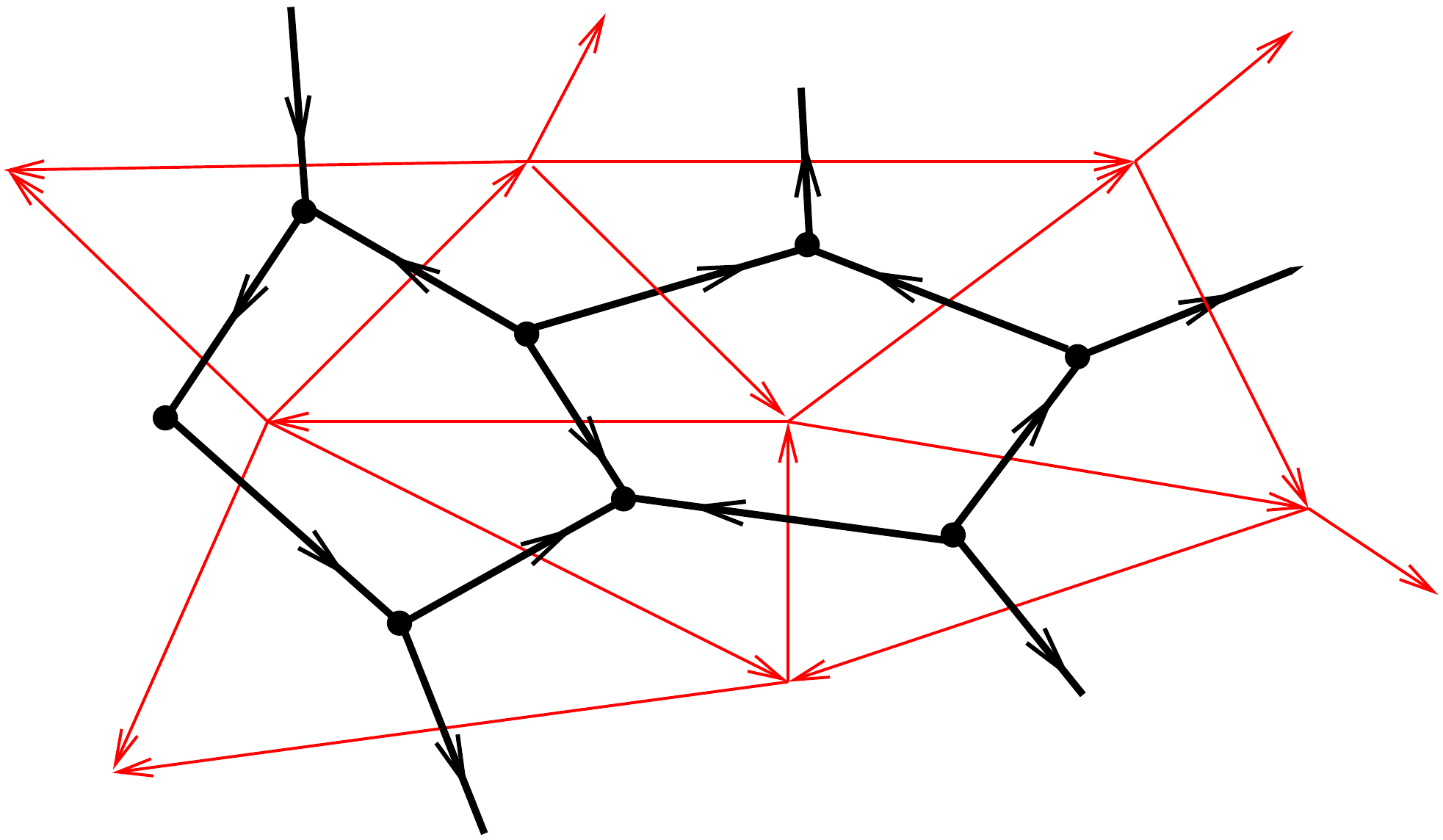}
    \caption{The local picture of a  triangulation $L$ (in red) and its dual cell decomposition $L^*$ (in black, {and with thicker edges}).}
    \label{fig:dual_lat}
\end{figure}

The total {Hilbert} space of the $\ooXY$-model is given by the free vector space on the set of all $\XY$-colourings of $(\Sigma,L)$.
Dually, such a colouring $\R = (\R_1 : L^1 \to Y, \R_2 : L^2 \to X)$ corresponds to labellings $\R^*_1 : (L^*)^1 \to Y$ and $\R^*_0 : (L^*)^0 \to X$. {To $\R^*=(\R^*_0,\R^*_1)$ we call the \emph{ dual colouring} to $\R$.} 
 {The total Hilbert space of the $\ooXY$-model will from now on be seen as the free vector space on the set of those dual colourings $\R^*$. We will graphically represent dual colourings as in Figure \ref{fig:dual_lat_with_labels}.} 
\begin{figure}[ht!]
 \labellist
\pinlabel $\small{\red{v}}$ at 355 166
\pinlabel $\small{\red{P}}$ at 411 199
\pinlabel $\small{\red{t_4}}$ at 355 66
\pinlabel $\small{\red{t_2}}$ at 227 319
\pinlabel $\small{\red{t_3}}$ at 94 172
\pinlabel $\small{\red{t_1=t}}$ at 546 300
\pinlabel $\small{\red{t_5}}$ at 586 128
\pinlabel $\small{x_1}$ at 551 217
\pinlabel $\small{x_2}$ at 377 309
\pinlabel $\small{x_3}$ at 144 234
\pinlabel $\small{x_4}$ at 241 113
\pinlabel $\small{x_5}$ at 472 110
\pinlabel $\small{y_1}$ at 482 265
\pinlabel $\small{y_2}$ at 294 298 
\pinlabel $\small{y_3}$ at 193 162
\pinlabel $\small{y_4}$ at 376 104
\pinlabel $\small{y_5}$ at 526 185
\pinlabel $\small{y'}$ at 641 256
\pinlabel $\small{\blue{P^*}}$ at 530 250
\pinlabel $\small{\red{Q}}$ at 350 230
\pinlabel $\small{\blue{Q^*}}$ at 355 272
\pinlabel $\small{\blue{t_1^*}}$ at 415 292
\endlabellist
\centering
\includegraphics[scale=0.4 ]{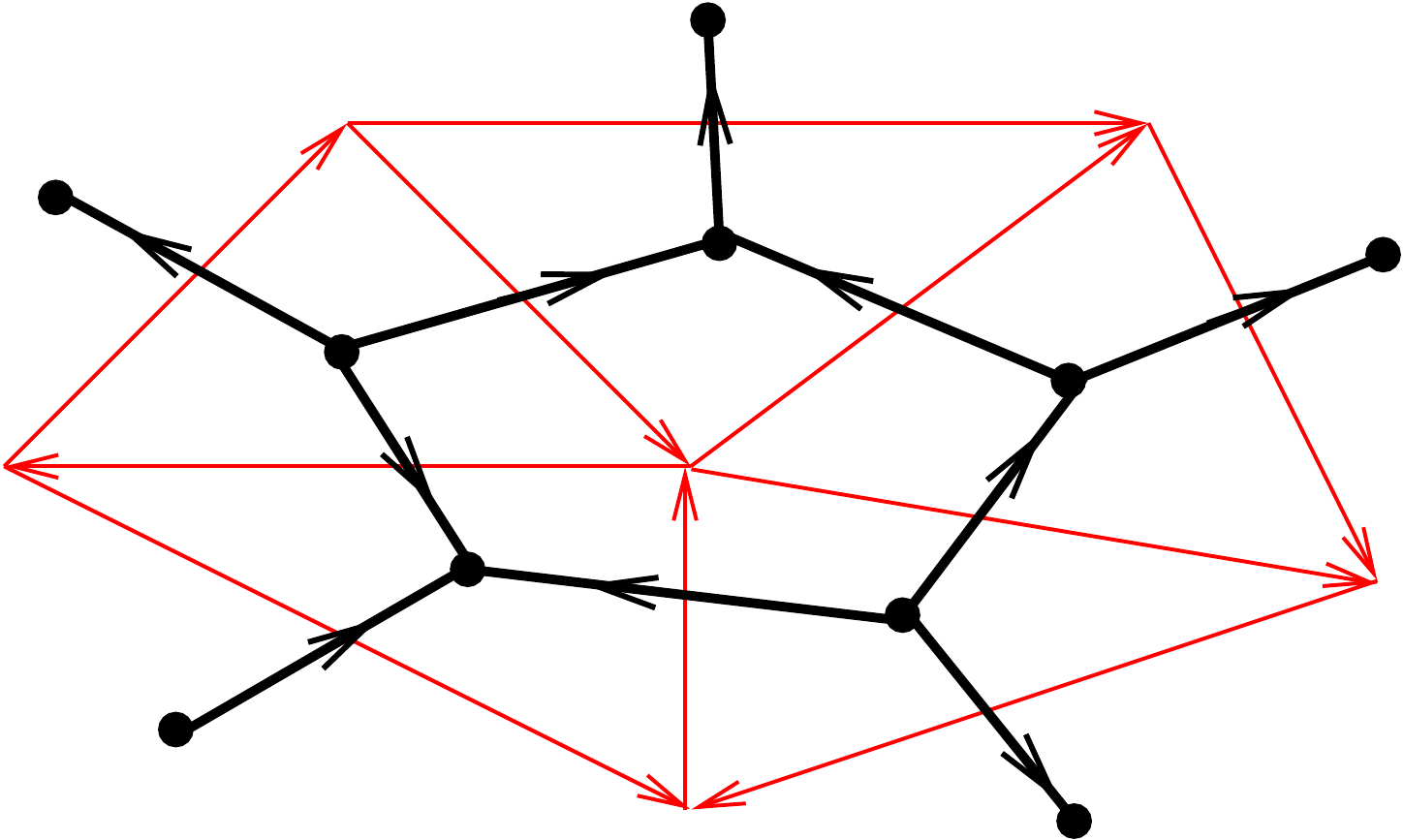}
\caption{A dual colouring $\R^*$ to an $\XY$-colouring $\R$. {The simplices of the original triangulation have labels in red. Cells of the dual cell decomposition are %written 
labelled
in blue.} The labellings given by the dual colouring $\R^*$ are all in black. Hence $y_1,\dots,y_5 \in Y$ and $x_1,\dots,x_5 \in X.$}
\label{fig:dual_lat_with_labels}
\end{figure}

\subsubsection*{Plaquette operators in the dual picture} Cf. \S \ref{sec:XY-plaqu}.
The plaquette operator $F_P^z\colon \mathcal{H}_L \to \mathcal{H}_L$, for $z \in Y$, is such that for an edge $t \in (L^*)^1$:%, if $v:=P^* \in (L^*)^0$:
\begin{equation}\label{Eq:pl-edge}
F_P^z(\R^*)_1(t)=\begin{cases} z\, (\R^*)_1(t), \textrm{ if } t \in (L^*)^1 \textrm{ starts in  } {P^*},\\
(\R^*)_1(t)\, z^{-1}, \textrm{ if } t \in (L^*)^1 \textrm{ ends in  } {P^*},\\
 (\R^*)_1(t), \textrm{ if } t \textrm{ is not adjacent to } {P^*};
\end{cases}
\end{equation}
and for a vertex $w\in (L^*)^0$:\begin{equation}\label{Eq:pl-vertex}
F_P^z(\R^*)_0(w)=\begin{cases} f(z)\, F_P^z(\R^*)_0(w), \textrm{ if } w={P^*}, \\
 F_P^z(\R^*)_0(w), \textrm{ if } w\neq {P^*}.
 \end{cases}
\end{equation}
{The graphical picture for the action of $F_P^g$ on the dual colouring in Figure \ref{fig:dual_lat_with_labels} is as in Figure \ref{fig:dual_lat_with_labels_plaquette-action}.}
\begin{figure}[ht!]
 \labellist
\pinlabel $\small{\red{v}}$ at 355 164
\pinlabel $\small{\red{P}}$ at 411 199
\pinlabel $\small{\red{t_4}}$ at 355 66
\pinlabel $\small{\red{t_2}}$ at 229 316
\pinlabel $\small{\red{t_3}}$ at 96 171
\pinlabel $\small{\red{t_1=t}}$ at 548 300
\pinlabel $\small{\red{t_5}}$ at 586 129
\pinlabel $\small{x_2}$ at 374 311
\pinlabel $\small{x_3}$ at 144 234
\pinlabel $\small{x_4}$ at 241 113
\pinlabel $\small{x_5}$ at 473 110
\pinlabel $\small{g\, y_1}$ at 484 267
\pinlabel $\small{y_2}$ at 291 305 
\pinlabel $\small{y_3}$ at 191 162
\pinlabel $\small{y_4}$ at 376 104
\pinlabel $\small{y_5 \, g^{-1}}$ at 543 180
\pinlabel $\small{g \, y'}$ at 649 258
\pinlabel $\small{f(g) x_1}$ at 571 215
\pinlabel $\small{\blue{P^*}}$ at 530 250
\pinlabel $\small{\red{Q}}$ at 350 230
\pinlabel $\small{\blue{Q^*}}$ at 355 273
\pinlabel $\small{\blue{t_1^*}}$ at 410 295
\endlabellist
\centering
\includegraphics[scale=0.4 ]{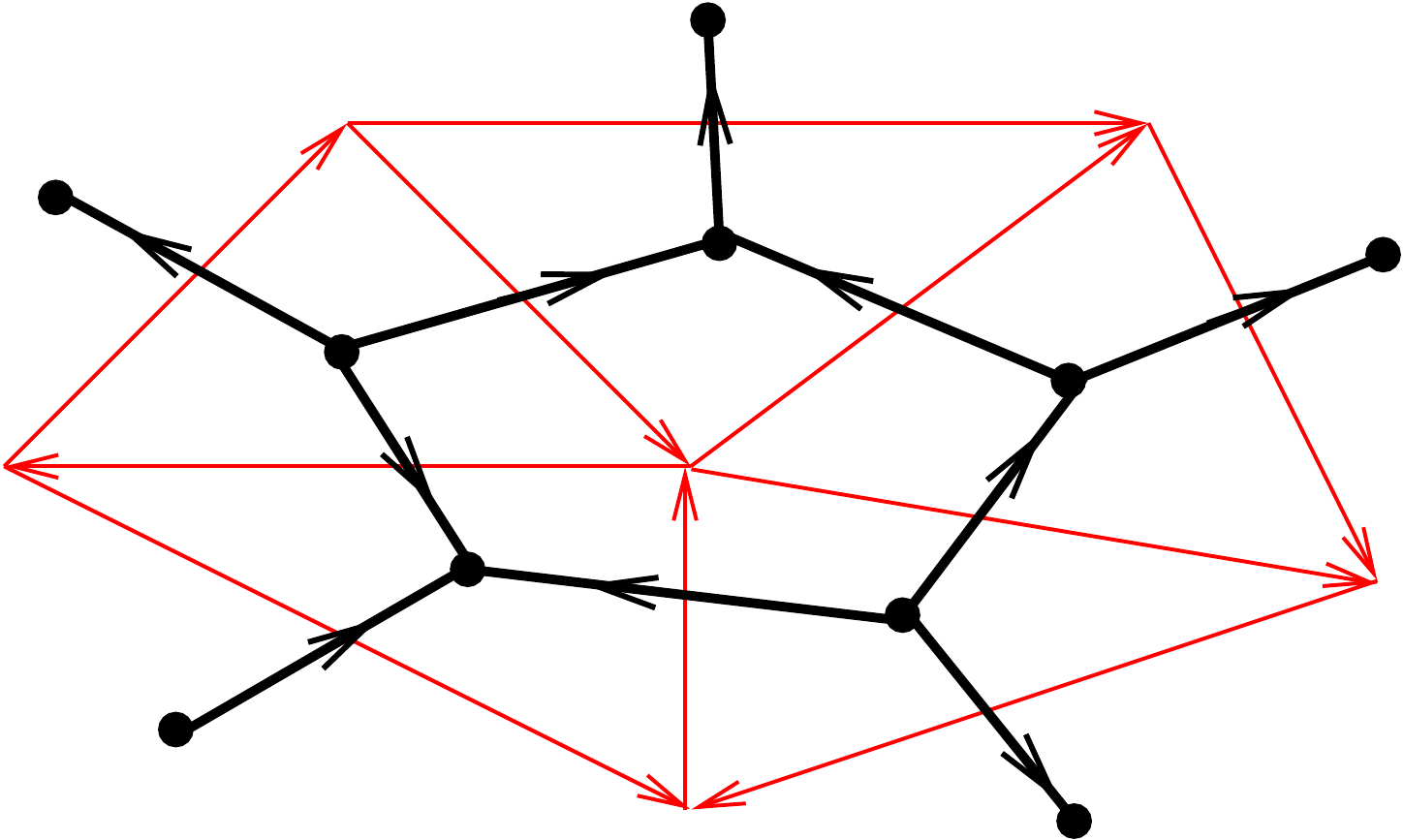}
\caption{The plaquette operator $F_P^g$ in the dual colouring picture. The figure shows the action on the colouring in Figure \ref{fig:dual_lat_with_labels}.}
\label{fig:dual_lat_with_labels_plaquette-action}
\end{figure}

\subsubsection*{The edge operators in the dual picture}
Let $t \in L^1$. {Cf \S \ref{sec:XYedge}.}
The edge operator $E_t^{\delta_x}$, for $x \in X$, is such that, if $Q$ is on the left-hand side of $t$ and $P$ is on the right-hand side of $t$, so we have an edge $P^* \xrightarrow{t^*} Q^*$ in the dual cell decomposition, {then}:
\[E_t^{\delta_x}( \R^*)=\begin{cases} \R, \textrm{ if } f\big((\R^*)_1(t^*) \big)\,x^{-1}\, (\R^*)_0(Q^*)=(\R^*)_0(P^*),\\
0, \textrm{ otherwise}.
\end{cases}
\]
So the edge projector {$E_t=E_t^{\delta_{1_X}}$} is
\begin{equation}
    E_t( \R^*)=\begin{cases} \R^*, \textrm{ if } f\big((\R^*)_1(t^*) \big)\, (\R^*)_0(Q^*)=(\R^*)_0(P^*),\\
0, \textrm{ otherwise}.
\end{cases}
\end{equation}
Looking at the dual colouring $\R^*$ in Figure  \ref{fig:dual_lat_with_labels}, the edge operators $E_{t_1}^{\delta_x}(\R)$ and $E_{t_2}^{\delta_{x'}}(\R)$ {satisfy:}
\[
 E_{t_1}^{\delta_x}(\R^*)=\begin{cases}
            \R^*, \textrm{ if }  xf(y_1)^{-1}x_1=x_2, \\
            0, \textrm{otherwise},
               \end{cases}
\textrm{ and }\quad
 E_{t_2}^{\delta_{x'}}(\R^*)=\begin{cases}
            \R^*, \textrm{ if }  x'f(y_2)^{-1}x_3=x_2,  \\
            0, \textrm{otherwise}.
               \end{cases}
\]
So the edge projectors $E_{t_1}=E_{t_1}^{\delta_{1_X}}$ and $E_{t_2}=E_{t_2}{\delta_{1_X}}$ are:
\[
 E_{t_1}(\R^*)=\begin{cases}
            \R^*, \textrm{ if }  f(y_1)^{-1}x_1=x_2, \\
            0, \textrm{otherwise},
               \end{cases}
\textrm{ and } \quad
 E_{t_2}(\R^*)=\begin{cases}
            \R^*, \textrm{ if }  f(y_2)^{-1}x_3=x_2,  \\
            0, \textrm{otherwise}.
               \end{cases}
\]
\begin{remark}\label{rem:in_act_grp}{Later on in \S\ref{sec:gsXY}, we will consider the action groupoid $X//Y$ of the action of $Y$ on $X$ such that $y\lact x:=f(y)x$; conventions for action groupoids are as in \S\ref{gsdgoxo}. Note that the edge projector on a dual edge coloured as $x \xrightarrow{y} x'$, where $x,x' \in X$ and $y \in Y$, hence chooses the configurations that make $x \xrightarrow{(x,y)} x'$ an arrow in the action groupoid $X//Y$. This will play a key role in determining the ground-state space of the $\ooXY$-model.}
\end{remark}
\subsubsection*{Vertex operators in the dual picture}
Let $\R^*=(\R^*_0,\R^*_1)$ be a dual colouring. Let $v \in L^1$ and $P$ be an adjacent plaquette. We resume the notation in \S \ref{sec:vertex_ops_XY}. The choice of $P$ equips the plaquette $v^* \in (L^*)^2$ of the dual cell decomposition with a base-point $P^* \in (L^*)^0$. Put:
\[\hol^{\R^*}_{\partial(v^*)} = \R^*_1(t_1^*)^{\theta_{v,t_1}} \dots \R_1^*(t_n)^{\theta_{v,t_n}}
.\]
Let  $w \in Y$, the vertex operator $V_{v,P}^{\delta_{w}}$ is such that:
\[V_{v,P}^{\delta_w}( \R^*)=\begin{cases} \R^*, \textrm{ if } 
{\hol^{\R^*}_{\partial(v^*)} = } \R^*_1(t_1^*)^{\theta_{v,t_1}} \dots \R_1^*(t_n)^{\theta_{v,1_n}} = w,
\\
0, \textrm{ otherwise}.
\end{cases}
\]
So the vertex projector is $V_v=V_{v,P}^{\delta_{1_Y}}$, {is such that}:
\begin{equation}\label{XYvertex_proj}
V_v( \R^*)=\begin{cases} \R^*, \textrm{ if } 
{\hol^{\R^*}_{\partial(v^*)}  }=\R^*_1(t_1^*)^{\theta_{v,t_1}} \dots \R_1^*(t_n)^{\theta_{v,t_n}} = 1_Y,
\\
0, \textrm{ otherwise}.
\end{cases}
\end{equation}
\noindent {Cf. \cite[\S 2.3]{balsam-kirillov} for the $X=1$ case.}

{If  $\R^*$ locally looks like Figure \ref{fig:dual_lat_with_labels}, the vertex operator $V_{v,P}^{\delta_w}$ and the vertex projector $V_v=V_{v,P}^{\delta_{1_Y}}$ are:}
\[ V_{v,P}^{\delta_w}(\R^*)=\begin{cases}
            \R^*, \textrm{ if }  y_1 \, y_2^{-1} \,  y_3\, y_4^{-1} \,  y_5=w,\\
            0, \textrm{otherwise},
               \end{cases} \textrm{ and } \quad V_{v}(\R^*)=\begin{cases}
            \R^*, \textrm{ if }  y_1 \, y_2^{-1} \,  y_3\, y_4^{-1} \,  y_5=1,\\ 0, \textrm{otherwise}.
               \end{cases} 
\]

\subsubsection{The ground-state space of the $\ooXY$-model and its relation with Quinn's {finite total homotopy} TQFT}\label{sec:gsXY}

As for the
$\GEXo$ model, we can prove that the ground-state space of the $\ooXY$-model on a surface $\Sigma$ with a triangulation $L$ is canonically isomorphic to the free vector space of homotopy classes of maps from the surface $\Sigma$ into a homotopy finite space $\mathcal{B}=B_{X // Y}.$ From this it follows that the ground-state space is canonically independent of the triangulation $L$ of $\Sigma$.

{Cf. Remark \ref{rem:in_act_grp}. Consider the action groupoid \S \ref{gsdgoxo} of the action of $Y$ on $X$ defined as
$y \lact x:=f(y) x.$}
\begin{lemma} There is a canonical isomorphism between the ground-state space of the $\ooXY$-model and the free vector space on the set of equivalence classes of groupoid functors $\pi_1(\Sigma, \Sigma_{L^*}^0) \to X//Y$, considered up to natural transformation. 
\end{lemma}
\noindent Here the notation $\Sigma_{L^*}^0$ means the set of vertices of $L^*$, i.e.\ $\Sigma_{L^*}^0$ is the 0-skeleton of the underlying CW-decomposition $L^*$ of $\Sigma$.
\begin{proof}
The result follows as in the proof of Lemma \ref{main_idea}.
First of all, the fundamental groupoid $\pi_1(\Sigma,\Sigma_{L^*}^0)$ is the free groupoid on the vertices and edges of the dual cell decomposition $L^*$, with one relation for each vertex $v \in L^0$ of $L$, i.e.\ plaquette $v^* \in (L^*)^2$ of the dual cell decomposition $L^*$; see \cite[\S 2.5]{companion}. The latter relation imposes that the composition of the arrows around $v^*$ is an identity, as in the formula for the vertex projector in Equation \eqref{XYvertex_proj}.

Consider the subspace of the total Hilbert space $\mathcal{H}_L$ obtained as:
\[ (\mathcal{H}_L)_0 := \big(\prod_{t \in L^1} E_t \prod_{v \in L^0} V_v \big) (\mathcal{H}_L)\]
Given the explicit formula of the vertex and edge projectors of the $\ooXY$-model in the dual picture, it follows that $(\mathcal{H}_L)_0$ has a  basis in one-to-one correspondence with the set of all groupoid functors $T\colon \pi_1(\Sigma, \Sigma_{L^*}^0) \to X//Y$. {(Remark \ref{rem:in_act_grp} is used here.)}

Each plaquette operator $F_P^g\colon \mathcal{H}_L \to \mathcal{H}_L$ restricts to a map $F_P^g\colon (\mathcal{H}_L)_0 \to (\mathcal{H}_L)_0$. This can be seen {from Equations \eqref{Eq:pl-edge}, \eqref{Eq:pl-vertex}  and \eqref{XYvertex_proj}.} Given that plaquette operators $F_P^g$ and $F_Q^h$, based at different plaquettes $P$ and $Q$ commute, {Lemma \ref{lem:plaquette-operators-commute},} we hence have a representation of $\prod_{P \in L^2} Y$ on $(\mathcal{H}_L)_0$ via the obvious product action by plaquette operators. 

The space $(\mathcal{H}_L)_0$ is the free vector space on all functors $T\colon \pi_1(\Sigma, \Sigma_{L^*}^0) \to X//Y$. Each element of $\prod_{P \in L^2} Y$ acts on a functor $T\colon \pi_1(\Sigma, \Sigma_{L^*}^0) \to X//Y$ by inducing a natural transformation of functors. This correspondence between actions by elements of $Y^{L^2}$ and natural transformations is easily seen to be one-to-one. (Again this is just about checking compatibility between the languages of plaquette operators and that of natural transformations between functors.)

Hence the ground-state space of the $\ooXY$-model, namely $\big(\prod_{P \in L^2} F_P)((\mathcal{H}_L)_0)$, is canonically isomorphic to the free vector space on the set of all groupoid functors $\pi_1(\Sigma, \Sigma_{L^*}^0) \to X//Y$, considered up to natural transformations.
\end{proof}

{As in \S \ref{gsdgoxo} and \cite[\S 5.2]{companion}, we now  use  \cite[Theorem A]{Brown_Higgins} (see also \cite[Theorem 7.16]{brown_hha}\cite[\S 11.4.iii]{Brown_Higgins_Sivera})} to find a canonical identification between the set of isomorphism classes of groupoid functors $\pi_1(\Sigma, \Sigma_{L^*}^0) \to X//Y$, considered up to natural transformation, and homotopy classes of maps $\Sigma \to B_{X//Y}$. 
\begin{theorem}Let $(\Sigma,L)$ be a surface with a triangulation {with a total order on the set of vertices.} Let $f\colon X \to Y$ be a map of finite groups. The ground state space of the $\ooXY$-model is canonically isomorphic to the free vector space on the set of homotopy classes of maps $\Sigma \to B_{X/ / Y}$, where $B_{X // Y}$ is the classifying space of the groupoid $X // Y$, and hence it is canonically triangulation independent. 

In particular there is a (2+1)-dimensional TQFT, namely Quinn's {finite total homotopy} TQFT $\mathscr{Q}_{B_{X//Y}},$ whose state space on a surface $\Sigma$ is given by the ground-state space of the $\ooXY$-model on $(\Sigma,L)$.
\end{theorem}
\begin{proof}
Again the first statement follows from \cite[Theorem A]{Brown_Higgins}. The second follows as in \S\ref{sec:relQuinn} from the construction of Quinn's {finite total homotopy} TQFT  \cite{Quinn} {and \cite[\S 4 and \S 8.2]{martins_porter21},} considering the space $B_{X // Y}$. {The latter is a 1-type, and moreover (as the groupoid $X//Y$ is finite) a homotopy finite space; see \cite{Brown_Higgins}.} 
\end{proof}

\begin{remark}As we can see from Subsubsection \ref{sec:XY_indual}, in the dual lattice decomposition to $(\Sigma,L)$, the $\ooXY$-model reduces to the Kitaev model based on $Y$ if $X=\{1\}$, \cite[Lemma 2.6]{balsam-kirillov}. In the general case the $\ooXY$-model gives a groupoid version of the Kitaev model based on the groupoid $X//Y$. We have additional edge operators, that do not exist in the original Kitaev model, that impose an energy penalty on configurations outside $X//Y$; see Remark \ref{rem:in_act_grp}. The case when $Y=\{1\}$ is just Potts model on $|X|$ {(see \cite{Martin_Potts}\cite[\S 1.1]{Fatimah}), considering again the dual lattice.} 
\end{remark}

\section{Appendix (full calculations)}

\subsection{{Proof of Proposition~\ref{pr:GEXY}: a class of examples of Hopf crossed modules}}\label{Proof:GEXY}

\begin{proof}{(Of Proposition~\ref{pr:GEXY})}
~ \\
\noindent \underline{$\Fun{Y}$ is a $\CC G$-module algebra and coalgebra}.

We check the module algebra conditions, where $\varphi, \psi \in \Fun{Y}$, and $g \in G$:
\begin{align*}
(g \lact (\varphi \psi))(y) &= (\varphi \psi)(g^{-1} \lact y) = \varphi(g^{-1} \lact y) \, \psi(g^{-1} \lact y)    = (g \lact \varphi)(y) \,(g \lact \psi)(y) \\
&= \big((g \lact \varphi)\, (g \lact \psi)\big)(y),
\end{align*}
and \[(g \lact 1_{\Fun{Y}})(y) = 1_{\Fun{Y}}(g^{-1} \lact y) = 1 = {\eps(g)} 1_{\Fun{Y}}(y).\]

We check the module coalgebra conditions, where $y,z \in Y$:
\begin{align*}
(\Delta(g \lact \varphi))(y,z) &= (g \lact \varphi)(yz) 
= \varphi(g^{-1} \lact (yz)) = \varphi((g^{-1} \lact y)(g^{-1} \lact z)) = \Delta(\varphi)(g^{-1} \lact y, g^{-1} \lact z) \\
&= \big((g \ot g)\lact\Delta(\varphi)\big)(y, z),
\end{align*}
and $\eps(g \lact \varphi) = (g \lact \varphi)(1) = \varphi(g^{-1} \lact 1) = \varphi(1) = \eps(\varphi)$.

Let us further explicitly verify the $\CC G$-linearity of the antipode: 
\[(g \lact S(\varphi))(y) = \varphi((g^{-1} \lact y)^{-1}) = \varphi(g^{-1} \lact y^{-1}) = S(g \lact \varphi)(y).\]
{Note that this also follows from Lemma~\ref{llem_Sprop}, together with the fact that $\CC G$ is cocommutative.}

\medskip

Since $\CC G$ is cocommutative (and hence the Yetter-Drinfeld condition \eqref{eq:yetter-drinfeld-condition-trivial} automatically holds), we can thus form the crossed product $\Fun{Y} \rtimes \CC G$, see Definition \ref{def:crossed-product-algebra}, which is a Hopf algebra with underlying vector space $\Fun{Y} \otimes \CC G$, and the following multiplication and comultiplication:
\begin{align*}
(\varphi \ot g)\cdot(\psi \ot h) &:= \varphi (g \lact \psi) \ot gh, \\
\Delta(\varphi \ot g) &:= (\varphi_{(1)} \ot g) \ot (\varphi_{(2)} \ot g).
\end{align*}

Next, consider the tensor product Hopf algebra $\Fun{X} \otimes \CC E$.

\medskip

\noindent\underline{$\Fun{X} \otimes \CC E$ is a $(\Fun{Y} \rtimes \CC G)$-module algebra and coalgebra}.

\noindent For $\varphi \ot g \in \Fun{Y} \ot \CC G$, $\xi \ot e \in \Fun{X} \ot \CC E$, {recall} \[(\varphi \ot g) \lact (\xi \ot e) = \varphi(1) (g \lact \xi) \ot (g \lact e),\] where $(g \lact \xi)(x) := \xi(g^{-1} \lact x)$
for $x \in X$.
This is indeed an $(\Fun{Y} \rtimes \CC G)$-action, because:
\begin{align*}
    (\varphi \ot g)\lact\big((\psi \ot h)\lact (\xi\ot e)\big) &= (\varphi \ot g)\lact (\psi(1) (h \lact \psi) \ot (h \lact e)) \\&= \varphi(1) \psi(1) (g \lact (h \lact \xi)) \ot (g \lact (h \lact e)) \\ &= \varphi(1) \psi(g^{-1} \lact 1) (g \lact (h \lact \xi)) \ot (g \lact (h \lact e)) \\ &= (\varphi (g \lact \psi) \ot g h) \lact (\xi \ot e)\\ &= \big( (\varphi \ot g)\cdot_\rtimes(\psi \ot h) \big) \lact (\xi \ot e).
    \end{align*}

Now we check the module algebra conditions, where $\varphi \ot g \in \Fun{Y} \ot \CC G$ and $\xi \ot e, \zeta \ot d \in \Fun{X} \ot \CC E$:
\begin{align*}
(\varphi \ot g) \lact \big( (\xi \ot e)(\zeta \ot d) \big) &= (\varphi \ot g) \lact (\xi \zeta \ot ed) \\
&= \varphi(1) (g \lact (\xi \zeta)) \ot g \lact (ed) \\
&= \varphi(1) (g \lact \xi) (g \lact \zeta) \ot (g \lact e)(g \lact d) \\
&= \varphi(1) \big( (g \lact \xi) \ot (g \lact e) \big) \big( (g \lact \zeta) \ot (g \lact d) \big) \\
&= \big(\varphi_{(1)}(1) (g \lact \xi) \ot (g \lact e)\big)\big(\varphi_{(2)}(1) (g \lact \zeta) \ot (g \lact d)\big) \\
&= \big( (\varphi_{(1)} \ot g) \lact (\xi \ot e)\big) \big( (\varphi_{(2)} \ot g) \lact (\zeta \ot d)\big),
\end{align*}
and:
\[(\varphi \ot g)\lact(1_{\Fun{X}} \ot 1_E) = \varphi(1) (g \lact 1_{\Fun{X}}) \ot (g \lact 1_E) = \varphi(1) (1_{\Fun{X}} \ot 1_E) = \eps(\varphi \ot g) 1_{\Fun{X}} \ot 1_E.\]

Next we check the module coalgebra conditions, where $\varphi \ot g \in \Fun{Y} \ot \CC G$ and $\xi \ot e \in \Fun{X} \ot \CC E$: 
\begin{align*}
\Delta((\varphi \ot g)\lact(\xi \ot e)) &= \Delta(\varphi(1) (g \lact \xi) \ot (g \lact e)) \\
&= \varphi(1) ((g \lact \xi)_{(1)} \ot (g \lact e)) \ot ((g \lact \xi)_{(2)} \ot (g \lact e)) \\
&= \varphi(1) ((g \lact \xi_{(1)}) \ot (g \lact e)) \ot ((g \lact \xi_{(2)}) \ot (g \lact e)) \\
&= \varphi_{(1)}(1) ((g \lact \xi_{(1)}) \ot (g \lact e)) \ot \varphi_{(2)}(1) ((g \lact \xi_{(2)}) \ot (g \lact e)) \\
&= ((\varphi_{(1)} \ot g) \ot (\varphi_{(2)} \ot g)) \lact \Delta(\xi \ot e)
\end{align*}
and \[\eps((\varphi \ot g) \lact (\xi \ot e)) = \varphi(1) \xi(g^{-1} \lact 1) \eps(g \lact e) = \varphi(1) \xi(1) = \eps(\varphi \ot g) \eps(\xi \ot e).\]

\medskip

\noindent\underline{$\Fun{X} \otimes \CC E$ satisfies the Yetter-Drinfeld condition %(with trivial $\Fun{Y} \rtimes \CC G$-coaction), 
{i.e. \eqref{eq:yetter-drinfeld-condition-trivial}.}}

For this we have to check that \[(\varphi_{(1)} \ot g) \ot (\varphi_{(2)} \ot g) \lact (\xi \ot e) = (\varphi_{(2)} \ot g) \ot (\varphi_{(1)} \ot g) \lact (\xi \ot e)\] holds for $\varphi \ot g \in \Fun{Y} \ot \CC G$ and $\xi \ot e \in \Fun{X} \ot \CC E$.
Indeed, it is easy to see that both sides of the equation are equal to $(\varphi \ot g) \ot (g \lact \xi \ot g \lact e)$. \\

\noindent\underline{There is a Hopf algebra morphism $\partial : \Fun{X} \otimes \CC E \to \Fun{Y} \rtimes \CC G$, 
{where}}
\begin{align*}
\partial : \Fun{X} \otimes \CC E &\lto \Fun{Y} \rtimes \CC G, \\
\xi \ot e &\lmapsto f^* \xi \ot \partial(e).
\end{align*} {Let us check the above is a} morphism of algebras. We have:
\begin{align*}
\partial\big((\xi \ot e)(\zeta \ot d)\big) &= \partial(\xi \zeta \ot ed) = f^*(\xi \zeta) \ot \partial(ed) 
= (f^*\xi) (f^*\zeta) \ot \partial(e) \partial(d) \\
&= (f^*\xi \ot \partial(e)\big) \big( f^*\zeta \ot \partial(d) \big) = \partial(\xi \ot e) \, \partial(\zeta \ot d),
\end{align*} where, in the penultimate step, we used the fact that $\partial(E)$ acts trivially on $X$.
Also \[\partial(1_{\Fun{X}} \ot 1_E) = f^*(1_{\Fun{X}}) \ot \partial(1_E) = 1_{\Fun{{Y}}} \ot 1_E.\]

Next we check that $\partial$ is a morphism of coalgebras:
\begin{align*}
\Delta(\partial(\xi \ot e)) &= ((f^* \xi)_{(1)} \ot \partial(e)) \ot ((f^* \xi)_{(2)} \ot \partial(e)) = ((\xi \circ f)_{(1)} \ot \partial(e)) \ot ((\xi \circ f)_{(2)} \ot \partial(e)) \\
&\stackrel{\text{$f$ group hom.}}{=} ((\xi_{(1)}\circ f) \ot \partial(e)) \ot ((\xi_{(2)}\circ f) \ot \partial(e)) 
= \partial(\xi_{(1)} \ot e) \ot \partial(\xi_{(2)} \ot e)
\end{align*}
and \[\eps(\partial(\xi \ot e)) = \xi(f(1_Y)) \eps(\partial(e)) = \xi(1_X) \eps(e) = \eps(\xi \ot e).\]

\medskip

\noindent \underline{$\partial : \Fun{X}\otimes\CC E \lto \Fun{Y}\rtimes \CC G$ satisfies the Peiffer conditions.}

 We check the first Peiffer condition.
On the one hand we have, for $y \in Y$,
\begin{align*}
\partial((\varphi \ot g)\lact(\xi \ot e))(y) &= \varphi(1) \partial(g \lact \xi) \ot g \lact e){(y)} \\
&= \varphi(1) \xi(g^{-1} \lact f(y)) \partial(g \lact e) \\
&= \varphi(1) \xi(f(g^{-1} \lact y)) \partial(g \lact e).
\end{align*}
On the other hand we have
\begin{align*}
\big((\varphi_{(1)} \ot g)\partial(\xi \ot e)S(\varphi_{(2)} \ot g)\big)(y) &= ((\varphi_{(1)} \ot g)(f^*\xi \ot \partial(e))(1 \ot g^{-1})(S(\varphi_{(2)}) \ot 1))(y) \\
%&= ((\varphi_{(1)}(g \lact f^*\xi) \ot g\partial(e))(1 \ot g^{-1})(S(\varphi_{(2)}) \ot 1))(y) \\
&= ((\varphi_{(1)}(g \lact f^*\xi) \ot g\partial(e))(1 \ot g^{-1})(S(\varphi_{(2)}) \ot 1))(y) \\
&= ((\varphi_{(1)}(g \lact f^*\xi) \ot g\partial(e)g^{-1})(S(\varphi_{(2)}) \ot 1))(y) \\
&= ((\varphi_{(1)}(g \lact f^*\xi) \ot \partial(g \lact e))(S(\varphi_{(2)}) \ot 1))(y) \\
&= (\varphi_{(1)}(g \lact f^*\xi)(\partial(g \lact e) \lact S(\varphi_{(2)})) \ot \partial(g \lact e))(y) \\
&= \varphi_{(1)}(y) \,\xi(f(g^{-1}\lact y)) \varphi_{(2)}((\partial(g \lact e)^{-1} \lact y)^{-1}) \partial(g \lact e) \\
&\stackrel{(*)}{=} \varphi_{(1)}(y) \,\xi(f(g^{-1}\lact y)) \varphi_{(2)}(y^{-1}) \partial(g \lact e) \\
&= \varphi(1) \xi(f(g^{-1}\lact y)) \partial(g \lact e),
\end{align*}
where in $(*)$ we use the assumption that $\im(\partial)$ acts trivially on $Y$.
In fact, this is the only place where this condition is used. 

It is only left to check the second Peiffer condition: On the one hand we have, for $x \in X$,
\begin{align*}
\big(\partial(\xi \ot e) \lact (\zeta \ot d)\big)(x) &= \big((f^*\xi \ot \partial(e)) \lact (\zeta \ot d)\big)(x) \\
&= \xi(f(1)) \big(\partial(e) \lact \zeta \ot \partial(e) \lact d\big)(x) \\
&= \xi(1) \zeta(\partial(e)^{-1} \lact x) \partial(e) \lact d  \\
&\stackrel{(*)}{=} \xi(1) \zeta(x) \partial(e) \lact d ,
\end{align*}
where in $(*)$ we use the assumption that $\im(\partial)$ acts trivially on $X$.
On the other hand we have
\begin{align*}
\big((\xi_{(1)} \ot e) (\zeta \ot d) S(\xi_{(2)} \ot e)\big)(x)
&= \big((\xi_{(1)} \ot e) (\zeta \ot d) (1 \ot e^{-1}) (S(\xi_{(2)}) \ot 1)\big)(x) \\
&= (\xi_{(1)} \zeta S(\xi_{(2)}))(x) ede^{-1} \\
&= (\xi_{(1)} \zeta S(\xi_{(2)}))(x) \partial(e) \lact d \\
&= \xi(x x^{-1}) \zeta(x) \partial(e)\lact d \\
&= \xi(1) \zeta(x) \partial(e)\lact d .
\end{align*}
\end{proof}

\subsection{{Commutativity of edge orientation reversals and base-point shifts with vertex, edge and plaquette operators: some calculations}} \label{sec:proofs-base-point-shifts-orientation-reversals}
{In this subsection, we sketch the rest of the proof of Proposition \ref{prop: commut-vertex-ops-shifts_etc}.}

{Until the end of this paper, fix a compact oriented surface $\Sigma$,  with a cell decomposition $L$, as in Definition \ref{def:cell-decomposition}, and a crossed module of Hopf algebras $(A \xrightarrow{\partial} H,\lact)$.}

\subsubsection{Base-point shifts commute with vertex operators}\label{subs:basepoint-shifs and vertex ops}
\begin{lemma}
Let $P \in L^2$ be any plaquette and let $v \in L^0$ be any vertex with adjacent plaquette $Q$.
Denote by $(V')_{v,Q}^h$ for any $h \in H$ the vertex operator for the vertex $v$ with respect to the cell decomposition $L'$, obtained from the given one $L$ by shifting the base-point of $P$ once in counterclockwise direction around $P$.
Then:
\[ (V')_{v,Q}^h \circ T_P^+ = T_P^+ \circ V_{v,Q}^h, \quad\text{ for all $h \in H$}. \]
The analogous result holds for negative base-point shifts.
\end{lemma}
\begin{proof}{(Sketch.)}
First note that this clearly holds when $v$ is not in the boundary of $P$, because then $(V')_{v,Q}^h = V_{v,Q}^h$, and $V_{v,Q}^h$ and $T_P^+$ have disjoint support. {So suppose that $v$ is in the boundary of $P$.}
\begin{figure}[ht!]
 \labellist
\pinlabel $\small{v}$ at 402 171
\pinlabel $v=v_{P_1}$ at 890 301
\pinlabel $v=v_{P_2}$ at 890 250
\pinlabel $\small{P_2=P}$ at 125 195
\pinlabel $\small{P_1}$ at 635 195
\pinlabel $\small{Q}$ at 360 390
\pinlabel $\small{e_4}$ at 472 320
\pinlabel $\small{e_1}$ at 293 320
\pinlabel $\small{e_2}$ at 237 106
\pinlabel $\small{e_3}$ at 498 106
\endlabellist
\centering
\includegraphics[scale=0.2]{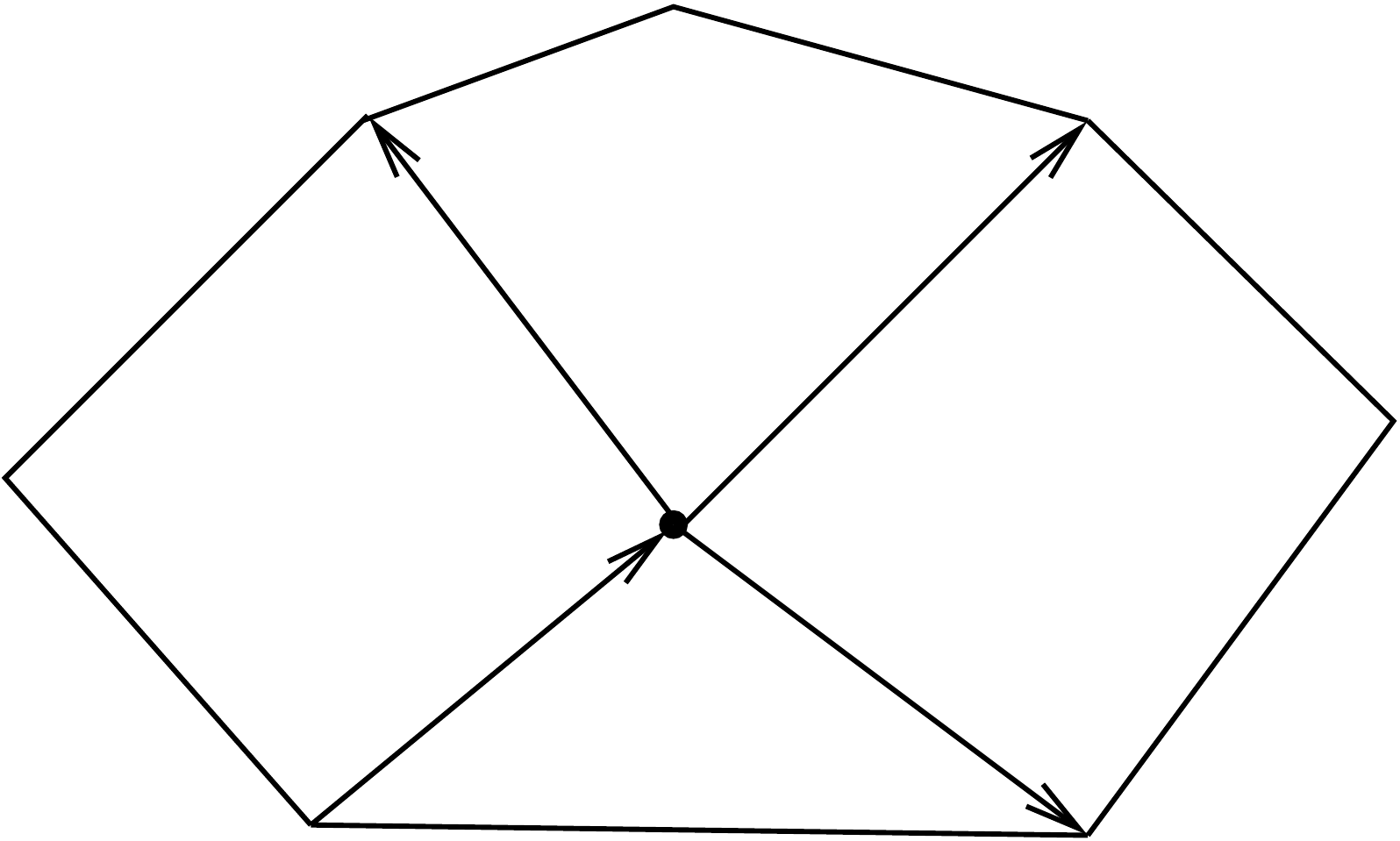}
\caption{Example case that vertex operators commute with base-point shifts: $(V')_{v,Q}^h \circ T_P^+ = T_P^+ \circ V_{v,Q}^h $ and $(V')_{v,Q}^h \circ T_P^-= T_P^- \circ V_{v,Q}^h$. The plaquettes $P_1$ and $P=P_2$ are based at $v$.}
\label{fig:base_shift}
\end{figure}

{Let us prove only that $ (V')_{v,Q}^h \circ T_P^\pm = T_P^\pm \circ V_{v,Q}^h$ for the configuration in Figure \ref{fig:base_shift}. This is representative of the type of calculations required, and addresses all possible orientations of the edge along which we shift the base-point.}

To prove that $ (V')_{v,Q}^h \circ T_P^+ = T_P^+ \circ V_{v,Q}^h$, we have, applying $T_P^+$ after applying the vertex operator:
\begin{align*}
v_1 & \ot v_2 \ot v_3  \ot v_4 \ot X_1 \ot X_2  \\& \stackrel{ V_{v,Q}^h}{\lmapsto} 
h_{(1)} v_1  \ot  v_2 S\big (h_{(2)}\big) \ot h_{(3)} v_3 \ot h_{(4)} v_4 \ot h_{(5)}\lact  X_1\ot h_{(6)}\lact X_2   \\& \stackrel{ T^+_P}{\lmapsto} 
(h_{(1)} v_1)_{(1)}  \ot  v_2 S\big (h_{(2)}\big) \ot h_{(3)} v_3 \ot h_{(4)} v_4 \ot h_{(5)}\lact  X_1\ot S\big( (h_{(1)} v_1)_{(2)}\big)  h_{(6)}\lact X_2 \\
&\stackrel{\text{coassoc.}}{=}h_{(1)} (v_1)_{(1)}  \ot  v_2 S\big (h_{(3)}\big) \ot h_{(4)} v_3 \ot h_{(5)} v_4 \ot h_{(6)}\lact  X_1\ot S\big( h_{(2)} (v_1)_{(2)}\big)  h_{(7)}\lact X_2 \\
&=h_{(1)} (v_1)_{(1)}  \ot  v_2 S\big (h_{(3)}\big) \ot h_{(4)} v_3 \ot h_{(5)} v_4 \ot h_{(6)}\lact  X_1\ot S\big( (v_1)_{(2)}\big) S\big(  h_{(2)}\big)  h_{(7)}\lact X_2 \\
&\stackrel{\text{Lemma \ref{lem:yetter-drinfeld-for-many-factors}}}{=} h_{(1)} (v_1)_{(1)}  \ot  v_2 S\big (h_{(4)}\big) \ot h_{(5)} v_3 \ot h_{(6)} v_4 \ot h_{(7)}\lact  X_1\ot S\big( (v_1)_{(2)}\big) S\big(  h_{(2)}\big)  h_{(3)}\lact X_2 \\
&= h_{(1)} (v_1)_{(1)}  \ot  v_2 S\big (h_{(2)}\big) \ot h_{(3)} v_3 \ot h_{(4)} v_4 \ot h_{(5)}\lact  X_1\ot S\big( (v_1)_{(2)}\big) \lact X_2 .
\end{align*}
On the other hand, applying the base-point shift before applying the  vertex operator:
\begin{align*}
v_1 \ot v_2 \ot v_3 & \ot v_4 \ot X_1 \ot X_2  \stackrel{T^+_P}{\lmapsto} 
(v_1)_{(1)} \ot v_2 \ot v_3  \ot v_4 \ot X_1 \ot  S\big((v_2)_{(2)}\big)\lact X_2 
\\& \stackrel{ (V')_{v,Q}^h}{\lmapsto} 
h_{(1)}(v_1)_{(1)} \ot v_2 S\big (h_{(2)}\big) \ot h_{(3)} v_3  \ot h_{(4)} v_4 \ot h_{(5)}X_1 \ot  S\big((v_1)_{(2)}\big)\lact X_2.
\end{align*}
Note that the same argument would work if we had chosen a different plaquette $Q$ to give us the convention for the vertex operator: this only affects which edge we consider to be our initial edge, $e_1$, and we would still be able to apply Lemma \ref{lem:yetter-drinfeld-for-many-factors}.

For $ (V')_{v,Q}^h \circ T_P^-= T_P^- \circ V_{v,Q}^h$, applying   $T_P^-$ after applying the vertex operator, we get
\begin{align*}
v_1 & \ot v_2 \ot v_3  \ot v_4 \ot X_1 \ot X_2  \\& \stackrel{V_{v,Q}^h}{\lmapsto} 
h_{(1)} v_1  \ot  v_2 S\big ( h_{(2)}\big) \ot h_{(3)} v_3 \ot h_{(4)} v_4 \ot h_{(5)}\lact  X_1\ot h_{(6)}\lact X_2   \\
&\stackrel{T_P^-}{\lmapsto} h_{(1)} v_1  \ot  \big ( v_2 S\big (h_{(2)}\big)\big)_{(1)} \ot h_{(3)} v_3 \ot h_{(4)} v_4 \ot h_{(5)}\lact  X_1\ot  \big ( v_2 S\big ( h_{(2)}\big)\big)_{(2)} (h)_{(6)}\lact X_2\\ &= 
h_{(1)} v_1  \ot  ( v_2)_{(1)} S\big (h_{(3)}\big)\ot h_{(4)} v_3 \ot h_{(5)} v_4 \ot h_{(6)}\lact  X_1\ot  ( v_2)_{(2)} S\big (h_{(2)}\big)  h_{(6)}\lact X_2 \\
& \stackrel{\text{Lemma \ref{lem:yetter-drinfeld-for-many-factors}}}{=}
h_{(1)} v_1  \ot  ( v_2)_{(1)} S\big (h_{(4)}\big)\ot h_{(5)} v_3 \ot h_{(6)} v_4 \ot h_{(7)}\lact  X_1\ot  ( v_2)_{(2)} S\big (h_{(2)}\big)  h_{(3)}\lact X_2 \\
& \mapsto 
h_{(1)} v_1  \ot  ( v_2)_{(1)} S\big (h_{(2)}\big)\ot h_{(3)} v_3 \ot h_{(4)} v_4 \ot h_{(5)}\lact  X_1\ot  ( v_2)_{(2)} \lact X_2.
\end{align*}
This exactly what we get if we apply first the negative base-point shift, and then the vertex operator. Again, the argument would still work if we have chosen a different plaquette $Q$ to give us the convention for the vertex operator. This choice only affects which edge we consider to be our initial edge, $e_1$.
\end{proof}

\subsubsection{{Commutation relations between edge operators and base-point shift operators}}\label{sec: commu_edge_shift}

We show now that base-point shifts commute with edge operators if the edge is not the one along which we shift the base-point.

\begin{lemma}\label{lem:shiftsvsedges}
Let $e \in L^1$ be an edge and let $P \in L^2$ be a plaquette adjacent to the edge $e$. 
We further require that the base-point of $P$ is not equal to the starting vertex of $e$ if $e$ is oriented counterclockwise around $P$, and not equal to the target vertex of $e$ if $e$ is oriented clockwise around $P$.
Then \[ T_P^+ \circ E_e^a = {E'}_e^a \circ T_P^+ \quad\text{ for all } a \in A , \] 
where ${E'}_e^a$ 
denotes the edge operator of the edge $e$ for the cell decomposition $L'$ obtained from the given one $L$ by shifting the base-point of $P$ once in counterclockwise direction. The analogous result holds for negative base-point shifts.
\end{lemma}
\begin{proof}
We assume for the proof that $P$ is the plaquette on the left of the edge $e$ (with respect to the orientations of $e$ and of the underlying surface $\Sigma$), so that $e$ is oriented in counterclockwise direction around $P$. {So the configuration is analogous to that in Figure \ref{fig:edge}.}
The calculations for the case that $e$ is oriented clockwise with respect to $P$ can be done analogously.

Let us consider the edge operator $E_e^a$, $a \in A$, for the edge $e$ and the base-point shift operator $T_P^+$ for the plaquette $P$.
The condition that the base-point of $P$ and the starting vertex of $e$ differ is just to say that we shift the base-point along an edge different from the edge $e$.

Denote by $(e_1, \dots, e_\ell)$ the (non-empty) set of edges in the boundary of $P$ connecting the base-point with the starting vertex of $e$ in counterclockwise order.
By applying edge orientation reversals where necessary, since these commute with the plaquette operator and base-point shifts arising here (by Proposition \ref{prop: commut-vertex-ops-shifts_etc}), we may assume that the edges $(e_1, \dots, e_\ell)$ are oriented counterclockwise.
Similarly, let $Q\in L^2$ be the plaquette on the right of  $e$, and let $(d_1, \dots, d_r)$ be as in Definition \ref{def:edge-operator}.
Using edge orientation reversals, we may assume that any edge in $(d_1, \dots, d_r)$ is oriented clockwise around $Q$.

Using Definition \ref{def:edge-operator} for the edge operator we calculate:
\begin{align*}
&(T_P^+ \circ E_e^a)(v_{e_1} \ot v_{e_2} \ot\cdots\ot v_{e_\ell} \ot v_e \ot v_{d_1} \ot\cdots\ot v_{d_r} \ot X_P \ot X_Q)
& \\
&= T_P^+ \Big( (v_{e_1})_{(1)} \ot (v_{e_2})_{(1)} \ot \cdots\ot (v_{e_\ell})_{(1)}  \ot \partial a_{(3)} v_e \ot (v_{d_1})_{(2)} \ot\cdots\ot (v_{d_r})_{(2)} \\
&\phantom{xxxxxxxxxxx} \ot \big( (v_{e_1})_{(2)} (v_{e_2})_{(2)}  \cdots (v_{e_\ell})_{(2)} \lact a_{(1)} \big) X_P \ot X_Q \big( (v_{d_1})_{(1)} \cdots (v_{d_r})_{(1)} \lact S( a_{(2)} )\big) \Big) 
& \\
&= (v_{e_1})_{(1,1)} \ot (v_{e_2})_{(1)} \ot \cdots\ot (v_{e_\ell})_{(1)}  \ot \partial a_{(3)} v_e \ot (v_{d_1})_{(2)} \ot\cdots\ot (v_{d_r})_{(2)} \\
&\phantom{xxxxxxxxxxx} \ot S\big ( (v_{e_1})_{(1,2)}\big )\lact \Big (\big( (v_{e_1})_{(2)} (v_{e_2})_{(2)} \cdots (v_{e_\ell})_{(2)} \lact a_{(1)} \big) X_P \Big)\\
&\phantom{xxxxxxxxxxx} \ot X_Q \big( (v_{d_1})_{(1)} \cdots (v_{d_r})_{(1)} \lact S( a_{(2)} )\big) 
& \\
&\stackrel{\text{$A$ mod. alg.}}{=} (v_{e_1})_{(1,1)} \ot (v_{e_2})_{(1)} \ot \cdots\ot (v_{e_\ell})_{(1)}  \ot \partial a_{(3)} v_e \ot (v_{d_1})_{(2)} \ot\cdots\ot (v_{d_r})_{(2)} \\
&\phantom{xxxxxxxxxxx} \ot S\big ( (v_{e_1})_{(1,2)}\big )_{(1)}\lact \Big (\big( (v_{e_1})_{(2)} (v_{e_2})_{(2)}  \cdots (v_{e_\ell})_{(2)} \lact a_{(1)} \big)\Big) \,\,\, S\big ( (v_{e_1})_{(1,2)}\big )_{(2)}\lact  X_P \\
&\phantom{xxxxxxxxxxx} \ot X_Q \big( (v_{d_1})_{(1)} \cdots (v_{d_r})_{(1)} \lact S (a_{(2)}) \big)  
& \\
&{=} (v_{e_1})_{(1,1)} \ot (v_{e_2})_{(1)} \ot \cdots\ot (v_{e_\ell})_{(1)}  \ot \partial a_{(3)} v_e \ot (v_{d_1})_{(2)} \ot\cdots\ot (v_{d_r})_{(2)} \\
&\phantom{xxxxxxxxxxx} \ot S\big ( (v_{e_1})_{(1,2,2)}\big )\lact \Big (\big( (v_{e_1})_{(2)} (v_{e_2})_{(2)} \cdots (v_{e_\ell})_{(2)} \lact a_{(1)} \big)\Big) \,\,\, S\big ( (v_{e_1})_{(1,2,1)}\big )\lact  X_P \\
&\phantom{xxxxxxxxxxx} \ot X_Q \big( (v_{d_1})_{(1)} \cdots (v_{d_r})_{(1)} \lact S (a_{(2)}) \big) \\
&{=} (v_{e_1})_{(1)} \ot (v_{e_2})_{(1)} \ot
 \cdots\ot (v_{e_\ell})_{(1)}  \ot \partial a_{(3)} v_e \ot (v_{d_1})_{(2)} \ot\cdots\ot (v_{d_r})_{(2)} \\
&\phantom{xxxxxxxxxxx} \ot S\big ( (v_{e_1})_{(3)}\big )\lact \Big (\big( (v_{e_1})_{(4)} (v_{e_2})_{(2)} \cdots (v_{e_\ell})_{(2)} \lact a_{(1)} \big)\Big) \,\,\, S\big ( (v_{e_1})_{(2)}\big )\lact  X_P \\
&\phantom{xxxxxxxxxxx} \ot X_Q \big( (v_{d_1})_{(1)} \cdots (v_{d_r})_{(1)} \lact S (a_{(2)}) \big) \\
&{=} (v_{e_1})_{(1)} \ot (v_{e_2})_{(1)}\ot  \cdots\ot (v_{e_\ell})_{(1)}  \ot \partial a_{(3)} v_e \ot (v_{d_1})_{(2)} \ot\cdots\ot (v_{d_r})_{(2)} \\
&\phantom{xxxxxxxxxxx} \ot \big( (v_{e_2})_{(2)} \cdots (v_{e_\ell})_{(2)} \lact a_{(1)} \big) \,\,\, S\big ( (v_{e_1})_{(2)}\big )\lact  X_P \\
&\phantom{xxxxxxxxxxx} \ot X_Q \big( (v_{d_1})_{(1)} \cdots (v_{d_r})_{(1)} \lact S (a_{(2)}) \big) \\
&=( {E'}_e^a \circ T_P^+ ) \big(v_{e_1} \ot v_{e_2} \ot \cdots\ot v_{e_\ell} \ot v_e \ot v_{d_1} \ot\cdots\ot v_{d_r} \ot X_P \ot X_Q\big).
\end{align*}
The calculation proving the commutativity for the case where $P$ lies to the right of $e$ is analogous.
\end{proof}

\subsubsection{Base-point shifts commute with {plaquette} operators {(for cocommutative elements)}}\label{sec:shift-plaquettes}

\begin{lemma} \label{lem:base-point-shifts-and-plaquette-operators}
Let $P$ and $Q$ be any two (not necessarily distinct) plaquettes.
Denote by $(F')^\varphi_P$ for any $\varphi\in H^*$ the plaquette operator for the plaquette $P$ with respect to the cell decomposition $L'$ obtained from $L$ by shifting the base-point of $Q$ once in counterclockwise direction around $Q$.
Then: \[ (F')^\varphi_P \circ T_Q^+ = T_Q^+ \circ F_P^\varphi, \quad\text{ for all cocommutative $\varphi\in H^*$} .\]
The analogous results holds for negative base-points shifts.
\end{lemma}
\noindent {In particular, putting $\varphi=\lambda, $ the Haar integral of $H^*$, the  equation above holds for the respective plaquette projectors at $P$.}
\begin{proof}
If $P$ and $Q$ are distinct and not adjacent, then the two maps have disjoint support in $\mathcal{H}_L$ and hence commute with each other.

\begin{figure}[ht!]
 \labellist
\pinlabel $\small{v_P}$ at -22 80
\pinlabel $P$ at 231 100
\pinlabel $e_1$ at 88 57
\pinlabel $e_2$ at 224 25
\pinlabel $e$ at 395 75
\pinlabel $Q$ at 510 60
\pinlabel $e_4$ at 406 169
\pinlabel $e_5$ at 112 158
\pinlabel $\small{v_Q}$ at 488 141
\endlabellist
\centering
\includegraphics[scale=0.3]{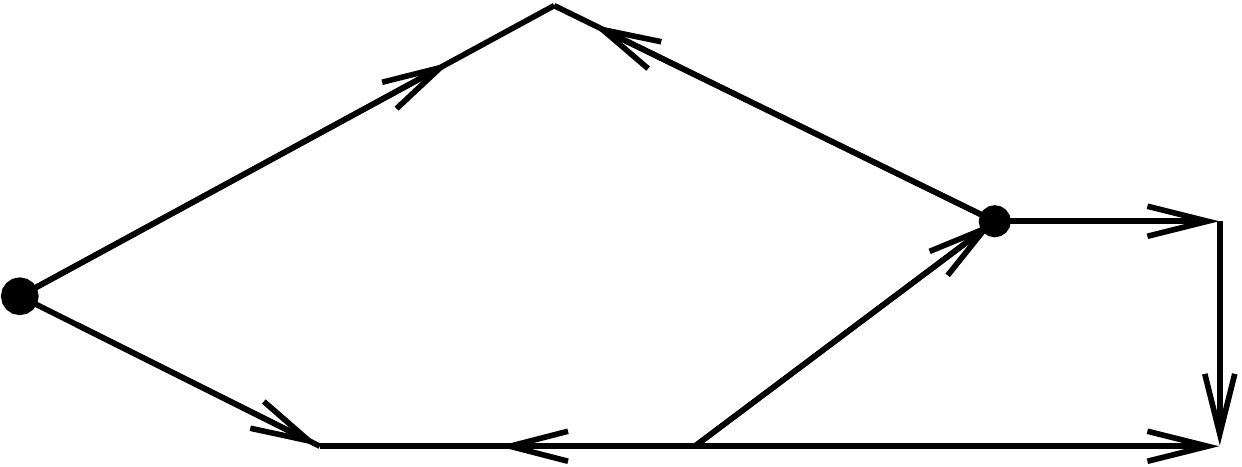}
\caption{{Two particular cases of commutativity between plaquette operators and plaquette base-point shifts. We will prove that $T^+_Q\circ F_P^\varphi= (F')_P^\varphi \circ T^+_Q$ and that $T^+_P \circ F_P^\varphi=(F')_P^\varphi\circ T^+_P$.}}
\label{fig:Plaquette-shits}
\end{figure}

Let $P$ and $Q$ be distinct but adjacent to each other and let $e \in L^1$ be the edge shared by the boundaries of both plaquettes and assume that the base-point shift $T_Q^+$ shifts the base-point of $Q$ along $e$. {(This is as in Figure \ref{fig:Plaquette-shits}.)}
Then the base-point shift $T_Q^+$ and any plaquette operator $F_P^\varphi$, $\varphi \in H^*$, for the adjacent plaquette $P$, commute with each other.
Indeed, the two operators affect the tensor factor $H$ in $\mathcal{H}_L$ associated with the edge $e$ only by applying co-multiplication and then acting with the additional resulting tensor factor of $H$ on a copy of $A$ associated with the two different plaquettes $P$ and $Q$.
Due to Lemma \ref{lem:yetter-drinfeld-for-many-factors}, these two operators thus commute with each other.

Now assume that $P=Q$ and denote by $(e_1, \dots, e_n)$ the edges in the boundary of the plaquette $P$ in counterclockwise order starting at the base-point of $P$.
Then for the transformed cell decomposition $L'$ the edges in the boundary of the plaquette $P$ in counterclockwise order starting at the base-point of $P$ are $(e_2, \dots, e_n, e_1)$.
By applying edge orientation reversals where necessary, we may assume that the edges $(e_1, \dots, e_n)$ are all oriented in counterclockwise direction around $P$.

Let $v_{e_1} \ot \cdots \ot v_{e_n} \in H^{\ot n}$ and $X \in A$.
We calculate, on the one hand:
\begin{align*}
&(T^+_P \circ F_P^\varphi)(v_{e_1} \ot \cdots \ot v_{e_n} \ot X) 
& \\
&\phantom{xxx}= T^+_P\big( (v_{e_1})_{(1)} \ot\cdots\ot (v_{e_n})_{(1)} \ot X_{(2)} \big)\ \varphi\left( (v_{e_1})_{(2)} \cdots (v_{e_n})_{(2)}\, ( S\circ \partial) (X_{(1)}) \right) 
& \\
&\phantom{xxx}= (v_{e_1})_{(1,1)} \ot (v_{e_2})_{(1)} \ot\cdots\ot (v_{e_n})_{(1)} \ot S\big( (v_{e_1})_{(1,2)}\big) \lact X_{(2)}\,\, \varphi\left( (v_{e_1})_{(2)} \cdots (v_{e_n})_{(2)}\, (S\circ \partial)( X_{(1)}) \right)\\
&\phantom{xxx}= (v_{e_1})_{(1)} \ot (v_{e_2})_{(1)} \ot\cdots\ot (v_{e_n})_{(1)} \ot S\big( (v_{e_1})_{(2)}\big) \lact X_{(2)} \,\,\varphi\left( (v_{e_1})_{(3)} \cdots (v_{e_n})_{(2)}\, (S\circ \partial)( X_{(1)}) \right).
\end{align*}
On the other hand, where we apply that $S$ is an anti-coalgebra morphism without comment: 
\begin{align*}
&((F')_P^\varphi \circ T^+_P)(v_{e_1} \ot \cdots \ot v_{e_n} \ot X) 
& \\
&\phantom{xxx}=
F_{P'}^\varphi \Big( (v_{e_1})_{(1)} \ot \cdots \ot v_{e_n} \ot S\big((v_{e_1})_{(2)}\big) \lact X \Big) 
& \\
&\phantom{xxx}= (v_{e_1})_{(1,1)} \ot (v_{e_2})_{(1)} \ot\cdots\ot (v_{e_n})_{(1)} \ot  \big( S \big ( (v_{e_1})_{(2)}\big) \lact X\big)_{(2)} \\
&\phantom{xxxxxxxxx} \varphi\left( (v_{e_2})_{(2)} \cdots (v_{e_n})_{(2)} (v_{e_1})_{(1,2)} \,\,   (S \circ \partial)\Big ( \big( S \big ( (v_{e_1})_{(2)}\big) \lact X\big)_{(1)} \Big) \right) 
& \\
&\phantom{xxx}\stackrel{\text{$A$ mod. coalg.}}{=}(v_{e_1})_{(1,1)} \ot (v_{e_2})_{(1)} \ot\cdots\ot (v_{e_n})_{(1)} \ot  \big( S \big ( (v_{e_1})_{(2,1)}\big) \lact X_{(2)} \\
&\phantom{xxxxxxxxx} \varphi\Big( (v_{e_2})_{(2)} \cdots (v_{e_n})_{(2)} (v_{e_1})_{(1,2)} \,\,   (S \circ \partial)\big (  S \big ( (v_{e_1})_{(2,2)}\big) \lact X_{(1)} \big) \Big)\\
&\phantom{xxx}\stackrel{\text{Pf 1.}}{=}(v_{e_1})_{(1,1)} \ot (v_{e_2})_{(1)} \ot\cdots\ot (v_{e_n})_{(1)} \ot  \big( S \big ( (v_{e_1})_{(2,1)}\big) \lact X_{(2)} \\
&\phantom{xxxxxxxxx} \varphi\Big( (v_{e_2})_{(2)} \cdots (v_{e_n})_{(2)}\, (v_{e_1})_{(1,2)} \, S((v_{e_1})_{(2,2,1})) (S \circ \partial)\big ( X_{(1)}\big)  \,\, (v_{e_1})_{(2,2,2)} \Big)
\\
&\phantom{xxx}%\stackrel{\text{Pf 1.}}
{=}(v_{e_1})_{(1)} \ot (v_{e_2})_{(1)} \ot\cdots\ot (v_{e_n})_{(1)} \ot  \big( S \big ( (v_{e_1})_{(3)}\big) \lact X_{(2)} \\
&\phantom{xxxxxxxxx} \varphi\Big( (v_{e_2})_{(2)} \cdots (v_{e_n})_{(2)} (v_{e_1})_{(2)} \,  S((v_{e_1})_{(4)}) \, (S \circ \partial)\big ( X_{(1)}\big)  \, (v_{e_1})_{(5)}) \Big)
\\
&\phantom{xxx}\stackrel{\text{Y-D.}}
{=}(v_{e_1})_{(1)} \ot (v_{e_2})_{(1)} \ot\cdots\ot (v_{e_n})_{(1)} \ot  \big( S \big ( (v_{e_1})_{(4)}\big) \lact X_{(2)} \\
&\phantom{xxxxxxxxx} \varphi\Big( (v_{e_2})_{(2)} \cdots (v_{e_n})_{(2)} (v_{e_1})_{(2)} \,\,  S((v_{e_1})_{(3)})(S \circ \partial)\big ( X_{(1)}\big)  \,\, (v_{e_1})_{(5)} \Big)
\\
&\phantom{xxx}%\stackrel{\text{Y-D.}}
{=}(v_{e_1})_{(1)} \ot (v_{e_2})_{(1)} \ot\cdots\ot (v_{e_n})_{(1)} \ot  \big( S \big ( (v_{e_1})_{(2)}\big) \lact X_{(2)} \\
&\phantom{xxxxxxxxx} \varphi\Big( (v_{e_2})_{(2)} \cdots (v_{e_n})_{(2)} \, (S \circ \partial)\big ( X_{(1)}\big)  \, (v_{e_1})_{(3}) \Big)\\
&\phantom{xxx}\stackrel{\varphi \textrm{ cocomut.}}
{=}(v_{e_1})_{(1)} \ot (v_{e_2})_{(1)} \ot\cdots\ot (v_{e_n})_{(1)} \ot  \big( S \big ( (v_{e_1})_{(2)}\big) \lact X_{(2)} \\
&\phantom{xxxxxxxxx} \varphi\Big(  (v_{e_1})_{(3})  (v_{e_2})_{(2)} \cdots (v_{e_n})_{(2)} (S \circ \partial)\big ( X_{(1)}\big) \Big).
\end{align*}
This shows that indeed $(F')_P^\varphi \circ T^+_P = T^+_P \circ F_P^\varphi$ for any cocommutative $\varphi \in H^*$.
\end{proof}

\begin{remark}\label{rem:movingbp}
{Now  we know that base-point shifts commute with plaquette operators for cocommutative elements of $H^*$, in particular for the Haar integral $\lambda$.  Therefore base-point shifts   preserve the subspace:} \[\mathcal{H}^\text{ff}_L := \im\left(\prod_{P\in L^2} F_P\right) \subseteq \mathcal{H}_L.\] 
{We call $\mathcal{H}^\text{ff}_L$ the \emph{fake-flat subspace}.}
{In the \sHKM, this subspace coincides with the subspace of the total Hilbert space spanned by the  fake-flat configurations \cite{companion}}.
\end{remark}

{Let $P$ be a plaquette. As mentioned before, moving the base-point $v_P$ of $P$ around the boundary of the plaquette $P$, to return to $v_P$ again, by using base-point shifts,  does not in general yield the identity in $\Hil_L$. However this operation restricts to the identity over the subspace $\mathcal{H}^\text{ff}_L$.}
\begin{lemma} \label{lem:base-point-all-around}
Let $P \in L^2$ be a plaquette.
Let $T_P^\circlearrowleft : \mathcal{H}_L \lto \mathcal{H}_L$ be the composition of successive base-point shifts, moving the base-point counterclockwise once around the entire plaquette.
Then:
\[ {T_P^\circlearrowleft}\mid_{\mathcal{H}^\mathrm{ff}_L} = \id{\mathcal{H}^\mathrm{ff}_L} . \]
\end{lemma}
\begin{proof}
Let $(e_1,\dots,e_n)$ be the edges in the boundary of $P$ in counterclockwise order starting and ending at its base-point. Assume that they are each oriented in counterclockwise direction, which we may by applying edge orientation reversals where necessary.
For $v_{e_1} \ot\cdots\ot v_{e_n}\ot X_P \in H^{\ot n} \ot A$ we then have
\begin{align*}
&{T_P^\circlearrowleft}(v_{e_1} \ot\cdots\ot v_{e_n}\ot X_P) \\
&= (v_{e_1})_{(1)} \ot\cdots\ot (v_{e_n})_{(1)} \ot \big(S((v_{e_n})_{(2)}) \cdots S((v_{e_1})_{(2)}) \lact X_P \big) .
\end{align*}
Now assume that $(\cdots \ot v_{e_1} \ot\cdots\ot v_{e_n}\ot X_P \ot \cdots) \in (\cdots \ot H^{\ot n} \ot A \ot \cdots) = \mathcal{H}_L$ is a fake-flat state, {in particular (cf. the proof of Proposition \ref{prop:gr_state}):}
\[ (v_{e_1})_{(1)} \ot\cdots\ot (v_{e_n})_{(1)} \ot (X_P)_{(2)} \ot \big( (v_{e_1})_{(2)} \cdots (v_{e_n})_{(2)} S\partial (X_P)_{(1)} \big) = v_{e_1} \ot\cdots\ot v_{e_n}\ot X_P \ot 1_H . \]
Then we see that for such a state we have:
\begin{align*}
&{T_P^\circlearrowleft}(v_{e_1} \ot\cdots\ot v_{e_n}\ot X_P) \\
&= (v_{e_1})_{(1)} \ot\cdots\ot (v_{e_n})_{(1)} \ot \big(S((v_{e_n})_{(2)}) \cdots S((v_{e_1})_{(2)}) \lact X_P \big) \\
&= (v_{e_1})_{(1)} \ot\cdots\ot (v_{e_n})_{(1)} \ot \big(S((v_{e_1})_{(2)} \cdots (v_{e_n})_{(2)}) \lact X_P \big) \\
&= (v_{e_1})_{(1)} \ot\cdots\ot (v_{e_n})_{(1)} \ot \big(S\big((v_{e_1})_{(2)} \cdots (v_{e_n})_{(2)} \partial S({X_P}_{(1)}) \partial ({X_P}_{(2)})_{(1)} \big) \lact ({X_P}_{(2)})_{(2)} \big) \\
&\stackrel{\text{fake-flatness}}{=} v_{e_1} \ot\cdots\ot v_{e_n} \ot \big(S\big( \partial (X_P)_{(1)} \big) \lact (X_P)_{(2)} \big) \\
&\stackrel{\substack{\text{Pf. 2}\\ \text{$\partial$ Hopf alg. map}}}{=} v_{e_1} \ot\cdots\ot v_{e_n} \ot S(X_P)_{(2)} (X_P)_{(3)} (X_P)_{(1)} \\
&\stackrel{\text{antip.prop.}}{=} v_{e_1} \ot\cdots\ot v_{e_n} \ot X_P ,
\end{align*}
which proves the claim.
\end{proof}

\subsection{Commutation relations between vertex, edge and plaquette operators} \label{sec:proofs-of-commutation-relations}

\subsubsection{Commutation relations among edge operators in the adequate case}\label{sec:commutation_edge}

\begin{lemma} \label{lem:edge-operators-commute}
Let $e_1$ and $e_2 \in L^1$ be two distinct edges.
Then
\begin{enumerate}\setlength\itemsep{0em}
    \item\label{it:1} If $e_1$ and $e_2$ are not edges of a common plaquette then:
\begin{equation}
E_{e_1}^{a_1} \circ E_{e_2}^{a_2} = E_{e_2}^{a_2} \circ E_{e_1}^{a_1}, \qquad\text{ for any {two} $a_1, a_2 \in A$.}
\end{equation}
(For instance, this is the case on the left image in Figure \ref{fig:edge-comm}.)
\item If $e_1$ and $e_2$ are edges in the boundary of a common plaquette, then, as long as the pair $(e_1,e_2)$ is adequate, then:
\begin{equation}
E_{e_1}^{a_1} \circ E_{e_2}^{a_2} = E_{e_2}^{a_2} \circ E_{e_1}^{a_1}, \qquad\text{ for any {two cocommutative elements} $a_1, a_2 \in A$.}
\end{equation}
(For instance, this is the case on the three middle images in Figure \ref{fig:edge-comm}.)
\end{enumerate}
\end{lemma}
\noindent NB: if the pair $(e_1,e_2)$ is not adequate, vertex projects do not commmute in general.
\begin{proof}
The edge operators $E_{e_1}^{a_1}$ and $E_{e_2}^{a_2}$ have intersecting supports only if they lie in the boundary of a common plaquette, or if not, but there exists an edge $e$ that lies simultaneously in the boundary of a plaquette adjacent to $e_1$ and in the boundary of a plaquette adjacent to $e_2$, see the left of Figure \ref{fig:edge-comm}.
\begin{figure}[ht!]
 \labellist
\pinlabel $\small{v_P}$ at 51 322
\pinlabel $\small{v_Q}$ at 518 20
\pinlabel $e_1$ at 139 333
\pinlabel $e_2$ at 437 115
\pinlabel $e$ at 260 303
\pinlabel $P$ at 130 190
\pinlabel $Q$ at 345 190
%%%%%%%%%%%%%%%%%
\pinlabel $\small{v_P}$ at 808 70
\pinlabel $e_1$ at 675 238
\pinlabel $e_2$ at 941 238
\pinlabel $P$ at 806 200
\pinlabel $\small{\textrm{adequate (i)}}$ at 825 -20
%%%%%%%%%%%%%%%%%%%%%%%%%%
\pinlabel $\small{v_P}$ at 1320 70
\pinlabel $e_1$ at 1187 238
\pinlabel $e_2$ at 1456 238
\pinlabel $P$ at 1318 200
\pinlabel $\small{\textrm{adequate  (ii)}}$ at 1337 -20
%%%%%%%%%%%%%%%%%%%%%%%%%%
\pinlabel $\small{v_P}$ at 1832 70
\pinlabel $e_1$ at 1699 238
\pinlabel $e_2$ at 1968 238
\pinlabel $P$ at 1830 200
\pinlabel $\small{\textrm{adequate (iii)}}$ at 1850 -20
%%%%%%%%%%%%%%%%%%%%%%%%%%
%%%%%%%%%%%%%%%%%%%%%%%%%%
\pinlabel $\small{v_P}$ at 2344 70
\pinlabel $e_1$ at 2211 238
\pinlabel $e_2$ at 2480 238
\pinlabel $P$ at 2342 200
\pinlabel $\small{\textrm{non-adequate}}$ at 2361 -20
\endlabellist
\centering
\includegraphics[scale=0.19]{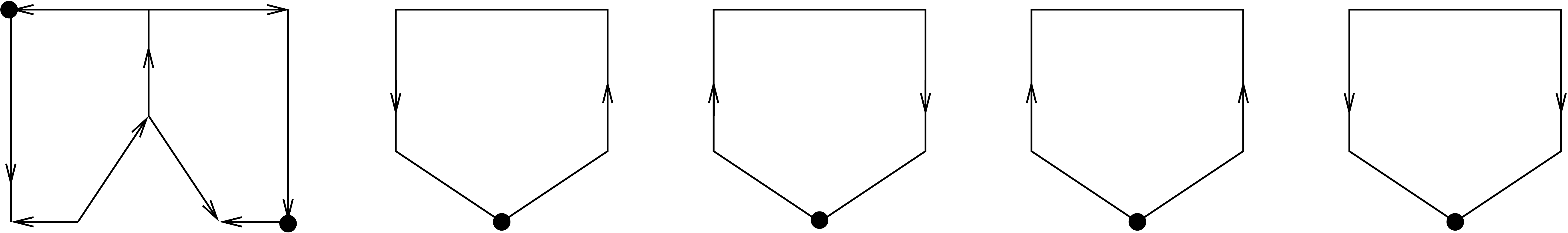}
\caption{{Different cases of commutativity, or non-commutativity, for non-adequate pairs, between edge operators, along $e_1$ and $e_2$}.}
\label{fig:edge-comm}
\end{figure}

Let us start with the latter case. In this easier case, $E_{e_1}^{a_1} \circ E_{e_2}^{a_2} = E_{e_2}^{a_2} \circ E_{e_1}^{a_1}$, regardless of $a_1$ and $a_2$ being cocommutative or not. The supports of the two operators only intersect at $e$. We now only need to observe that if $v\in H$ is the value associated to edge $e$, then, using the Yetter-Drinfeld condition \eqref{eq:yetter-drinfeld-condition-trivial}, given any $m,n \in A$:
$$v_{(1,1)} \otimes v_{(2)} \lact m \otimes  v_{(1,2)}\lact n= v_{(2,1)} \otimes v_{(2,2)} \lact m \otimes  v_{(1)}\lact n, $$
for the case when the edge $e$, in Figure \ref{fig:edge-comm} is oriented upwards, and, now using calculation \eqref{eq:compatibility<>},
$$v_{(1,1)} \otimes S\big( v_{(2)}\big) \lact m \otimes  S\big (v_{(1,2)}\big)\lact n= v_{(2,1)} \otimes S\big (v_{(2,2)} \big)\lact m \otimes S\big( v_{(1)}\big)\lact n,$$
for when $e$ is oriented downwards. We have now dealt with Item \ref{it:1} of the lemma.

Let us now deal with the case when $e_1$ and $e_2$, lie in the boundary of a common plaquette $P \in L^2$, but yet the configuration $(e_1,e_2)$ is adequate. 
We show in what follows that 
$E_{e_1}^{a_1} \circ E_{e_2}^{a_2} = E_{e_2}^{a_2} \circ E_{e_1}^{a_1}$ for any cocommutative $a_1, a_2 \in A$.
\begin{figure}[ht!]
 \labellist
\pinlabel $\small{v_{P_2}}$ at 400 445
\pinlabel $\small{v_{P_1}}$ at 1138 221
\pinlabel $\small{v_{P}}$ at 417 45
\pinlabel $d_1$ at 637 284
\pinlabel $d_2$ at 528 445
\pinlabel $e_2$ at 355 284
\pinlabel $e_1$ at 700 105
\pinlabel $P$ at 490 234
\pinlabel $P_2$ at 70 328
\pinlabel $P_1$ at 854 104
\endlabellist
\centering
\includegraphics[scale=0.19]{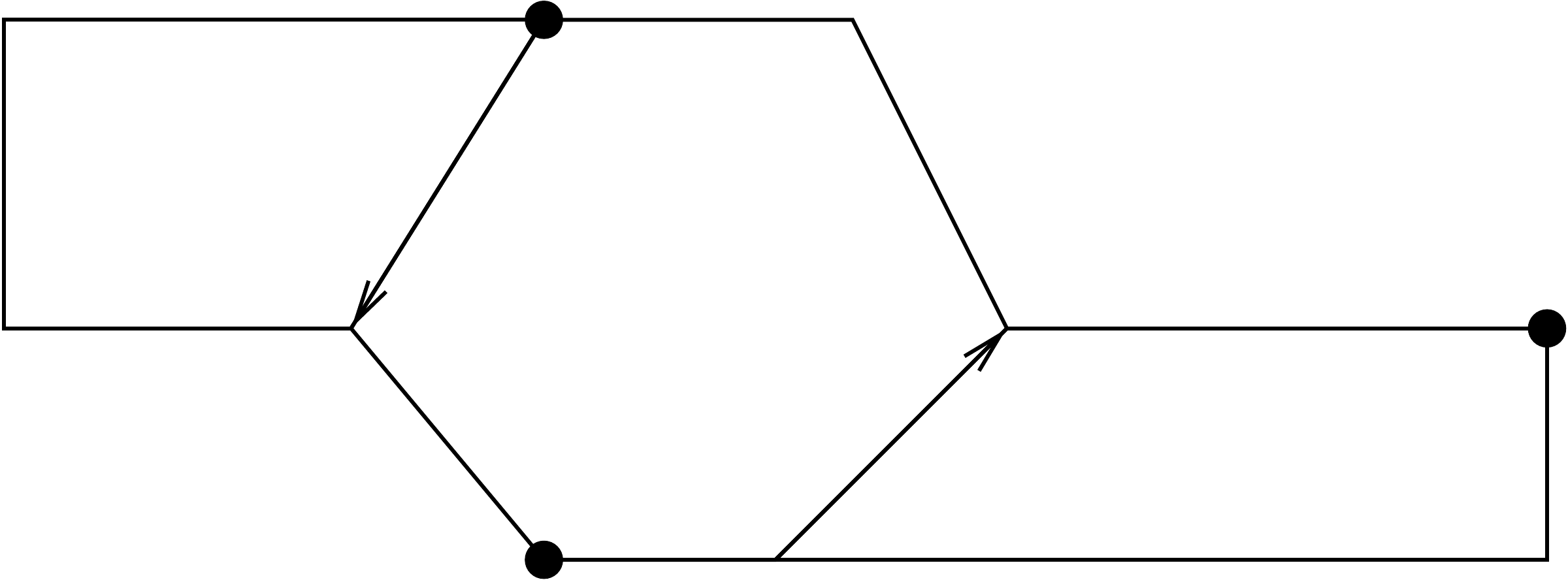}
\caption{{Configuration to prove $E_{e_1}^{a_1} \circ E_{e_2}^{a_2} = E_{e_2}^{a_2} \circ E_{e_1}^{a_1}$. We  move $v_{P_1}$, $v_P$ to the initial point of $e_1$.}}
\label{fig:edge-comm2}
\end{figure}

Denote by $P_1$ the other plaquette adjacent to $e_1$ and denote by $P_2$ the other plaquette adjacent to $e_2$. Given the requirement, in Definition~\ref{def:cell-decomposition}, that two different plaquettes have at most one edge in common in their boundary,  the plaquettes $P,P_1$ and $P_2$ are all different.

By applying base-point shifts that do not affect $e_1$ and $e_2$ {(Lemma \ref{lem:shiftsvsedges})} and exchanging $e_1$ and $e_2$ if necessary, we may assume that the base-point of $P$ is equal to the starting vertex of $e_1$. This is a point where the condition of adequacy comes in.
Furthermore, we may assume that the base-point of $P_1$ is  equal to the starting vertex of $e_1$, and that the base-point of $P_2$ is equal to the starting vertex of $e_2$.

For the explicit calculation of the commutation relations of edge operators, for adequate configurations, it is convenient to distinguish the following cases, already marked in Figure \ref{fig:edge-comm}.
\begin{enumerate}[label=(\roman*)] \setlength\itemsep{0em}
\item \label{item:e1_c-clockwise_e2_c-clockwise}
Both edges $e_1$ and $e_2$ are oriented counterclockwise with respect to $P$. {This is as  in Figure \ref{fig:edge-comm2}}.
\item \label{item:e1_anti_e2_anti}
Both edges $e_1$ and $e_2$ are oriented clockwise with respect to $P$.
\item \label{item:e1_e2_opposite}
$e_1$ and $e_2$ are in opposite orientation to each other with respect to $P$.
\end{enumerate}

\noindent\ref{item:e1_c-clockwise_e2_c-clockwise}:
Let $(d_1, \dots, d_k)$ be the edges between $e_1$ and $e_2$ (possibly none) in counterclockwise order around $P$.
By applying edge orientation reversals where necessary, we may assume that the edges $(d_1,\dots,d_k)$ are all oriented in counterclockwise direction around $P$.
Let $v_{e_1} \ot v_{d_1} \ot\cdots\ot v_{d_k} \ot v_{e_2} \in H_{e_1} \ot H_{d_1} \ot\cdots\ot H_{d_k} \ot H_{e_2} = H^{\ot (k+2)}$ and let {$X_{P} \ot X_{P_1}  \ot X_{P_2} \in A_P \ot A_{P_1}\ot A_{P_2}$.}
For the edge operators $E_{e_1}^{a}$ and $E_{e_2}^{a'}$, if $a \in A$ is cocommutative, we then have:
\begin{align*}
&(E_{e_2}^{a'}\circ  E_{e_1}^a) (v_{e_1} \ot v_{d_1} \ot\cdots\ot v_{d_k} \ot v_{e_2} \ot X_P \ot X_{P_1} {\ot  X_{P_2})} 
&\\
&= E_{e_2}^{a'} \big( \partial(a_{(3)}) v_{e_1} \ot v_{d_1} \ot\cdots\ot v_{d_k} \ot v_{e_2} \ot a_{(1)} X_P \ot X_{P_1} S(a_{(2)}) {\ot  X_{P_2}}\big) 
&\\
%%%%%%%%%%%%%%%%%%%%%%%%%%%%%%%%%
&= \big( \partial(a_{(3)}) v_{e_1} \big)_{(1)} \ot (v_{d_1})_{(1)} \ot\cdots\ot (v_{d_k})_{(1)} \ot \partial(a'_{(3)}) v_{e_2} \\
&\phantom{xxxxxxx} \ot \big( \big( \partial(a_{(3)}) v_{e_1} \big)_{(2)} (v_{d_1})_{(2)} \cdots (v_{d_k})_{(2)} \lact a'_{(1)}\big) a_{(1)} X_P  \ot X_{P_1} S(a_{(2)})  {\ot  X_{P_2} S (a'_{(2)}} )
&\\ 
%%%%%%%%%%%%%%%%%%%%%%%%%%%%%%%%%
&=  \partial(a_{(3)}) (v_{e_1} )_{(1)} \ot (v_{d_1})_{(1)} \ot\cdots\ot (v_{d_k})_{(1)} \ot \partial(a'_{(3)}) v_{e_2} \\
&\phantom{xxxxxxx} \ot  \big (\big (\partial(a_{(4)} ) (v_{e_1})_{(2)} (v_{d_1})_{(2)} \cdots (v_{d_k})_{(2)} \big)\lact a'_{(1)}\big)  a_{(1)} X_P \ot X_{P_1} S(a_{(2)})  {\ot  X_{P_2} S (a'_{(2)}} )
&\\
&\stackrel{\text{Pf. 2}}{=} 
 \partial(a_{(3)}) (v_{e_1} )_{(1)} \ot (v_{d_1})_{(1)} \ot\cdots\ot (v_{d_k})_{(1)} \ot \partial(a'_{(3)}) v_{e_2} \\
&\phantom{xxxxxxx} \ot   a_{(4)} \, \Big (\big ( (v_{e_1})_{(2)} (v_{d_1})_{(2)} \cdots (v_{d_k})_{(2)} \big)\lact a'_{(1)} \Big) \,\,S(a_{(5)}) a_{(1)} X_P\ot X_{P_1} S(a_{(2)})  {\ot  X_{P_2} S (a'_{(2)}} )
%%%%% 
&\\
&\stackrel{\text{$a$ cocomm.}}{=}
\partial(a_{(2)}) (v_{e_1} )_{(1)} \ot (v_{d_1})_{(1)} \ot\cdots\ot (v_{d_k})_{(1)} \ot \partial(a'_{(3)}) v_{e_2} \\
&\phantom{xxxxxxx} \ot   a_{(3)} \,\,  \Big (\big ( (v_{e_1})_{(2)} (v_{d_1})_{(2)} \cdots (v_{d_k})_{(2)} \big)\lact a'_{(1)}\Big)\,\, S(a_{(4)}) a_{(5)} X_P \ot X_{P_1} S(a_{(1)})  {\ot  X_{P_2} S (a'_{(2)}} )
%%%%%
&\\
&\stackrel{\text{antip. prop.}}{=} 
\partial(a_{(2)}) (v_{e_1} )_{(1)} \ot (v_{d_1})_{(1)} \ot\cdots\ot (v_{d_k})_{(1)} \ot \partial(a'_{(3)}) v_{e_2} \\
&\phantom{xxxxxxx} \ot   a_{(3)}\,\,   \Big (\big ( (v_{e_1})_{(2)} (v_{d_1})_{(2)} \cdots (v_{d_k})_{(2)} \big)\lact a'_{(1)}\Big) X_P  \ot X_{P_1} S(a_{(1)})  {\ot  X_{P_2} S (a'_{(2)}} )
%%%%
\\
&\stackrel{\text{$a$ cocomm.}}{=} \partial(a_{(3)}) (v_{e_1})_{(1)} \ot (v_{d_1})_{(1)} \ot\cdots\ot (v_{d_k})_{(1)} \ot \partial(a'_{(3)}) v_{e_2} \\
&\phantom{xxxxxxx} \ot a_{(1)} \Big( \big( (v_{e_1})_{(3)} (v_{d_1})_{(3)} \cdots (v_{d_k})_{(3)}\big) \lact a'_{(1)}\Big) X_P  \ot X_{P_1} S(a_{(2)}) \ot X_{P_2} S \big( a'_{(2)} \big)
&\\
&= (E_{e_1}^a\circ  E_{e_2}^{a'}) (v_{e_1} \ot v_{e_2} \ot X_P \ot X_{P_1} {\ot  X_{P_2}}) .
\end{align*}
\noindent\ref{item:e1_anti_e2_anti}:
Here the calculation is as for \ref{item:e1_c-clockwise_e2_c-clockwise}
\medskip

\noindent\ref{item:e1_e2_opposite}:
{In this case, $E_{e_1}^a$ and $E_{e_2}^{a'}$ for any $a, a' \in A$ act on the tensor factors $A_P$ by multiplication from opposite sides, and exactly the same in $A_{P_1}$ and $A_{P_2}$.  Hence $E_{e_1}^a$ and $E_{e_2}^{a'}$  commute, regardless of $a$ and $a'$ being cocommutative or not. Note that it is crucial that the configuration is adequate,  so that the action of $E_{e_1}^a$ does not affect the edges that are used to compute  $E_{e_2}^{a'}$, and vice-versa.  }
\end{proof}

\subsubsection{Commutation relations between vertex and edge operators}

\begin{lemma} \label{lem:edge-operators-and-vertex-operators}
Let $v \in L^0$ be a vertex, with adjacent plaquette $P \in L^2$, and let $e \in L^1$ be an edge.
Then for any $h \in H$ and $a \in A$ the following hold:
\begin{enumerate}\setlength\itemsep{0em}
\item\label{IT1} $E_e^{h_{(1)}\lact a} \circ V_{v,P}^{h_{(2)}} = V_{v,P}^h \circ E_e^a$, if $v$ is the starting vertex of $e$.\\
In particular, by  Lemma \ref{lem:haar-integral-central-in-crossed-product}, the vertex projector $V_v = V_{v,P}^\ell$ and edge projector $E_e = E_e^\Lambda$ commute with each other, where $\ell \in H$ and $\Lambda \in A$ are the Haar integrals of the respective Hopf algebras.

\item\label{IT2} $E_e^a \circ V_{v,P}^h = V_{v,P}^h \circ E_e^a$, if $v$ is not the starting vertex of $e$.
\end{enumerate}
\end{lemma}
\begin{proof}(Sketch.)
Let $P_1$ and $P_2$ be the plaquettes on the left and right of the edge $e$ with respect to the orientations of $\Sigma$ and $e$. 
The supports of the corresponding vertex operator $V_{v,P}^h$, where $h \in H$, and edge operator $E_e^a$, where $a \in A$, intersect only if $v$ lies in the boundary of one of the two (or both) plaquettes $P_1$ and $P_2$ adjacent to $e$. {Some possible configurations are in  figure \ref{fig:vertex_edge}.}
\begin{figure}[ht!]
 \labellist
\pinlabel $\small{e}$ at 260 233
\pinlabel $\small{v_{P_1}}$ at 41 385
\pinlabel $\small{v_{P_2}}$ at 426 385
\pinlabel $\small{P_1}$ at 76 286
\pinlabel $\small{P_2}$ at 392 286
\pinlabel $\small{P}$ at 358   32
\pinlabel $\small{v}$ at 238 55
\pinlabel $\small{e_2}$ at 150 97
\pinlabel $\small{e_1}$ at 330 97
\pinlabel $\small{\textrm{case (i)}}$ at 230 -29
%%%%%%%%%%%%%%%%%%%
\pinlabel $\small{e}$ at 886 100
\pinlabel $\small{v_{P_1}}$ at 681 246
\pinlabel $\small{v_{P_2}}$ at 1058 260
\pinlabel $\small{P_1}$ at 702 120
\pinlabel $\small{P_2}$ at 1020 120
\pinlabel $\small{P}$ at 965   315
\pinlabel $\small{v}$ at 864 305
\pinlabel $\small{e_2}$ at 928 267
\pinlabel $\small{e_1}$ at 800 267
\pinlabel $\small{\textrm{case (ii)}}$ at 850 -29
%%%%%%%%%%%%%%%%%%%%
\pinlabel $\small{e}$ at 1513 233
\pinlabel $\small{v_{P_1}}$ at 1293 385
\pinlabel $\small{v_{P_2}}$ at 1680 385
\pinlabel $\small{P_1}$ at 1328 286
\pinlabel $\small{P_2}$ at 1642 286
\pinlabel $\small{P}$ at 1600   33
\pinlabel $\small{v}$ at 1742 60
\pinlabel $\small{\textrm{case (iii)}}$ at 1490 -29
\endlabellist
\centering
\includegraphics[scale=0.23]{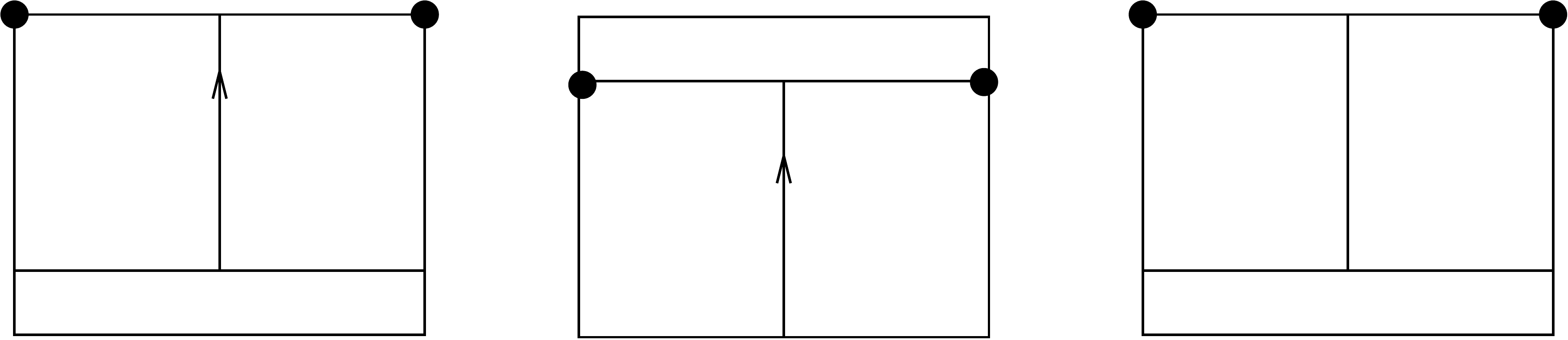}
\caption{{Different configuration for commmutativity between edge and vertex operators. In all cases, (i), (ii) and (iii), we shall move the base-points of $P_1$ and $P_2$ to the initial point of $e$, using a path that does not pass by edge $e$.}}
\label{fig:vertex_edge}
\end{figure}

{We show the calculations for the particular cases of Figure \ref{fig:vertex_edge}. Item \ref{IT1} corresponds to case (i), and Item \ref{IT2} to cases (ii) and (iii).} 

\medskip

\noindent Case (i):   We move $v_{P_1}$ and $v_{P_2}$ to $v$, applying Items  \ref{it:bps-vertex} and \ref{it:bps-edge} of Proposition \ref{prop: commut-vertex-ops-shifts_etc}, and make edges $e_1$ and $e_2$ point away from $v$, applying Items \ref{it:vertex-eor} and \ref{it:edge-eor} of Proposition \ref{prop: commut-vertex-ops-shifts_etc}. Let us show that if $h\in H$ and $a \in A$ \[ E_e^{h_{(1)}\lact a} \circ V_{v,P}^{h_{(2)}} = V_{v,P}^h \circ E_e^a .\] We have:
\begin{align*}
\big ( E_e^{h_{(1)}\lact a} \circ & V_{v,P}^{h_{(2)}}\big)\big ( v_{e_1} \ot v_e \ot v_{e_2} \ot X_{P_1} \ot X_{P_2} \big )
%%%%
\\&=
 E_e^{h_{(1)}\lact a} \big ( h_{(2)}v_{e_1} \ot  h_{(3)}v_e \ot  h_{(4)}v_{e_2} \ot   h_{(5)}\lact X_{P_1} \ot  h_{(6)}\lact X_{P_2}  \big )
\\&=
 h_{(2)}v_{e_1} \ot \partial( (h_{(1)}\lact a)_{(3)}) h_{(3)}v_e \ot  h_{(4)}v_{e_2}\\
 & \qquad\ot  \big(  (h_{(1)}\lact a)_{(1)} \big)   h_{(5)}\lact X_{P_1} \ot  h_{(6)}\lact X_{P_2} \, S \big(  (h_{(1)}\lact a)_{(2)} \big)
 %%%%%
 \\&=
 h_{(4)}v_{e_1} \ot \partial( h_{(3)} \lact a_{(3)})) h_{(5)}v_e \ot  h_{(6)}v_{e_2}\\
 & \qquad\ot  \big(  h_{(1)}\lact a_{(1)} \big)   h_{(7)}\lact X_{P_1} \ot  h_{(8)}\lact X_{P_2} S \big(  h_{(2)}\lact a_{(2)} \big) 
 \\&\stackrel{\text{$S_A$ $H$-lin.}}{=} h_{(4)}v_{e_1} \ot \partial( h_{(3)} \lact a_{(3)})) h_{(5)}v_e \ot  h_{(6)}v_{e_2}\\
 & \qquad\ot  \big(  h_{(1)}\lact a_{(1)} \big)   h_{(7)}\lact X_{P_1} \ot  h_{(8)}\lact X_{P_2} \,\,  h_{(2)}\lact S \big(  a_{(2)} \big).
\end{align*}
On the other hand,
\begin{align*}
\big (   V_{v,P}^{h} \circ &  E_e^{ a}\big)\big ( v_{e_1} \ot v_e \ot v_{e_2} \ot X_{P_1} \ot X_{P_2}  )
%%%%%
\\&=  V_{v,P}^{h} \big (  v_{e_1} \ot \partial(a_{(3)} ) v_e \ot v_{e_2} \ot a_{(1)} X_{P_1} \ot X_{P_2} S(a_{(2)} \big ) 
\\&=  h_{(1)} v_{e_1} \ot h_{(2)} \partial(a_{(3)} ) v_e \ot h_{(3)} v_{e_2} \ot h_{(4)}\lact \big (a_{(1)} X_{P_1}\big) \ot h_{(5)}\lact\big(  X_{P_2} S(a_{(2)}\big) 
%%%%%
\\&=  h_{(1)} v_{e_1} \ot h_{(2)} \partial(a_{(3)} ) v_e \ot h_{(3)} v_{e_2} \\ &\qquad \ot \big ( h_{(4)}\lact a_{(1)} \big ) \big (h_{(5)} \lact X_{P_1}\big) \ot \big (h_{(6)}\lact   X_{P_2} \big) \big (h_{(7)}\lact  S(a_{(2)}\big)
\\&=  h_{(1)} v_{e_1} \ot  \partial(h_{(2)} \lact a_{(3)} ) h_{(3)} v_e \ot h_{(4)} v_{e_2}\\ &\qquad \ot \big ( h_{(5)}\lact a_{(1)} \big ) \big (h_{(6)} \lact X_{P_1}\big) \ot \big (h_{(7)}\lact   X_{P_2} \big) \big (h_{(8)}\lact  S(a_{(2)}\big).\end{align*}
We used the 1st Peiffer condition in the last step.
The two expressions coincide, due to Lemma \ref{lem:yetter-drinfeld-for-many-factors}. 

\noindent Case (ii) Here $v$ is the end vertex of $e$, and again we move $v_{P_1}$ and $v_{P_2}$ to the initial point of $e$, and make edges $e_1$ and $e_2$ point away from $v$. Let us show that if $h\in H$ and $a \in A$ \[ E_e^{ a} \circ V_{v,P}^{h} = V_{v,P}^h \circ E_e^a .\] We have:
\begin{align*}
\big ( E_e^{ a} & \circ V_{v,P}^{h}\big)\big ( v_{e_1} \ot v_e \ot v_{e_2} \ot X_{P_1} \ot X_{P_2}\big )
%%%%
=
 E_e^{ a} \big ( h_{(1)}v_{e_1} \ot  v_e S(h_{(2)}) \ot  h_{(3)}v_{e_2} \ot X_{P_1} \ot X_{P_2}  \big )\\
 &=h_{(1)}v_{e_1} \ot  \partial(a_{(3)}) v_e S(h_{(2)}) \ot  h_{(3)}v_{e_2} \ot a_{(1)} X_{P_1} \ot X_{P_2} S(a_{(2)} ).
 \end{align*}
On the other hand,
\begin{align*}
\big (   V_{v,P}^{h} \circ &  E_e^{ a}\big)\big ( v_{e_1} \ot v_e \ot v_{e_2}  \ot X_{P_1} \ot X_{P_2} \big )
=  V_{v,P}^{h} \big (  v_{e_1} \ot \partial(a_{(3)} ) v_e  \ot v_{e_2} \ot a_{(1)} X_{P_1} \ot X_{P_2} S(a_{(2}) \big )
\\&=  h_{(1)} v_{e_1} \ot \partial(a_{(3)} ) v_e S(h_{(2)}) \ot h_{(3)}  v_{e_2}  \ot a_{(1)} X_{P_1} \ot X_{P_2} S(a_{(2)} \big ).
\end{align*}
These two trivially coincide.

\noindent Case (iii): Here $v$ is neither the initial nor the end-point of $e$. We can move the base-points of $P_1$ and of $P_2$ to be the starting vertex of $e$ (see above). Since $v$ is not adjacent to $e$, the corresponding vertex operator and edge operator have disjoint support and therefore clearly commute with each other.
\end{proof}

\subsubsection{Commutation relations between edge and plaquette operators}

\begin{lemma} \label{lem:edge-and-plaquette-operators-commute}
Let $P\in L^2$ be a plaquette and let $e \in L^1$ be an edge.
If $e$ is in the boundary of $P$, then
\[ F_P^{\varphi} \circ E_e^{a} = E_e^a \circ F_P^\varphi,
\quad\text{ for all cocommutative } a \in A, {\textrm{and cocommutative }} \varphi \in H^*. \]
Moreover, if additionally the starting vertex of $e$ is the base-point of $P$, then
\[ F_P^{\varphi\left(S \partial a_{(3)} \cdot ? \cdot \partial a_{(1)} \right)} \circ E_e^{a_{(2)}} = E_e^a \circ F_P^\varphi,
\quad\text{ for all } a \in A, \varphi \in H^* , \]
 if $e$ is oriented counterclockwise around $P$,  and if $e$ is oriented clockwise around $P$,
\[ F_P^\varphi \circ E_e^a = E_e^a \circ F_P^\varphi,
\quad\text{ for all } a \in A, \varphi \in H^*  .\]

If $e$ is not in the boundary of $P$, then
\[ F_P^{\varphi} \circ E_e^{a} = E_e^a \circ F_P^\varphi,
\quad\text{ for all } a \in A, \varphi \in H^*. \]
\end{lemma}
\begin{proof} {We refer to Figure \ref{fig:edge-plaquettes}.}
 There are two types of situations in which the plaquette operators $F_P^\varphi$ and the edge operators $E_e^a$ can have a non-trivially intersecting support.
Either the edge $e$ is in the boundary of the plaquette $P$, this is case (i).  Or it is not, but some  edge in the path connecting the base-point of a plaquette, say $Q$, which is adjacent to the edge $e$, with the starting vertex of $e$ lie in the boundary of $P$. This is case (ii). We deal with the cases separately. 

\begin{figure}[ht!]
 \labellist
\pinlabel $\small{Q}$ at 1239 197
\pinlabel $\small{P}$ at 867 175
\pinlabel $\small{e}$ at 1243 280
\pinlabel $\small{c}$ at 1125 209
\pinlabel $\small{f}$  at 888 278
\pinlabel $\small{v_{P}}$ at 737 125
\pinlabel $\small{v_{Q}}$ at 1360 11
\pinlabel $\textrm{case (ii)}$ at 1000 -20
\pinlabel $\small{P}$ at 290 175  
\pinlabel $\small{v_P}$ at 52 140  
\pinlabel $\small{e}$ at 482 87 
\pinlabel $\small{e'}$ at 250 280
\pinlabel $\small{v_Q}$ at 447 -10
\pinlabel $\small{Q}$ at 569 62
\pinlabel $\small{e}$ at 482 87 
\pinlabel $\small{v_{Q'}}$ at 56 317
\pinlabel $\small{Q'}$ at 500 324 
 \pinlabel $\textrm{case (i)}$ at 250 -20
\endlabellist
\centering
\includegraphics[scale=0.27]{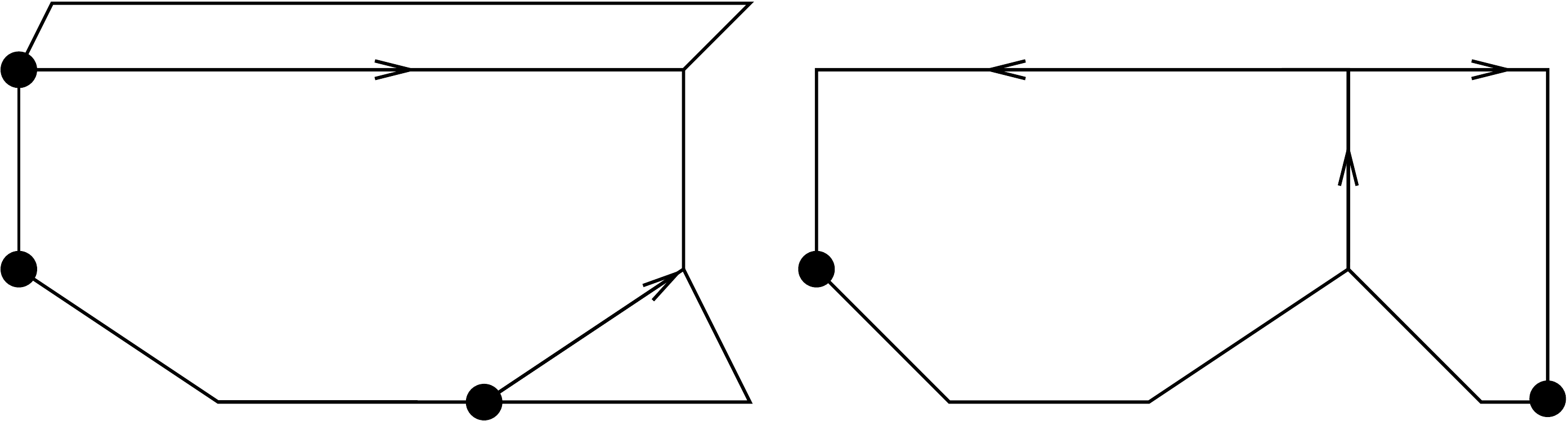}
\caption{{Two particular cases of commutation relations between edge and plaquette operators.}}
\label{fig:edge-plaquettes}
\end{figure}
{Case (ii):}
Let us refer to Figure \ref{fig:edge-plaquettes}. Let us see why  $$F_P^{\varphi} \circ E_e^{a} = E_e^a \circ F_P^\varphi,
\quad\text{ for all } a \in A, \varphi \in H^*,$$
and also, in order to treat other orientation conventions that:
$$F_Q^{\varphi} \circ E_f^{a} = E_f^a \circ F_Q^\varphi,
\quad\text{ for all } a \in A, \varphi \in H^*.$$
This follows from that given $h \in H, n \in A$, where $h$ is the element of $H$ associated to the  edge $c$,
\begin{align*}
(h_{(1)})_{(2)} \ot h_{(2)} \ot  (h_{(1)})_{(1)} \lact n & = (h_{(2)})_{(1)}  \ot (h_{(2)})_{(2)}  \ot  h_{(1)}\lact n,\\
(h_{(1)})_{(2)}\ot S(h_{(2)}) \ot (h_{(1)})_{(1)}\lact n&= (h_{(1)})_{(1)} \ot S((h_{(1)})_{(2)})\ot h_{(2)}\lact n. 
\end{align*}
 
\noindent {Case (i):}
Let us first assume that $e$ is oriented in counterclockwise direction around $P$. {We let $Q$ be the other plaquette adjacent to $e$.} 
For now assume further that the starting vertex of $e$ is the base-point of $P$.
Denote by $(e = e_1, \dots, e_n)$ the edges in the boundary of $P$ in counterclockwise order starting and ending at the base-point of $P$.
By applying edge orientation reversals where necessary, we may assume that also the edges $(e_2,\dots,e_n)$ are oriented in counterclockwise direction around $P$.
Then for $v_{e_1} \ot\cdots\ot v_{e_n} \ot X_P \ot X_Q \in H_{e_1} \ot\cdots\ot H_{e_n} \ot A_P \ot A_Q = H^{\ot n} \ot A \ot A$, the plaquette operator $F_P^\varphi$ for any $\varphi \in H^*$ acts, according to Definition \ref{def:plaquette-operator}, as
\begin{align*}
&F_P^\varphi (v_{e_1} \ot\cdots\ot v_{e_n} \ot X_P \ot X_Q) \\
&\phantom{xxx}= (v_{e_1})_{(1)} \ot\cdots\ot (v_{e_n})_{(1)} \ot (X_P)_{(2)} \ot X_Q \ \varphi\big( (v_{e_1})_{(2)} \cdots (v_{e_n})_{(2)} S\partial (X_P)_{(1)} \big) ,
\end{align*}
and the edge operator $E_e^a$ for any $a \in A$ acts, according to Definition \ref{def:edge-operator}, as
\begin{align*}
&E_e^a (v_{e_1} \ot\cdots\ot v_{e_n} \ot X_P \ot X_Q) = \partial(a_{(3)}) v_{e_1} \ot v_{(e_2)} \ot\cdots\ot v_{e_n} \ot a_{(1)} X_P \ot X_Q S(a_{(2)}) .
\end{align*}
We calculate, also {recalling that $S=S^{-1}$ when we apply the antipode property,}
\begin{align*}
&\big(F_P^{\varphi\left(S \partial a_{(3)} \cdot ? \cdot \partial a_{(1)} \right)}\circ  E_e^{a_{(2)}}\big)(v_{e_1} \ot\cdots\ot v_{e_n} \ot X_P \ot X_Q) 
& \\
&\phantom{xxx}= F_P^{\varphi\left(S \partial a_{(5)} \cdot ? \cdot \partial a_{(1)} \right)} \big( \partial(a_{(4)}) v_{e_1} \ot v_{e_2} \ot\cdots\ot v_{e_n} \ot a_{(2)} X_P \ot X_Q S(a_{(3)}) \big) 
& \\
&\phantom{xxx}= (\partial(a_{(4)}) v_{e_1})_{(1)} \ot (v_{e_2})_{(1)} \ot\cdots\ot (v_{e_n})_{(1)} \ot (a_{(2)} X_P)_{(2)} \ot X_Q S(a_{(3)}) \\
&\phantom{xxxxxxx} \varphi\big( S \partial a_{(5)} (\partial(a_{(4)}) v_{e_1})_{(2)} (v_{e_2})_{(2)} \cdots (v_{e_n})_{(2)} S \partial (a_{(2)} X_P)_{(1)} \partial a_{(1)} \big) 
& \\
&\phantom{xxx}= \partial(a_{(5)}) (v_{e_1})_{(1)} \ot (v_{e_2})_{(1)} \ot\cdots\ot (v_{e_n})_{(1)} \ot a_{(3)} (X_P)_{(2)} \ot X_Q S(a_{(4)}) \\
&\phantom{xxxxxxx} \varphi\big( S \partial a_{(7)} \partial(a_{(6)}) (v_{e_1})_{(2)} (v_{e_2})_{(2)} \cdots (v_{e_n})_{(2)} S \partial (X_P)_{(1)} S \partial a_{(2)} \partial a_{(1)} \big) 
& \\
&\phantom{xxx}\stackrel{\text{antipode prop.}}{=} \partial(a_{(3)}) (v_{e_1})_{(1)} \ot (v_{e_2})_{(1)} \ot\cdots\ot (v_{e_n})_{(1)} \ot a_{(1)} (X_P)_{(2)} \ot X_Q S(a_{(2)}) \\
&\phantom{xxxxxxx} \varphi\big( (v_{e_1})_{(2)} (v_{e_2})_{(2)} \cdots (v_{e_n})_{(2)} S \partial (X_P)_{(1)} \big)
& \\
&\phantom{xxx}= \big(E_e^a\circ F_P^\varphi \big)(v_{e_1} \ot\cdots\ot v_{e_n} \ot X_P \ot X_Q) .
\end{align*}
{For cocommutative elements $a \in A$ and $\varphi \in H^*$, this implies that the operators commute:}
\begin{align*}
E_e^a \circ F_P^\varphi &= F_P^{\varphi\left(S \partial a_{(3)} \cdot ? \cdot \partial a_{(1)} \right)} \circ E_e^{a_{(2)}} \stackrel{\text{$\varphi$ cocom.}}{=} F_P^{\varphi\left(  ? \cdot \partial a_{(1)} S \partial a_{(3)} \right)} \circ E_e^{a_{(2)}} \\
&\stackrel{\text{$a$ cocom.}}{=} F_P^{\varphi\left(  ? \cdot \partial a_{(2)} S \partial a_{(1)} \right)} \circ E_e^{a_{(3)}} 
= F_P^\varphi \circ E_e^a .
\end{align*}

Next assume that $e$ lies in the boundary of $P$, but the starting vertex of $e$ is not equal to the base-point of $P$.
Then, since all edge operators at $e$ commute with base-point shifts along edges other than $e$, and the plaquette operators for cocommutative elements of $H^*$ commute with base-point shifts, see Lemma \ref{lem:base-point-shifts-and-plaquette-operators}, we may reduce the situation to the above situation where the starting vertex of $e$ and the base-point of $P$ were the same, and we obtain:
\[ F_P^{\varphi\left(S \partial a_{(3)} \cdot ? \cdot \partial a_{(1)} \right)} \circ E_e^{a_{(2)}} = E_e^a \circ F_P^\varphi, \quad\text{ for all } a \in A, \text{ and cocommutative } \varphi \in H^* , \]
and, hence, for $a \in A$ and $\varphi \in H^*$ both cocommutative we obtain again: $F_P^\varphi \circ E_e^a = E_e^a \circ F_P^\varphi$.

Finally, assume now that the edge $e$, in the boundary of $P$ is oriented clockwise around $P$. This is the situation for $e'$ in figure \ref{fig:edge-plaquettes}.
We again first consider the case where the starting vertex of $e$ coincides with the base-point of $P$.
Again by applying base-point shifts where necessary, we may assume that the base-point of the other plaquette $Q$,  {in Figure \ref{fig:edge-plaquettes} denoted $Q'$} adjacent to $e$ is also equal to the starting vertex of $e$.
Denoting by $(e_1, \dots, e_n = e)$ the edges in the boundary of $P$ in counterclockwise order, we may assume, by applying edge orientation reversals where necessary, that the edges $(e_1, \dots, e_{n-1})$ are oriented counterclockwise around $P$, whereas $e_n = e$ is by assumption oriented clockwise around $P$.
Now a very similar calculation as the first one in this proof above will show that
\[ F_P^\varphi \circ E_e^a = E_e^a \circ F_P^\varphi, \qquad\text{for all } a \in A, \varphi \in H^* . \]
We calculate:
\begin{align*}
&\big(F_P^\varphi \circ E_e^a \big)(v_{e_1} \ot\cdots\ot v_{e_n} \ot X_P \ot X_Q) 
& \\
&\phantom{xxx}= F_P^\varphi (v_{e_1} \ot\cdots\ot v_{e_{n-1}} \ot \partial a_{(3)} v_{e_n} \ot X_P S a_{(2)} \ot a_{(1)} X_Q) 
& \\
&\phantom{xxx}= (v_{e_1})_{(1)} \ot\cdots\ot (v_{e_{n-1}})_{(1)} \ot (\partial a_{(3)} v_{e_n})_{(2)} \ot (X_P S a_{(2)})_{(2)} \ot a_{(1)} X_Q \\
&\phantom{xxxxxxx}\varphi\big( (v_{e_1})_{(2)} \cdots (v_{e_{n-1}})_{(2)} S(\partial a_{(3)} v_{e_n})_{(1)} S\partial (X_P S a_{(2)})_{(1)} \big) 
& \\
&\phantom{xxx}\stackrel{\substack{\text{$\Delta$ alg. map} \\ \text{$S$ anti-alg.map}}}{=} (v_{e_1})_{(1)} \ot\cdots\ot (v_{e_{n-1}})_{(1)} \ot \partial a_{(5)} (v_{e_n})_{(2)} \ot (X_P)_{(2)} S a_{(2)} \ot a_{(1)} X_Q \\
&\phantom{xxxxxxxxxxxxx}\varphi\big( (v_{e_1})_{(2)} \cdots (v_{e_{n-1}})_{(2)} S(v_{e_n})_{(1)} S(\partial a_{(4)})  \partial a_{(3)} S\partial (X_P)_{(1)} \big) 
& \\
&\phantom{xxx}\stackrel{\text{antipod.prop.}}{=} (v_{e_1})_{(1)} \ot\cdots\ot (v_{e_{n-1}})_{(1)} \ot \partial a_{(3)} (v_{e_n})_{(2)} \ot (X_P)_{(2)} S a_{(2)} \ot a_{(1)} X_Q \\
&\phantom{xxxxxxxxxxxxx}\varphi\big( (v_{e_1})_{(2)} \cdots (v_{e_{n-1}})_{(2)} S(v_{e_n})_{(1)} S\partial (X_P)_{(1)} \big) 
& \\
&= \big(E_e^a \circ F_P^\varphi \big)(v_{e_1} \ot\cdots\ot v_{e_n} \ot X_P \ot X_Q).
\end{align*}
Finally, for the case that the starting vertex of $e$ and the base-point of $P$ do not coincide, as above it follows using base-point shifts that $F_P^\varphi \circ E_e^a = E_e^a \circ F_P^\varphi$ still holds for all cocommutative $\varphi \in H^*$ and all $a \in A$.
\end{proof}

\subsubsection{Commutation relations between vertex and plaquette operators}

\begin{lemma} \label{lem:plaquette-operators-and-vertex-operators}
Let $P \in L^2$ be a plaquette and let $v \in L^0$ be a vertex.
If $v$ is in the boundary of $P$, then
\[ F_P^{\varphi} \circ V_{v,P'}^h = V_{v,P'}^h \circ F_P^\varphi,
\quad\text{ for all cocommutative } h \in H, \textrm{ and cocommutative } \varphi \in H^*, \]
where $P' \in L^2$ may be any plaquette with base-point $v$.

If additionally $v$ is the base-point of $P$, then
\[ F_P^{\varphi(S h_{(3)} \cdot ? \cdot h_{(1)} )} \circ V_{v,P}^{h_{(2)}} = V_{v,P}^h \circ F_P^\varphi,
\quad\text{ for all } h \in H, \varphi \in H^*. \]
 
If $v$ is not in the boundary of $P$, then
\[ F_P^{\varphi} \circ V_{v,P'}^h = V_{v,P'}^h \circ F_P^\varphi,
\quad\text{ for all } h \in H, \varphi \in H^* . \]
\end{lemma}
\noindent {Note that since $h$ is cocommutative, the plaquette $P'$ has no effect in the conventions for $V^h_{v,P'}$.}
\begin{proof}
If $v$ is not in the boundary of $P$, then the operators have disjoint support and therefore clearly commute with each other as claimed.
Let hence $v$ be in the boundary of $P$.

Let us first assume that $v$ is the base-point of $P$.
Denote by $(e_1, \dots, e_n)$ the edges in the boundary of $P$ in counterclockwise order starting at the base-point of $P$.
By applying edge orientation reversals where necessary we may assume that these edges are all {oriented counterclockwise around $P$.}

Let $(f_1, \dots, f_k)$ be the edges incident to the vertex $v$ in counterclockwise order starting and ending at $P$.
Note that $e_1 = f_k$ and $f_1 = e_n$.
By applying edge orientation reversals where necessary we may assume that the edges $f_2, \dots, f_{k-1}$ are oriented away from $v$.
Let $(Q_1, \dots, Q_\ell)$ be the plaquettes which have $v$ as their base-points, such that $Q_1 = P$.
Let us calculate for $(v_{e_2} \ot\cdots\ot v_{e_n} \ot v_{f_2} \ot\cdots\ot v_{f_k}) \ot (X_P \ot X_{Q_2} \ot\cdots\ot X_{Q_\ell}) \in H^{\ot n+k-2} \ot A^{\ot \ell}$:
\begin{align*}
&\big( F_P^{\varphi\left(S h_{(3)} \cdot ? \cdot h_{(1)} \right)} V_{v,P}^{h_{(2)}} \big) (v_{e_2} \ot\cdots\ot v_{e_n} \ot v_{f_2} \ot\cdots\ot v_{f_k} \ot X_P \ot X_{Q_2} \ot\cdots\ot X_{Q_\ell}) \\
&\phantom{xxx}= F_P^{\varphi\left(S h_{(3)} \cdot ? \cdot h_{(1)} \right)} \big( v_{e_2} \ot\cdots\ot v_{e_{n-1}} \ot v_{e_n} S(h_{(2)}) \ot h_{(3)} v_{f_2} \ot\cdots\ot h_{(k+1)} v_{f_k} \\
&\phantom{xxxxx} \ot h_{(k+2)} \lact X_P \ot h_{(k+3)} \lact X_{Q_2} \ot\cdots\ot h_{(k+\ell+1)} \lact X_{Q_\ell} \big) 
& \\
&\phantom{xxx}= (v_{e_2})_{(1)} \ot\cdots\ot (v_{e_{n-1}})_{(1)} \ot (v_{e_n} S(h_{(2)}))_{(1)} \ot h_{(3)} v_{f_2} \ot\cdots\ot h_{(k)} v_{f_{k-1}} \ot (h_{(k+1)} v_{f_k})_{(1)} \\
&\phantom{xxxxx} \ot (h_{(k+2)} \lact X_P)_{(2)} \ot h_{(k+3)} \lact X_{Q_2} \ot\cdots\ot h_{(k+\ell+1)} \lact X_{Q_\ell} \\
&\phantom{xxxxxxx} \varphi\left(S h_{(k+\ell+2)} (h_{(k+1)} v_{f_k})_{(2)} (v_{e_2})_{(2)} \cdots (v_{e_{n-1}})_{(2)} (v_{e_n} S(h_{(2)}))_{(2)} S \partial (h_{(k+2)} \lact X_P)_{(1)} h_{(1)} \right) \\
&\phantom{xxx}= (v_{e_2})_{(1)} \ot\cdots\ot (v_{e_{n-1}})_{(1)} \ot (v_{e_n})_{(1)} S(h_{(3)}) \ot h_{(4)} v_{f_2} \ot\cdots\ot h_{(k+1)} v_{f_{k-1}} \ot h_{(k+2)} (v_{f_k})_{(1)} \\
&\phantom{xxxxx} \ot h_{(k+5)} \lact (X_P)_{(2)} \ot h_{(k+6)} \lact X_{Q_2} \ot\cdots\ot h_{(k+\ell+4)} \lact X_{Q_\ell} \\
&\phantom{xxxxxxx} \varphi\left(S h_{(k+\ell+5)} h_{(k+3)} (v_{f_k})_{(2)} (v_{e_2})_{(2)} \cdots (v_{e_{n-1}})_{(2)} (v_{e_n})_{(2)} S(h_{(2)}) S \partial (h_{(k+4)} \lact (X_P)_{(1)}) h_{(1)} \right) \\
&\phantom{xxx}\stackrel{\text{Y.-D.}}{=} (v_{e_2})_{(1)} \ot\cdots\ot (v_{e_{n-1}})_{(1)} \ot (v_{e_n})_{(1)} S(h_{(4)}) \ot h_{(5)} v_{f_2} \ot\cdots\ot h_{(k+2)} v_{f_{k-1}} \ot h_{(k+3)} (v_{f_k})_{(1)} \\
&\phantom{xxxxxxx} \ot h_{(k+5)} \lact (X_P)_{(2)} \ot h_{(k+6)} \lact X_{Q_2} \ot\cdots\ot h_{(k+\ell+4)} \lact X_{Q_\ell} \\
&\phantom{xxxxxxxxx} \varphi\left(S h_{(k+\ell+5)} h_{(k+4)} (v_{f_k})_{(2)} (v_{e_2})_{(2)} \cdots (v_{e_{n-1}})_{(2)} (v_{e_n})_{(2)} S(h_{(3)}) S \partial (h_{(2)} \lact (X_P)_{(1)}) h_{(1)} \right) \\
&\phantom{xxx}\stackrel{\text{Pf. 1}}{=} (v_{e_2})_{(1)} \ot\cdots\ot (v_{e_{n-1}})_{(1)} \ot (v_{e_n})_{(1)} S(h_{(5)}) \ot h_{(6)} v_{f_2} \ot\cdots\ot h_{(k+3)} v_{f_{k-1}} \ot h_{(k+4)} (v_{f_k})_{(1)} \\
&\phantom{xxxxxxx} \ot h_{(k+6)} \lact (X_P)_{(2)} \ot h_{(k+7)} \lact X_{Q_2} \ot\cdots\ot h_{(k+\ell+5)} \lact X_{Q_\ell} \\
&\phantom{xxxxxxxxx} \varphi\left(S h_{(k+\ell+6)} h_{(k+5)} (v_{f_k})_{(2)} (v_{e_2})_{(2)} \cdots (v_{e_{n-1}})_{(2)} (v_{e_n})_{(2)} S(h_{(4)}) h_{(3)} S\partial(X_P)_{(1)} S h_{(2)} h_{(1)} \right) \\
&\phantom{xxx}\stackrel{\text{}}{=} (v_{e_2})_{(1)} \ot\cdots\ot (v_{e_{n-1}})_{(1)} \ot (v_{e_n})_{(1)} S(h_{(1)}) \ot h_{(2)} v_{f_2} \ot\cdots\ot h_{(k-1)} v_{f_{k-1}} \ot h_{(k)} (v_{f_k})_{(1)} \\
&\phantom{xxxxx} \ot h_{(k+2)} \lact (X_P)_{(2)} \ot h_{(k+3)} \lact X_{Q_2} \ot\cdots\ot h_{(k+\ell+1)} \lact X_{Q_\ell} \\
&\phantom{xxxxxxx} \varphi\left(S h_{(k+\ell+2)} h_{(k+1)} (v_{f_k})_{(2)} (v_{e_2})_{(2)} \cdots (v_{e_{n-1}})_{(2)} (v_{e_n})_{(2)} S\partial(X_P)_{(1)} \right) \\
&\phantom{xxx}\stackrel{\text{Y.-D.}}{=} (v_{e_2})_{(1)} \ot\cdots\ot (v_{e_{n-1}})_{(1)} \ot (v_{e_n})_{(1)} S(h_{(1)}) \ot h_{(2)} v_{f_2} \ot\cdots\ot h_{(k-1)} v_{f_{k-1}} \ot h_{(k)} (v_{f_k})_{(1)} \\
&\phantom{xxxxxxx} \ot h_{(k+1)} \lact (X_P)_{(2)} \ot h_{(k+2)} \lact X_{Q_2} \ot\cdots\ot h_{(k+\ell)} \lact X_{Q_\ell} \\
&\phantom{xxxxxxxxx} \varphi\left(S h_{(k+\ell+2)} h_{(k+\ell+1)} (v_{f_k})_{(2)} (v_{e_2})_{(2)} \cdots (v_{e_{n-1}})_{(2)} (v_{e_n})_{(2)} S\partial(X_P)_{(1)} \right) \\
&\phantom{xxx}\stackrel{\text{}}{=} (v_{e_2})_{(1)} \ot\cdots\ot (v_{e_{n-1}})_{(1)} \ot (v_{e_n})_{(1)} S(h_{(1)}) \ot h_{(2)} v_{f_2} \ot\cdots\ot h_{(k-1)} v_{f_{k-1}} \ot h_{(k)} (v_{f_k})_{(1)} \\
&\phantom{xxxxxxx} \ot h_{(k+1)} \lact (X_P)_{(2)} \ot h_{(k+2)} \lact X_{Q_2} \ot\cdots\ot h_{(k+\ell)} \lact X_{Q_\ell} \\
&\phantom{xxxxxxxxx} \varphi\left( (v_{f_k})_{(2)} (v_{e_2})_{(2)} \cdots (v_{e_{n-1}})_{(2)} (v_{e_n})_{(2)} S\partial(X_P)_{(1)} \right) \\
&\phantom{xxx}= (V_{v,P}^h F_P^\varphi)(v_{e_2} \ot\cdots\ot v_{e_n} \ot v_{f_2} \ot\cdots\ot v_{f_k} \ot X_P \ot X_{Q_2} \ot\cdots\ot X_{Q_\ell}),
\end{align*}
proving the second relation in the statement of the lemma.

In particular, for cocommutative $h \in H$ and $\varphi \in H^*$, the relation just shown implies that the vertex and plaquette operators commute:
\begin{align*}
V_{v,P}^h \circ F_P^\varphi &= F_P^{\varphi\left(S h_{(3)} \cdot ? \cdot h_{(1)} \right)} \circ V_{v,P}^{h_{(2)}} \stackrel{\text{$\varphi$ cocom.}}= F_P^{\varphi\left( ? \cdot h_{(1)} S h_{(3)} \right)} \circ V_{v,P}^{h_{(2)}} \\
&\stackrel{\text{$h$ cocom.}}= F_P^{\varphi\left( ? \cdot h_{(2)} S h_{(1)} \right)}\circ V_{v,P}^{h_{(3)}}\stackrel{\text{}}= F_P^\varphi \circ V_{v,P}^h .
\end{align*}

Finally, let $v$ be different from the base-point of $P$.
In this case, we may apply base-point shifts where necessary, in order to reduce it to the previous case where $v$ is the base-point of $P$, since plaquette operators for cocommutative elements of $H^*$ and vertex operators commute with base-point shifts.
This proves the first relation in the statement of the lemma.
\end{proof}

\bibliography{leverhulme.bib}

\bibliographystyle{unsrt}

\end{document}